%% file: paper.tex
\newif\ifTechRep
\TechReptrue

\documentclass[sigconf]{acmart}
\usepackage[font=small, labelformat=simple]{subcaption}
\usepackage[font=small, labelformat=simple]{caption}
\usepackage[ruled,vlined,noend]{algorithm2e}
\usepackage{enumitem}
\usepackage{color}
\usepackage{multicol}
\usepackage{multirow}
\usepackage{array}
\usepackage{color, colortbl}
\usepackage{booktabs}
\usepackage{balance}
\usepackage{rotating}
\usepackage{pdfrender}
\usepackage{stmaryrd}
\usepackage{anyfontsize}
\usepackage{amsthm}
\usepackage{soul}
\usepackage{xcolor}
\usepackage{soul}

\setlist{nolistsep,leftmargin=*}
\definecolor{vlightgray}{gray}{0.85}
\let\paragraph\relax 

\newtheorem{theorem}{Theorem}
\newtheorem{definition}[theorem]{Definition}
\newtheorem{proposition}[theorem]{Proposition}
\newtheorem{lemma}[theorem]{Lemma}
\newtheorem{example}[theorem]{Example}

\newcommand{\invariant}{constraint\xspace}
\newcommand{\invariants}{constraints\xspace}
\newcommand{\Invariant}{Constraint\xspace}
\newcommand{\Invariants}{Constraints\xspace}
\newcommand{\DI}{Conformance Constraint\xspace}
\newcommand{\DIs}{Conformance Constraints\xspace}
\newcommand{\di}{conformance constraint\xspace}
\newcommand{\dis}{conformance constraints\xspace}
\newcommand{\Di}{Conformance constraint\xspace}
\newcommand{\Dis}{Conformance constraints\xspace}

\newcommand{\paragraph}[1]{\noindent\textbf{#1}}

\newcommand{\ignore}[1]{{}}

\newcommand{\system}{\textsc{CCSynth}\xspace}
\newcommand{\nc}{unsafe\xspace}
\newcommand{\Nc}{Unsafe\xspace}
\newcommand{\extune}{ExTuNe\xspace}
\newcommand{\View}{Projection\xspace}
\newcommand{\Views}{Projections\xspace}
\newcommand{\view}{projection\xspace}
\newcommand{\views}{projections\xspace}
\newcommand{\maxand}{\rhd}
\newcommand{\norm}{\mathtt{simp}}

\newcommand{\citeTechRep}{\cite{technicalReport}}

\newcommand{\appOrTechRep}{\ifTechRep the Appendix\xspace\else our technical report~\citeTechRep\xspace\fi}

\newcommand{\mathcolorbox}[2]{\colorbox{#1}{$\displaystyle #2$}}

\definecolor{green}{RGB}{0,140,0}

\newcommand{\green}{black}
\newcommand{\blue}{black}
\newcommand{\red}{black}

\newcommand{\reviseone}[1]{{\color{\green} #1}}
\newcommand{\revisetwo}[1]{{\color{\blue} #1}}
\newcommand{\revisethree}[1]{{\color{\red} #1}}

\AtBeginDocument{\providecommand\BibTeX{{\normalfont B\kern-0.5em{\scshape i\kern-0.25em b}\kern-0.8em\TeX}}}

\setcopyright{acmlicensed}
\copyrightyear{2021}
\acmYear{2021}
\acmDOI{}
\acmISBN{978-1-4503-6735-6/20/06}
\acmConference[]{Technical Report}{January}{2021}
\acmBooktitle{}
\acmPrice{15.00}

\setcopyright{none}
\renewcommand\footnotetextcopyrightpermission[1]{} 
\begin{document}

\title[Conformance Constraint Discovery: Measuring Trust in Data-Driven Systems]{Conformance Constraint Discovery: \\Measuring Trust in Data-Driven Systems}
\titlenote{An earlier version of this paper had a different title: ``Data Invariants: On Trust in Data-Driven Systems''.}
\ifTechRep
\subtitle{Technical Report}
\fi

\author{Anna Fariha}
\authornote{Work done while the author was an intern at Microsoft.}
\authornote{Both authors contributed equally to this research.}
\affiliation{
  \institution{University of Massachusetts}
  \institution{Amherst, MA, USA}
}
\email{afariha@cs.umass.edu}

\author{Ashish Tiwari}
\authornotemark[3]
\author{Arjun Radhakrishna}
\author{Sumit Gulwani}
\affiliation{
  \institution{Microsoft}
}
\email{{astiwar, arradha, sumitg}@microsoft.com}

\author{Alexandra Meliou}
\affiliation{
  \institution{University of Massachusetts}
    \institution{Amherst, MA, USA}
}
\email{ameli@cs.umass.edu}

\renewcommand{\shortauthors}{Fariha and Tiwari, et al.}
\begin{abstract}
\input{0_abstract}
\end{abstract}

\maketitle

\section{Introduction}\label{sec:introduction}
\input{1_introduction}

\section{Case Studies}\label{sec:casestudies}
\input{2_case-studies}

\section{\DIs}\label{sec:data-invs}
\input{3_data-invariants}

\section{\DI Synthesis}\label{sec:synth-data-inv}
\input{4_synthesizing-data-invariants_1}
\input{4_synthesizing-data-invariants_2}

\section{Trusted Machine Learning}\label{sec:di-for-tml}
\input{5_di_tml}

\section{Experimental Evaluation}\label{sec:experiments}
\input{6_experiment}

\section{Related Work}\label{sec:applications}
\input{7_related-work}

\section{Summary and Future Directions}\label{discussion}
\input{8_summary}

\bibliographystyle{abbrv}
\bibliography{paper}

\ifTechRep
\appendix
\input{9_appendix}

\fi 

\end{document}

%% file: 0_abstract.tex
%
\looseness-1 The reliability of inferences made by data-driven systems hinges
on the data's continued conformance to the systems' initial settings and
assumptions. When serving data (on which we want to apply inference) deviates
from the profile of the initial training data, the outcome of inference becomes
unreliable. 
We introduce \emph{conformance constraints}, a new data profiling primitive
tailored towards quantifying the degree of \emph{non-conformance}, which can
effectively characterize if inference over that tuple is \emph{untrustworthy}.
\Dis are constraints over certain
arithmetic expressions (called \emph{projections}) involving the numerical
attributes of a dataset, which existing data profiling primitives such as
functional dependencies and denial constraints cannot model.

The key finding we present is that projections that incur \emph{low variance}
on a dataset construct effective \dis. This principle yields the surprising
result that low-variance components of a principal component analysis, which
are usually discarded for dimensionality reduction, generate stronger \dis than
the high-variance components. Based on this result, we provide a highly
scalable and efficient technique---linear in data size and cubic in the number
of attributes---for discovering conformance constraints for a dataset. To
measure the degree of a tuple's non-conformance with respect to a dataset, we
propose a \emph{quantitative semantics} that captures how much a tuple violates
the \dis of that dataset.
We demonstrate the value of \dis on two applications: \emph{trusted machine
learning} and \emph{data drift}. We empirically show that \dis offer mechanisms
to (1)~reliably detect tuples on which the inference of a machine-learned model
should not be trusted, and (2)~quantify data drift more accurately than the
state of the art.

%% file: 1_introduction.tex

\looseness-1 Data is central to modern systems in a wide range of domains,
including healthcare, transportation, and finance. The core of modern
data-driven systems typically comprises of models learned from large datasets,
and they are usually optimized to target particular data and workloads. While
these data-driven systems have seen wide adoption and success, their
reliability and proper function hinge on the data's continued conformance to
the systems' initial settings and assumptions. If the serving data (on which
the system operates) deviates from the profile of the initial data (on which
the system was trained), then system performance degrades and system behavior
becomes unreliable. A mechanism to assess the trustworthiness of a system's
inferences is paramount, especially for systems that perform safety-critical or
high-impact operations.

A machine-learned (ML) model typically works best if the serving dataset
follows the profile of the dataset the model was trained on; when it doesn't,
the model's inference can be unreliable. One can profile a dataset in many
ways, such as by modeling the data distribution of the dataset, or by finding
the (implicit) \emph{constraints} that the dataset satisfies.
Distribution-oriented approaches learn data likelihood (e.g., joint or
conditional distribution) from the training data, and can be used to check if
the serving data is unlikely. An unlikely tuple does not necessarily imply that
the model would fail for it. The problem with the distribution-oriented
approaches is that they tend to overfit, and thus, are overly conservative
towards unseen tuples, leading them to report many such false positives.

\looseness-1 We argue that certain constraints offer a more effective and
robust mechanism to quantify trust of a model's inference on a serving tuple.
The reason is that learning systems implicitly exploit such constraints during
model training, and build models that assume that the constraints will continue
to hold for serving data. For example, when there exist high correlations among
attributes in the training data, learning systems will likely reduce the
weights assigned to redundant attributes that can be deduced from others, or
eliminate them altogether through dimensionality reduction. If the serving data
preserves the same correlations, such operations are inconsequential;
otherwise, we may observe model failure.

\looseness-1 In this paper, we characterize datasets with a new data-profiling
primitive, \emph{\dis}, and we present a mechanism to identify \emph{strong}
\dis, whose violation indicates unreliable inference. \Dis specify constraints
over \emph{arithmetic relationships} involving multiple numerical attributes of
a dataset. We argue that a tuple's conformance to the \dis is more critical for
accurate inference than its conformance to the training data distribution. This
is because any violation of \dis is likely to result in a catastrophic failure
of a learned model that is built upon the assumption that the \dis will always
hold. Thus, we can use a tuple's deviation from the \dis as a proxy for the
trust on a learned model's inference for that tuple. We proceed to describe a
real-world example of \dis, drawn from our case-study evaluation on
\emph{trusted machine learning} (TML).

\setlength{\tabcolsep}{.5em}
\renewcommand{\arraystretch}{.9}
\begin{figure}[t]
\centering
    \resizebox{1\columnwidth}{!}
	{\small
	
    \begin{tabular}{lcccc}
    \toprule
     							& \textbf{Departure} 		& \textbf{Departure Time} 		& \textbf{Arrival Time} 		& \textbf{Duration (min)} 	\\
								& \textbf{Date} 	 		& \textbf{[DT]} 				&  \textbf{[AT]} 				& \textbf{[DUR]} 		  	\\
    \midrule
		$t_1$ 					& May 2 					& \texttt{14:30} 				& \texttt{18:20}  				& 230 					  	\\
		$t_2$ 					& July 22 					& \texttt{09:05}				& \texttt{12:15}  				& 195 					  	\\
        \revisetwo{$t_3$} 		& 	\revisetwo{June 6} 		& \revisetwo{\texttt{10:20}}	  & \revisetwo{\texttt{12:20}}  & \revisetwo{115} 			\\ 
		$t_4$ 					& May 19 					& \texttt{11:10} 				& \texttt{13:05}  			   	& 117 						\\
		\rowcolor{vlightgray}
		$t_5$ 					& April 7 					& \texttt{22:30}  			  & \texttt{06:10}  			  	& 458						\\

    \bottomrule
    \end{tabular}
    }
	\vspace{-3mm}
    \caption{\small
	 Sample of the airlines dataset (details are in
	 Section~\ref{exp-invariants-for-ML}), showing departure, arrival, and
	 duration only. The dataset does not report arrival date, but an arrival time
	 earlier than departure time (e.g., last row), indicates an overnight flight.
	 \reviseone{All times are in 24 hour format} and in the same time zone. There is some
	 noise in the values.
	} 
	\vspace{-2mm}
    \label{fig:flights}
\end{figure}

\begin{example}\label{ex:tml}
	\looseness-1

We used a dataset with flight information that includes data on departure and
arrival times, flight duration, etc.\ (Fig.~\ref{fig:flights}) to train a
linear regression model to predict flight delays. \revisetwo{The model was
trained on a subset of the data that happened to include only daytime flights
(such as the first four tuples)}. In an empirical evaluation of the regression
accuracy, we found that the mean absolute error of the regression output more
than quadruples for overnight flights (such as the last tuple $t_5$), compared
to daytime flights. The reason is that tuples representing overnight flights
deviate from the profile of the training data \revisetwo{that only contained
daytime flights}. Specifically, daytime flights satisfy the \di that ``arrival
time is later than departure time and their difference is very close to the
flight duration'', which does not hold for overnight flights. Note that this
\invariant is just based on the covariates (predictors) and does not involve
the target attribute $delay$. Critically, although this \di is unaware of the
regression task, it was still a good proxy of the regressor's performance.
\revisetwo{In contrast, approaches that model data likelihood may report long
daytime flights as unlikely, since all flights in the training data
($t_1$--$t_4$) were also short flights, resulting in false alarms, as the model
works very well for most daytime flights, regardless of the duration (i.e., for
both short and long daytime flights).}
\end{example}

\revisetwo{Example~\ref{ex:tml} demonstrates that when training data has
\emph{coincidental} relationships (e.g., the one between $AT$, $DT$, and $DUR$ for
daytime flights), then ML models may \emph{implicitly} assume them as
\emph{invariants}. \Dis can capture such data invariants and flag
non-conforming tuples (overnight flights) during serving.}\label{extake}

\smallskip

\paragraph{\Dis.} \Dis complement the existing data profiling literature, as
the existing constraint models, such as functional dependencies and denial
constraints, cannot model arithmetic relationships. For example, the \di of
Example~\ref{ex:tml} is: $-\epsilon_1 \le AT - DT - DUR \le \epsilon_2$, where
$\epsilon_1$ and $\epsilon_2$ are small values. \Dis can capture complex linear
dependencies across attributes within a \emph{noisy} dataset. For example, if
the flight departure and arrival data reported the hours and the minutes across
separate attributes, the \invariant would be on a different arithmetic
expression: $(60\cdot arrHour + arrMin) - (60\cdot depHour + depMin) -
duration$.

\looseness-1 The core component of a \di is the arithmetic expression, called
\emph{projection}, which is obtained by a linear combination of the numerical
attributes. There is an unbounded number of projections that we can use to form
arbitrary \dis. For example, for the projection $AT$, we can find a broad range
$[\epsilon_3, \epsilon_4]$, such that all training tuples in
Example~\ref{ex:tml} satisfy the \di $\epsilon_3 \le AT \le \epsilon_4$.
However, this \invariant is too inclusive and a learned model is unlikely to
exploit such a weak constraint. In contrast, the projection $AT - DT - DUR\;$
leads to a stronger \di with a narrow range as its bounds, which is selectively
permissible, and thus, more effective.

\smallskip

\paragraph{Challenges and solution sketch.} The principal challenge is to
discover an \emph{effective} set of conformance constraints that are likely to
affect a model's inference implicitly. We first characterize ``good''
projections (that construct effective constraints) and then propose a method to
discover them. We establish through theoretical analysis two important results:
(1)~A projection is good over a dataset if it is almost constant (i.e., has low
variance) for all tuples in that dataset. (2)~A set of projections,
collectively, is good if the projections have small pair-wise correlations. We
show that low variance components of a principal component analysis (PCA) on a
dataset yield such a set of \views. Note that this is different from---and in
fact completely opposite to---the traditional approaches
(e.g.,~\cite{DBLP:conf/kdd/QahtanAWZ15}) that perform multidimensional analysis
based on the high-variance principal components, after reducing dimensionality
using PCA.

\smallskip

\paragraph{Scope.} \looseness-1 Fig.~\ref{relatedWorkMatrix} summarizes prior
work on related problems, but the scope of our setting differs significantly.
Specifically, we can detect if a serving tuple is non-conforming with respect
to the training dataset \emph{only based on its predictor attributes}, and
require no knowledge of the ground truth. This setting is essential in many
practical applications when we observe \emph{extreme verification
latency}~\cite{souzaSDM:2015}, where ground truths for serving tuples are not
immediately available. For example, consider a self-driving car that is using a
trained controller to generate actions based on readings of velocity, relative
positions of obstacles, and their velocities. In this case, we need to
determine, only based on the sensor readings (predictors), when the driver
should be alerted to take over vehicle control, as we cannot use ground-truths
to generate an alert.

Furthermore, we \emph{do not assume access to the model}, i.e., model's
predictions on a given tuple. This setting is necessary for (1)~safety-critical
applications, where the goal is to quickly alert the user, without waiting for
the availability of the prediction, (2)~auditing and privacy-preserving
applications where the prediction cannot be shared, and (3)~when we are unaware
of the detailed functionality of the system due to privacy concerns or lack of
jurisdiction, but only have some meta-information such as the system trains
some linear model over the training data.

We focus on identifying \emph{tuple-level} non-conformance as opposed to
dataset-level non-conformance that usually requires observing entire data's
distribution. However, our tuple-level approach trivially extends (by
aggregation) to the entire dataset.

\input{table-comparison}

\smallskip

\paragraph{Contrast with prior art.} We now discuss where \dis fit with respect
to the existing literature (Fig.~\ref{relatedWorkMatrix}) on data profiling and
literature on modeling trust in data-driven inferences

\subsubsection*{Data profiling techniques} \Dis fall under the umbrella of data
profiling, which refers to the task of extracting technical metadata about a
given dataset~\cite{DBLP:journals/vldb/AbedjanGN15}. A key task in data
profiling is to learn relationships among attributes. Functional dependencies
(FD)~\cite{papenbrock2015functional} and their variants only capture if a
relationship exists between two sets of attributes, but do not provide a
closed-form (parametric) expression of the relationship. Using the FD $\{AT,
DT\} \rightarrow$ $\{DUR\}$ to model the \invariant of Example~\ref{ex:tml}
suffers from several limitations. First, since the data is noisy, no exact FD
can be learned. Metric FDs~\cite{koudas2009metric} allow small variations in
the data (similar attribute values are considered identical), but hinge on
appropriate distance metrics and thresholds. For example, if $time$ is split
across two attributes ($hour$ and $minute$), the distance metric is
non-trivial: it needs to encode that $\langle hour = 4, min = 59 \rangle$ and
$\langle hour = 5, min = 1\rangle$ are similar, while $\langle hour = 4, min =
1\rangle$ and $\langle hour = 5, min = 59\rangle$ are not. In contrast, \dis
can model the composite attribute ($60 \cdot hour + minute$) by automatically
discovering the coefficients $60$ and $1$ for such a composite attribute.

\looseness-1 Denial constraints (DC)~\cite{DBLP:journals/pvldb/ChuIP13,
DBLP:journals/pvldb/BleifussKN17, pena2019discovery,
DBLP:journals/corr/abs-2005-08540} encapsulate a number of different
data-profiling primitives such as FDs and their variants (e.g,~\cite{
DBLP:conf/icde/FanGLX09}). Exact DCs can adjust to noisy data by adding
predicates until the constraint becomes exact over the entire dataset, but this
can make the constraint extremely large and complex, which might even fail to
provide the desired generalization. For example, a finite DC---whose language
is limited to universally-quantified first-order logic---cannot model the
constraint of Example~\ref{ex:tml}, which involves an arithmetic expression
(addition and multiplication with a constant). Expressing \dis requires a
richer language that includes linear arithmetic expressions. \revisetwo{Pattern
functional dependencies (PFD)~\cite{qahtan2020pattern} move towards addressing
this limitation of DCs, but they focus on text attributes: they are regex-based
and treat digits as characters. However, modeling arithmetic relationships of
numerical attributes requires interpreting digits as numbers.}

\looseness-1 To adjust for noise, FDs and DCs either relax the notion of
constraint violation or allow a user-defined fraction of tuples to violate the
(strict) constraint~\cite{pena2019discovery, huhtala1999tane,
kruse2018efficient, DBLP:conf/sigmod/IlyasMHBA04, koudas2009metric,
caruccio2016discovery, DBLP:journals/corr/abs-2005-08540}. Some
approaches~\cite{DBLP:conf/sigmod/IlyasMHBA04, DBLP:conf/sigmod/ZhangGR20,
DBLP:conf/sigmod/YanSZWC20} use statistical techniques to model other types of
data profiles such as correlations and conditional dependencies. However, they
require additional parameters such as noise and violation thresholds and
distance metrics. In contrast, \dis do not require any parameter from the user
and work on noisy datasets.

\revisetwo{Existing data profiling techniques are not designed to learn what ML
models exploit and they are sensitive to noise in the numerical attributes.
Moreover, data constraint discovery algorithms typically search over an
exponential set of candidates, and hence, are not scalable: their complexity
grows exponentially with the number of attributes or quadratically with data
size. In contrast, our technique for deriving \dis is highly scalable (linear
in data size) and efficient (cubic in the number of attributes). It does not
explicitly explore the candidate space, as PCA---which lies at the core of our
technique---performs the search \emph{implicitly} by iteratively refining
weaker \invariants to stronger ones.} \label{nocandidate}

\subsubsection*{Learning techniques} While \emph{ordinary least square} finds
the lowest-variance projection, it minimizes observational error on only the
target attribute, and thus, does not apply to our setting. \emph{Total least
square} offers a partial solution to our problem as it takes observational
errors on all predictor attributes into account. However, it finds only one
projection---the lowest variance one---that fits the data tuples best. But
there may exist other projections with slightly higher variances and we
consider them all. As we show empirically in
Section~\ref{exp-invariants-for-drift}, constraints derived from multiple
projections, collectively, capture various aspects of the data, and result in
an effective data profile targeted towards certain tasks such as data-drift
quantification. (More discussion is in \appOrTechRep.)

\smallskip 

\paragraph{Contributions.} We make the following contributions:

\begin{itemize}

     \item We ground the motivation of our work with two case studies on trusted
     machine learning (TML) and data drift. (Section~\ref{sec:casestudies})
     
     \item We introduce and formalize \dis, a new data profiling primitive that
     specify constraints over arithmetic relationships among numerical
     attributes of a dataset. We describe a \emph{conformance language} to
     express \dis, and a \emph{quantitative semantics} to quantify how much a
     tuple violates the \dis. In applications of constraint
     violations, some violations may be more or less critical than others. To
     capture that, we consider a notion of \invariant importance, and weigh
     violations against \invariants accordingly. (Section~\ref{sec:data-invs})
      
     \item We formally establish that strong \dis are constructed from
     projections with small variance and small mutual correlation on the given
     dataset. Beyond simple linear \invariants (e.g., the one in
     Example~\ref{ex:tml}), we derive \emph{disjunctive} \invariants, which are
     disjunctions of linear \invariants. We achieve this by dividing the
     dataset into disjoint partitions, and learning linear \invariants for each
     partition. We provide an efficient, scalable, and highly parallelizable
     algorithm for computing a set of linear \dis and disjunctions over them.
     We also analyze its runtime and memory complexity.
     (Section~\ref{sec:synth-data-inv})
     
     \item We formalize the notion of \emph{\nc} tuples in the context of
     trusted machine learning and provide a mechanism to detect \nc tuples
     using \dis. (Section~\ref{sec:di-for-tml})
	 
     \item We empirically analyze the effectiveness of \dis in our two
     case-study applications---TML and data-drift quantification. We show that
     \dis can reliably predict the trustworthiness of linear models and
     quantify data drift precisely, outperforming the state of the art.
     (Section~\ref{sec:experiments})
	 
  \end{itemize}

%% file: table-comparison.tex

\newcolumntype{?}{!{\vrule width 1pt}}
\newcommand{\sln}{2}
\newcommand{\eln}{18} 
\newcommand{\cm}{\checkmark}
\newcommand{\gm}{!}
\let\st\relax
\newcommand{\st}{$\star$}
\newcommand{\na}{$\bot$}
\newcommand{\nr}{$\varoslash$}
\newcommand*{\bc}{
  \textpdfrender{
    TextRenderingMode=FillStroke,
    LineWidth=.5pt,
  }{\checkmark}
}

\begin{figure}[t] 
	\setlength\tabcolsep{1pt} 
	\centering	
	\resizebox{1\columnwidth}{!}{ \small 
	\begin{tabular}
		{?c|c|p{50mm}?c|c|c|c|c|c?p{3mm}|c?c|c|c?c|c|c|c?c|c?} 
		
		\hline
		
		\hline
		
		\multicolumn{3}{?c?}{\cellcolor{vlightgray}{\multirow[c]{1}{*}{\cellcolor{vlightgray}{\rotatebox[origin=lb]{0}{{\textbf{Legend}}}}}}} & 
		
		\multicolumn{6}{c?}{{\multirow[c]{1}{*}{\rotatebox[origin=lb]{0}{{\scriptsize constraints}}} }} &
		\multicolumn{2}{c?}{{\multirow[c]{1}{*}{\rotatebox[origin=lb]{0}{{\scriptsize violation}}}}} &
		\multicolumn{3}{c?}{{\multirow[c]{1}{*}{\rotatebox[origin=lb]{0}{{\scriptsize setting}}}}} &
		\multicolumn{4}{c?}{{\multirow[c]{1}{*}{\rotatebox[origin=lb]{0}{{\scriptsize technique}}}}} &
		\multicolumn{2}{c?}{{\multirow[c]{1}{*}{\rotatebox[origin=lb]{0}{{\scriptsize TML}}}}}\\
		\hline
		
		\hline

		\multicolumn{2}{?r}{\cellcolor{vlightgray}{\small HP:	}} & 			\multicolumn{1}{l?}{\cellcolor{vlightgray}\small Hyper Parameter } &&&&&&&&&&&&&&&&&\\
		\multicolumn{2}{?r}{\cellcolor{vlightgray}{\small FD:	}} & 			\multicolumn{1}{l?}{\cellcolor{vlightgray}\small Functional Dependency} &&&&&&&&&&&&&&&&&\\
		\multicolumn{2}{?r}{\cellcolor{vlightgray}{\small DC:	}} & 			\multicolumn{1}{l?}{\cellcolor{vlightgray}\small Denial Constraint}		&&&&&&&&&&&&&&&&&\\
		\multicolumn{2}{?r}{\cellcolor{vlightgray}{\small \nr:	}} & 			\multicolumn{1}{l?}{\cellcolor{vlightgray}\small Does not require } &&&&&&&&&&&&&&&&&\\
		\multicolumn{2}{?r}{\cellcolor{vlightgray}{\small \na:	}} & 			\multicolumn{1}{l?}{\cellcolor{vlightgray}\small Not applicable } &&&&&&&&&&&&&&&&&\\
		\multicolumn{2}{?r}{\cellcolor{vlightgray}{\small \st:	}} & 			\multicolumn{1}{l?}{\cellcolor{vlightgray}\small Supports via extension} &&&&&&&&&&&&&&&&&\\
		\multicolumn{2}{?r}{\cellcolor{vlightgray}{\small !:	}} & 			\multicolumn{1}{l?}{\cellcolor{vlightgray}\small Partially}& 		 
		
		\multirow[t]{4}{*}{\begin{sideways}{{\small parametric}}\end{sideways}} &
		\multirow[t]{4}{*}{\begin{sideways}{{\small arithmetic}}\end{sideways}} &
		\multirow[t]{4}{*}{\begin{sideways}{{\small approximate}}\end{sideways}} &	
		\multirow[t]{4}{*}{\begin{sideways}{{\small conditional}}\end{sideways}}&
		\multirow[t]{4}{*}{\begin{sideways}{{\small notion of weight}}\end{sideways}} &
		\multirow[t]{4}{*}{\begin{sideways}{{\small interpretable}}\end{sideways}} &

		\multirow[t]{4}{*}{\begin{sideways}{{\small continuous}}\end{sideways}} &
		\multirow[t]{4}{*}{\begin{sideways}{{\small tuple-wise}}\end{sideways}} &

		\multirow[t]{4}{*}{\begin{sideways}{{\small noisy data}}\end{sideways}} &
		\multirow[t]{4}{*}{\begin{sideways}{{\small numerical attr.}}\end{sideways}} &
		\multirow[t]{4}{*}{\begin{sideways}{{\small categorial attr.}}\end{sideways}} &
		
		\multirow[t]{4}{*}{\begin{sideways}{{\small $\!\varoslash\!$ thresholds}}\end{sideways}} &
		\multirow[t]{4}{*}{\begin{sideways}{{\small $\!\varoslash\!$ distance metric}}\end{sideways}} &
		\multirow[t]{4}{*}{\begin{sideways}{{\small $\!\varoslash\!$ HP tuning}}\end{sideways}}&
		\multirow[t]{4}{*}{\begin{sideways}{{\small scalable}}\end{sideways}} &

		\multirow[t]{4}{*}{\begin{sideways}{{\small task agnostic}}\end{sideways}} &
		\multirow[t]{4}{*}{\begin{sideways}{{\small $\!\varoslash\!$ access to model}}\end{sideways}}\\
		
		\hline 
		
		\hline
		
		\multirow{13}{*}{\begin{sideways}{Data Profiling}\end{sideways}} & 
		 
						   \multicolumn{2}{l?}{\textbf{\DIs} } 																&\bc&\bc&\bc&\bc&\bc&\st&	\multicolumn{1}{c|}{\bc} &\bc&	\bc&\bc	&\bc&		\bc&\bc &\bc&\bc&		\bc&\bc\\
		\cline{\sln-20}& \multicolumn{2}{l?}{FD~\cite{papenbrock2015functional}} 											&	&	&	&	&	&\cm&		&	&	   &	&\cm&	\cm&\cm &\cm&   &		\multicolumn{2}{c?}{\multirow{12}{*}{\begin{sideways}{not addressed in prior work}\end{sideways}}}\\
		\cline{\sln-\eln}& \multicolumn{2}{l?}{Approximate FD~\cite{kruse2018efficient}} 									&	&	&\cm&	&	&\cm&		&	&	\cm&	&\cm&	   &\cm &\cm&	&	  \multicolumn{2}{c?}{}\\
		\cline{\sln-\eln}& \multicolumn{2}{l?}{Metric FD~\cite{koudas2009metric}} 											&	&	&\cm&	&	&\cm&		&	&	\cm&\cm	&\cm&	   &	&\na&\na&	  \multicolumn{2}{c?}{}\\
		\cline{\sln-\eln}& \multicolumn{2}{l?}{Conditional FD~\cite{DBLP:conf/icde/FanGLX09}} 								& ! &	&	&\cm&	&\cm&		&\cm&	\cm&	&\cm&	   &\cm &\cm&   &	  \multicolumn{2}{c?}{}\\
		\cline{\sln-\eln}& \multicolumn{2}{l?}{Pattern FD~\cite{qahtan2020pattern}} 										& ! &	&	&	&	&\cm&		&\cm&	\cm&	&\cm&	   &\cm &\cm&   &	  \multicolumn{2}{c?}{}\\		
		\cline{\sln-\eln}& \multicolumn{2}{l?}{Soft FD~\cite{DBLP:conf/sigmod/IlyasMHBA04}} 								&	&	&\cm&	&\cm&\cm&	\cm &	&	\cm&\cm	&\cm&	   &\cm &\cm&\cm&	  \multicolumn{2}{c?}{}\\
		\cline{\sln-\eln}& \multicolumn{2}{l?}{Relaxed FD~\cite{caruccio2016discovery}} 									&	&	&\cm&	&	&\cm&		&	&	\cm&\cm	&\cm&	   &	&\cm&	&	  \multicolumn{2}{c?}{}\\
		\cline{\sln-\eln}& \multicolumn{2}{l?}{FDX~\cite{DBLP:conf/sigmod/ZhangGR20}} 										&	&	&	&	&	&\cm&		&	&	\cm&	&\cm&	   &	&	&\cm&	  \multicolumn{2}{c?}{}\\
		\cline{\sln-\eln}& \multicolumn{2}{l?}{Differential Dependency~\cite{song2011differential}} 						&	&	&	&	&	&\cm&		&	&	   &\cm	&\cm&	   &	&\cm&	&	  \multicolumn{2}{c?}{}\\
		\cline{\sln-\eln}& \multicolumn{2}{l?}{DC~\cite{DBLP:journals/pvldb/ChuIP13, DBLP:journals/pvldb/BleifussKN17}} 	& ! &	&\cm&\cm&\cm&\cm&		&\cm&	\cm&\cm	&\cm&	   &\cm &\cm&	&	  \multicolumn{2}{c?}{}\\
		\cline{\sln-\eln}& \multicolumn{2}{l?}{Approximate DC~\cite{pena2019discovery, DBLP:journals/corr/abs-2005-08540}} 	& ! &	&\cm&\cm&\cm&\cm&		&\cm&	\cm&\cm	&\cm&	   &\cm &\cm&	&	  \multicolumn{2}{c?}{}\\
		\cline{\sln-\eln}& \multicolumn{2}{l?}{Statistical Constraint~\cite{DBLP:conf/sigmod/YanSZWC20}} 					&	&	&\cm&	&	&\cm&	\cm	&	&	\cm&\cm &\cm&	   &\cm &\cm&\cm&	  \multicolumn{2}{c?}{}\\
			
		\hline
						
		\hline
		
		\multirow{7}{*}{\begin{sideways}{Learning}\end{sideways}} & 
						   \multicolumn{2}{l?}{Ordinary Least Square} 														&\cm&\cm&\cm&	&	&\st		&\cm&\cm		&\cm&\cm&\st		&\cm&\cm&\cm&\cm		&	&	\\
		\cline{\sln-20}& \multicolumn{2}{l?}{Total Least Square} 															&\cm&\cm&\cm&	&	&\st		&\cm&\cm		&\cm&\cm&\st		&\cm&\cm&\cm&\cm		&\cm&\cm\\
		\cline{\sln-20}& \multicolumn{2}{l?}{Auto-encoder~\cite{DBLP:journals/corr/abs-1812-02765}} 						&\multicolumn{6}{c?}{\na}		&\cm&\cm		&\cm&\cm&\cm		&	&	&	&\cm		&\cm&\cm\\
		\cline{\sln-20}& \multicolumn{2}{l?}{Schelter et al.~\cite{DBLP:conf/sigmod/SchelterRB20}\textsuperscript{+}} 	&\multicolumn{6}{c?}{\na}		&\cm&			&	&\cm&\cm		&	&\cm&\cm&\cm		&	&	\\
		\cline{\sln-20}& \multicolumn{2}{l?}{Jiang et al.~\cite{DBLP:conf/nips/JiangKGG18}} 								&\multicolumn{6}{c?}{\na}		&\cm&\cm		&\cm&\cm&\cm		&\cm&	&	&\cm		&	&	\\
		\cline{\sln-20}& \multicolumn{2}{l?}{Hendrycks et al.~\cite{DBLP:journals/corr/HendrycksG16c}} 					&\multicolumn{6}{c?}{\na}		&	&\cm		&\cm&\cm&\cm		&\cm&\cm&\cm&\cm		&	&	\\
		\cline{\sln-20}& \multicolumn{2}{l?}{Model's Prediction Probability} 												&\multicolumn{6}{c?}{\na}		&\cm&\cm		&\multicolumn{7}{c?}{varies}				&   &	\\
		\hline
		
		\hline
		\multicolumn{20}{r}{\textsuperscript{+} {\scriptsize Requires additional information}}\\

	\end{tabular}
	} 	
	\vspace{-4mm} 	
	\caption{\Dis complement existing data profiling primitives and provide an 
	efficient mechanism to quantify trust in prediction, with
	minimal assumption on the setting. 
    } 
	\label{relatedWorkMatrix} 
	\vspace{-4mm} 
	
\end{figure}

%% file: 2_case-studies.tex

\looseness-1 Like other data-profiling primitives, \dis have general
applicability in understanding and describing datasets. But their true power
lies in quantifying the degree of a tuple's non-conformance with respect to a
reference dataset. Within the scope of this paper, we focus on two case studies
in particular to motivate our work: trusted machine learning and data drift. We
provide an extensive evaluation over these applications in
Section~\ref{sec:experiments}.

\smallskip

\looseness-1 \paragraph{Trusted machine learning (TML)} refers to the problem
of quantifying trust in the inference made by a machine-learned model on a new
serving tuple~~\cite{DBLP:conf/hicons/TiwariDJCLRSS14,
DBLP:journals/crossroads/Varshney19, DBLP:journals/corr/abs-1904-07204,
DBLP:conf/kdd/Ribeiro0G16, DBLP:conf/nips/JiangKGG18}.
This is particularly useful in case of extreme verification
latency~\cite{souzaSDM:2015}, where ground-truth outputs for new serving tuples
are not immediately available to evaluate the performance of a learned model,
when auditing models for trustworthiness, and in privacy-preserving
applications where even the model's predictions cannot be shared.
When a model is trained using a dataset, the \dis for that dataset specify a
safety envelope~\cite{DBLP:conf/hicons/TiwariDJCLRSS14} that characterizes the
tuples for which the model is expected to make trustworthy predictions. If a
serving tuple falls outside the safety envelope (violates the \dis), then the
model is likely to produce an untrustworthy inference. Intuitively, the higher
the violation, the lower the trust. Some classifiers produce a confidence
measure along with the class prediction, typically by applying a softmax
function to the raw numeric prediction values. However, such a confidence
measure is not well-calibrated~\cite{DBLP:conf/nips/JiangKGG18,
guo2017calibration}, and therefore, cannot be reliably used as a measure of
trust in the prediction. Additionally, a similar mechanism is not available for
inferences made by regression models.

\looseness-1 In the context of TML, we formalize the notion of \emph{\nc
tuples}, on which the prediction may be untrustworthy. We establish that \dis
provide a sound and complete procedure for detecting \nc tuples, which
indicates that the search for \dis should be guided by the class of models
considered by the corresponding learning system (Section~\ref{sec:di-for-tml}). 

\smallskip

\paragraph{Data drift}~\cite{DBLP:conf/kdd/QahtanAWZ15,
DBLP:journals/tnn/KunchevaF14, DBLP:journals/csur/GamaZBPB14,
DBLP:journals/jss/BarddalGEP17} specifies a significant change in a dataset
with respect to a reference dataset, which typically requires that systems be
updated and models retrained. Aggregating tuple-level non-conformances over a
dataset gives us a \emph{dataset-level} non-conformance, which is an effective
measurement of data drift. To quantify how much a dataset $D'$ drifted from a
reference dataset $D$, our three-step approach is: (1)~compute \dis for $D$,
(2)~evaluate the \invariants on all tuples in $D'$ and compute their violations
(degrees of non-conformance), and (3)~finally, aggregate the tuple-level
violations to get a dataset-level violation. If all tuples in $D'$ satisfy the
\invariants, then we have no evidence of drift. Otherwise, the aggregated
violation serves as the drift quantity.

\smallskip

While we focus on these two applications here, we mention other applications of
\dis in \appOrTechRep.

%% file: 3_data-invariants.tex

\def\Inv{\mathtt{Inv}}
\def\Dom{\mathtt{Dom}}
\def\coDom{\mathtt{coDom}}
\def\DDom{\mathbf{Dom}}
\def\CC{\mathcal{C}}
\def\Real{{\mathbb{R}}}

\def\Bool{{\mathtt{Bool}}}
\def\true{{\mathtt{True}}}
\def\false{{\mathtt{False}}}
\def\vecw{{\vec{w}}}
\def\vecv{{\vec{v}}}
\def\lb{{\mathtt{lb}}}
\def\ub{{\mathtt{ub}}}
\def\atomic{{atomic}}
\def\Atomic{{Atomic}}
\def\M{{\infty}}
\newcommand{\semq}[1]{{[\![{#1}]\!]}}
\newcommand\avg[1]{{\mathtt{\mu}}(#1)}
\newcommand\sign[1]{{\mathtt{sign}({#1})}}
\newcommand\delF[2]{{\Delta{#1}({#2})}}
\newcommand\stddev[1]{{\mathtt{\sigma}}(#1)}
\def\month{M}

In this section, we define \dis that allow us to capture complex arithmetic
dependencies involving numerical attributes of a dataset. Then we propose a
language for representing them. Finally, we define quantitative semantics over
\dis, which allows us to quantify their violation.

\smallskip

\looseness-1
\paragraph{Basic notations.}
We use $\mathcal{R} (A_1, A_2, \dots, A_m)$ to denote a relation schema where
$A_i$ denotes the $i^{th}$ attribute of $\mathcal{R}$. We use $\Dom_i$ to
denote the domain of attribute $A_i$. Then the set ${\DDom}^m=\Dom_1\times
\cdots\times \Dom_m$ specifies the domain of all possible tuples. We use $t
\in\DDom^m$ to denote a tuple in the schema $\mathcal{R}$. A dataset
$D\subseteq {\DDom}^m$ is a specific instance of the schema $\mathcal{R}$. For
ease of notation, we assume some order of tuples in $D$ and we use $t_i \in D$
to refer to the $i^{th}$ tuple and $t_i.A_j \in \Dom_j$ to denote the value of
the $j^{th}$ attribute of $t_i$.

\smallskip

\looseness-1 \paragraph{\Di.} A \di $\Phi$ characterizes a set of allowable or
conforming tuples and is expressed through a \emph{conformance language}
(Section~\ref{sec:cl}). We write $\Phi(t)$ and $\neg\Phi(t)$ to denote that $t$
satisfies and violates $\Phi$, respectively.

\begin{definition}[\Di]\label{def:di2}
     A \di for a dataset $D \,{\subseteq}\, {\DDom}^m$ is a formula \linebreak
     $\Phi: {\DDom}^m \mapsto \{\true,\false\}$ such that $|\{t \in D \mid
     \neg\Phi(t)\}| \ll |D|$.
\end{definition}
The set $\{t \in D \mid \neg\Phi(t)\}$ denotes atypical tuples in $D$ that do
not satisfy the \di $\Phi$. In our work, we do not need to know the set of
atypical tuples, nor do we need to purge the atypical tuples from the dataset.
Our techniques derive \invariants in ways that ensure there are very few
atypical tuples (Section~\ref{sec:synth-data-inv}). 

\subsection{Conformance Language}\label{sec:cl}
\looseness-1 \paragraph{\View.} A central concept in our conformance language
is \linebreak \emph{\view}. Intuitively, a \view is a derived attribute that
specifies a ``lens'' through which we look at the tuples. More formally, a
\view is a function $F: \DDom^m \mapsto\Real$ that maps a tuple $t\in\DDom^m$
to a real number $F(t) \in \Real$. In our language for \dis, we only consider
\views that correspond to linear combinations of the numerical attributes of a
dataset. Specifically, to define a \view, we need a set of numerical
coefficients for all attributes of the dataset and the \view is defined as a
sum over the attributes, weighted by their corresponding coefficients. We extend
a projection $F$ to a dataset $D$ by defining $F(D)$ to be the sequence of
reals obtained by applying $F$ on each tuple in $D$ individually.

\smallskip

{
\paragraph{Grammar.} Our language for \dis consists of formulas
$\Phi$ generated by the following grammar:
{\small
\begin{eqnarray*}
    \phi & := & \lb \leq F(\vec{A}) \leq \ub \; \mid\; \wedge(\phi,\ldots,\phi)
    \\[-0.3em]
    \psi_A & := & \vee(\;(A = c_1) \maxand \phi,\; (A = c_2) \maxand \phi, \;\ldots)
    \\[-0.3em]
    \Psi & := & \psi_A \; \mid \; \wedge(\psi_{A_1}, \psi_{A_2}, \ldots)
	\\[-0.3em]
    \Phi & := & \phi  \;\mid\; \Psi
\end{eqnarray*}
}
\looseness-1 
\noindent 
The language consists of (1)~bounded constraints $\lb \leq F(\vec{A})\leq \ub$
where $F$ is a \view on $\DDom^m$, $\vec{A}$ is the tuple of formal parameters
$(A_1, A_2, \ldots, A_m)$, and $\lb, \ub \in \Real$ are reals; (2)~equality
constraints $A = c$ where $A$ is an attribute and $c$ is a constant in
$A$'s domain; and (3)~operators ($\maxand$, $\wedge$, and $\vee$,) that connect
the constraints.
}
Intuitively, $\maxand$ is a switch operator that specifies which \invariant $\phi$
applies based on the value of the attribute $A$, $\wedge$ denotes conjunction,
and $\vee$ denotes disjunction.
Formulas generated by $\phi$ and $\Psi$ are called {\em{simple \invariants}}
and {\em{compound \invariants}}, respectively. Note that a formula generated by
$\psi_A$ only allows equality constraints on a single attribute, namely $A$,
among all the disjuncts.

\begin{example}\label{ex:constraints}
    Consider the dataset $D$ consisting of the first four tuples  $\{t_1, t_2, t_3, t_4\}$ of Fig.~\ref{fig:flights}.
    A simple \invariant for $D$ is:
	{
	$$
    \phi_1: -5\leq AT - DT - DUR \leq 5.
	$$
	}
    Here, the projection $F(\vec{A}) {=} AT {-} DT {-} DUR$, with attribute coefficients
    $\langle 1, - 1, -1 \rangle$, $\lb = {-}5$, and $\ub = 5$. A compound
    \invariant is:
	{
	\begin{align*}
     \psi_2:   \month &= \text{``May''}  \maxand -2 \leq \mathcolorbox{lightgray}{AT - DT - DUR} \leq 0 \\[-0.4em]
     \vee \;\; \month &= \text{``June''} \maxand \phantom{-}0  \leq \mathcolorbox{lightgray}{AT - DT - DUR} \leq 5\\[-0.4em]
     \vee \;\; \month &= \text{``July''} \, \maxand {-}5 \leq \mathcolorbox{lightgray}{AT - DT - DUR} \leq 0
	\end{align*}
	}
     For ease of exposition, we assume that all times are converted to minutes
     (e.g., \texttt{06:10} $= 6 {\times} 60 {+} 10 = 370$) and $\month$ denotes
     the departure month, extracted from $Departure$ $Date$.
\end{example}

Note that arithmetic expressions that specify linear combination of numerical
attributes (highlighted above) are disallowed in denial
constraints~\cite{DBLP:journals/pvldb/ChuIP13} which only allow raw attributes
and constants (more details are in \appOrTechRep).

\subsection{Quantitative Semantics}\label{quant-sem}

\looseness-1 \Dis have a natural Boolean semantics: a tuple either satisfies a
\invariant or it does not. However, Boolean semantics is of limited use in
practice, because it does not quantify the degree of \invariant violation. We
interpret \dis using a quantitative semantics, which quantifies violations, and
reacts to noise more gracefully than Boolean semantics.

The quantitative semantics $\semq{\Phi}(t)$ is a measure of the violation of
$\Phi$ on a tuple $t$---with a value of $0$ indicating no violation and a value
greater than 0 indicating some violation. In Boolean semantics, if $\Phi(t)$ is
$\true$, then $\semq{\Phi}(t)$ will be $0$; and if $\Phi(t)$ is $\false$, then
$\semq{\Phi}(t)$ will be $1$. Formally, $\semq{\Phi}$ is a mapping from
$\DDom^m$ to $[0,1]$.

\smallskip

\noindent\emph{Quantitative semantics of simple \invariants.} 
\revisethree{We build upon $\epsilon$-insen\-sitive loss~\cite{vapnik1997support}
to define the quantitative semantics of simple \invariants, where the bounds
$\lb$ and $\ub$ define the $\epsilon$-insensitive
zone:}\footnote{\revisethree{For a target value $y$, predicted value $\hat{y}$,
and a parameter $\epsilon$, the $\epsilon$-insensitive loss is $0$ if $|y -
\hat{y}| < \epsilon$ and $|y - \hat{y}| - \epsilon$ otherwise.
}}\label{fn1}
\begin{align*}
    \semq{\lb\leq F(\vec{A})\leq \ub}(t) &:= \eta(\alpha \cdot \max(0, F(t)-\ub, \lb - F(t)))
    \\
    \semq{\wedge(\phi_1,\ldots,\phi_K)}(t) &:= \textstyle\sum_{k}^{K} \gamma_k \cdot \semq{\phi_k}(t)
\end{align*}
Below, we describe the parameters of the quantitative semantics, and provide further details on them in \appOrTechRep.

\smallskip 
\paragraph{Scaling factor $\alpha\in\Real^+$.}\\ \Views are unconstrained
functions and different \views can map the same tuple to vastly different
values. We use a scaling factor $\alpha$ to standardize the values computed by
a \view $F$, and to bring the values of different \views to the same comparable
scale. The scaling factor is automatically computed as the inverse of the
standard deviation: $\tfrac{1}{\stddev{F(D)}}$. 
We set $\alpha$ to a large positive number when $\stddev{F(D)}$ = $0$.

\smallskip
\paragraph{Normalization function $\eta(.): \Real\mapsto [0,1]$.}\\
The normalization function
maps values in the range $[0,\infty)$ to the range $[0,1)$.
While any monotone mapping from $\Real^{\geq 0}$ to $[0,1)$ can be used, we
pick $\eta(z) = 1 - e^{-z}$.

\smallskip 
\looseness-1
\paragraph{Importance factors $\gamma_k\in \Real^+$, $\textstyle\sum_{k}^{K}\gamma_k {=} 1$.}\\
The weights $\gamma_k$ control the contribution of each bounded-\view \invariant in a
conjunctive formula. This allows for prioritizing \invariants that are more
significant than others within the context of a particular application.
In our work, we derive the importance factor of a \invariant automatically,
based on its \view's standard deviation over $D$.

\subsubsection*{Quantitative semantics of compound \invariants} 
Compound \invariants are first simplified into simple \invariants, and they get
their meaning from the simplified form.
We define a function
$\norm(\psi, t)$ that takes a compound \invariant $\psi$ and a tuple $t$ and returns a simple
\invariant. It is defined recursively as follows:
\begin{align*}
    &\norm(\vee(\;(A = c_1) \maxand\phi_1, (A = c_2) \maxand \phi_2,\ldots), t) := 
      \phi_k  \; \mbox{if $t.A = c_k$}
    \\
    &\norm(\wedge(\psi_{A_1}, \psi_{A_2},\ldots), t)  := \wedge(\norm(\psi_{A_1}, t), \norm(\psi_{A_2},t),\ldots)
\end{align*}
\begin{sloppypar} If the condition in the definition above does not hold for any $c_k$, then
$\norm(\psi,t)$ is undefined and $\norm(\wedge(\dots,\psi,\dots),t)$ is
also undefined. If $\norm(\psi,t)$ is undefined, then $\semq{\psi}(t) :=
1$. When $\norm(\psi,t)$ is defined, the quantitative semantics of $\psi$ is
given by:
$$
\semq{\psi}(t) := \semq{\norm(\psi,t)}(t)
$$
\end{sloppypar}

Since compound \invariants simplify to simple \invariants, we mostly focus on 
simple \invariants. Even there, we pay special attention to bounded-projection 
\invariants ($\phi$) of the form \linebreak $\lb \le F(\vec{A}) \le \ub$, which lie at the core of simple \invariants. 

\begin{example}\label{ex:semantics}
    Consider the \invariant $\phi_1$ from Example~\ref{ex:constraints}.
	For $t \in D$,
    $\semq{\phi_1}(t)= 0$ since $\phi_1$ is satisfied by all tuples in $D$. The
    standard deviation of the projection $F$ over $D$, $\sigma(F(D)) {=}
    \sigma(\{0,-5, 5, -2\}) {=} 3.6$. Now consider the last tuple $t_5 \not\in
    D$. $F(t_5) = (370 - 1350)-458 = -1438$, which is way below the lower bound $-5$ of $\phi_1$. 
	Now we compute how much $t_5$ violates $\phi_1$: 
	$\semq{\phi_1}(t_5) = \semq{-5 \leq F(\vec{A}) \leq 5}(t_5) 
		 			   = \eta(\alpha \cdot \max(0, -1438 - 5, -5 + 1438)) = 1 - e^{-\frac{1433}{3.6}} \approx 1$.
	Intuitively, this implies that $t_5$ strongly violates $\phi_1$.
\end{example}

%% file: 4_synthesizing-data-invariants_1.tex

\looseness-1 In this section, we describe our techniques for deriving \dis. We
start with the synthesis of simple \invariants (the $\phi$ \invariants in
our language specification), followed by compound \invariants (the $\Psi$
\invariants in our language specification). Finally, we analyze the time and
memory complexity of our algorithm.

\subsection{Simple \DIs}\label{sec:conjunctive} Synthesizing simple \dis
involves (a)~discovering the \views, and (b)~discovering the lower and upper
bounds for each \view. We start by discussing~(b), followed by the principle to
identify effective \views, based on which we solve~(a).

\subsubsection{Synthesizing Bounds for \Views}\label{synth-bounds} Fix a \view $F$
and consider the bounded-\view \invariant $\phi$: $\lb \leq
F(\vec{A}) \leq \ub$. 
Given a dataset $D$, a trivial choice for the bounds
that are valid on all tuples in $D$ 
is: $\lb = \min(F(D))$ and $\ub = \max(F(D))$. However, this choice is very
sensitive to noise: adding a single atypical tuple to $D$ can produce very
different \invariants.
Instead, we use a more robust choice as follows:
\begin{align*}
    \lb = \avg{F(D)} - C \cdot \stddev{F(D)}, \;\;\;\;
    \ub = \avg{F(D)} + C \cdot \stddev{F(D)}
\end{align*}
\looseness-1
Here, $\avg{F(D)}$ and $\stddev{F(D)}$ denote the mean and standard deviation
of the values in $F(D)$, respectively, and $C$ is some positive constant.
With these bounds, $\semq{\phi}(t) = 0$ implies that 
        $F(t)$ is within $C \times \stddev{F(D)}$ from the mean $\avg{F(D)}$. 
In our experiments, we set $C = 4$, which ensures that in expectation, very few
tuples in $D$ will violate the \invariant for many distributions of the values
in $F(D)$. Specifically, if $F(D)$ follows a normal distribution, then
$99.99\%$ of the population is expected to lie within $4$ standard deviations
from mean. \reviseone{Note that we make no assumption on the original data
distribution of each attribute.}

Setting the bounds $\lb$ and $\ub$ as $C\cdot\stddev{F(D)}$-away from the mean,
and the scaling factor $\alpha$ as $\frac{1}{\stddev{F(D)}}$, guarantees the
following property for our quantitative semantics:
\begin{lemma}\label{lemma:helper}
    Let $D$ be a dataset and let 
    $\phi_k$ be $\lb_k \leq F_k(\vec{A})\leq \ub_k$ for $k=1,2$.
    Then, for any tuple $t$, if 
    $\frac{|F_1(t) - \avg{F_1(D)}|}{\stddev{F_1(D)}} \geq \frac{|F_2(t) - \avg{F_2(D)}|}{\stddev{F_2(D)}}$,
    then $\semq{\phi_1}(t) \geq \semq{\phi_2}(t)$.
\end{lemma}
This means that larger deviation from the mean (proportionally to
the standard deviation) results in higher degree of violation under our
semantics. The proof follows from the fact that the normalization function
$\eta(.)$ is monotonically increasing, and hence, $\semq{\phi_k}(t)$ is a
monotonically non-decreasing function of $\frac{|F_k(t) -
\avg{F_k(D)}|}{\stddev{F_k(D)}}$.

\subsubsection{Principle for Synthesizing \Views}\label{view-synthesis}

\looseness-1 
We start by investigating what makes a \invariant more effective than others.
An effective \invariant (1)~should not overfit the data, but rather generalize
by capturing the properties of the data, and (2)~should not underfit the data,
because it would be too permissive and fail to identify deviations effectively.
Our flexible bounds (Section~\ref{synth-bounds}) serve to avoid overfitting. In
this section, we focus on identifying the principles that help us avoid
underfitting. We first describe the key technical ideas for characterizing
effective \views through example and then proceed to formalization.
 
\begin{example}\label{ex:badproj}
	 \looseness-1 Let $D$ be a dataset of three tuples
	 \{(1,1.1),(2,1.7),(3,3.2)\} with two attributes $X$ and $Y$. Consider two
	 arbitrary \views: $X$ and $Y$. For $X$: $\mu(X(D)) = 2$ and $\sigma(X(D)) =
	 0.8$. So, bounds for its \di are: $\lb = 2 - 4 \times 0.8 = -1.2$ and $\ub =
	 2 + 4 \times 0.8 = 5.2$. This gives us the \di: $-1.2 \le X \le 5.2$.
	 Similarly, for $Y$, we get the \di: $-1.6 \le Y \le 5.6$.
	 Fig.~\ref{subfig:views1} shows the conformance zone (clear region) defined by
	 these two \dis. The shaded region depicts non-conformance zone. The
	 conformance zone is large and too permissive: it allows many atypical tuples
	 with respect to $D$, such as $(0,4)$ and $(4, 0)$.
\end{example}

A natural question arises: are there other \views that can better characterize
conformance with respect to the tuples in $D$? The answer is yes and next we
show another pair of \views that shrink the conformance zone significantly.

\begin{example}\label{ex:goodproj} 
	 In Fig.~\ref{subfig:views3}, the clear region is defined by the \dis $-0.8
	 \le X - Y \le 0.8$ and $-2.8 \le X + Y \le 10.8$, over \views $X-Y$ and
	 $X+Y$, respectively. The region is indeed much smaller than the one in
	 Fig.~\ref{subfig:views1} and allows fewer atypical tuples.
\end{example}

How can we derive \view $X-Y$ from the \views $X$ and $Y$, given $D$? Note that
$X$ and $Y$ are highly correlated in $D$. In Lemma~\ref{LEMMA:MAIN}, we show
that two highly-correlated \views can be linearly combined to construct another
\view with lower standard deviation that generates a \emph{stronger} \invariant.
We proceed to formalize stronger \invariant---which defines whether a
\invariant is more effective than another in quantifying violation---and
\emph{incongruous} tuples---which help us estimate the subset of the data
domain for which a \invariant is stronger than the others.

\begin{definition}[Stronger \invariant]\label{def:stronger}
	\begin{sloppypar}
A \di $\phi_1$ is stronger
than another \di $\phi_2$ on a subset $H \subseteq \DDom^m$ if
$\forall t \in H, \; \semq{\phi_1}(t) \geq \semq{\phi_2}(t)$.
	\end{sloppypar}
\end{definition}

Given a dataset $D\subseteq \DDom^m$ and a \view $F$, for any tuple $t$, let
$\delF{F}{t} = F(t) - \avg{F(D)}$. For \views $F_1$ and $F_2$, the correlation
coefficient $\rho_{F_1,F_2}$ (over $D$) is defined as
$\frac{\frac{1}{|D|}\sum_{t\in D}
\delF{F_1}{t}\delF{F_2}{t}}{\stddev{F_1(D)}\stddev{F_2(D)}}$.
\begin{definition}[Incongruous tuple]	
A tuple $t$ is {\em{incongruous}} w.r.t.\ a
\view pair $\!\langle F_1, F_2\rangle\!$ on $D$ if: 
$\delF{F_1}{t}\cdot\delF{F_2}{t}\cdot\rho_{F_1,F_2}<0$.
\end{definition}

\begin{figure}
	\centering
	\resizebox{0.8\columnwidth}{!}{
	\begin{subfigure}{.18\textwidth}
	  \includegraphics[width=1\linewidth]{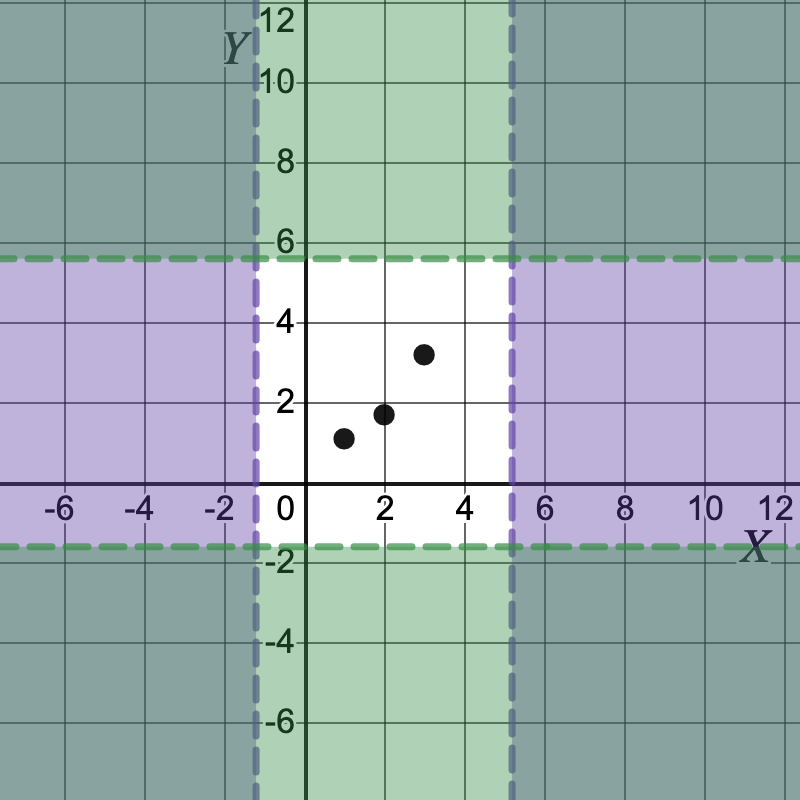}
	  \vspace{-6mm}
	  \caption{}
	  \label{subfig:views1}
	\end{subfigure}%
	\hspace{5mm}
	\begin{subfigure}{.18\textwidth}
	  \includegraphics[width=1\linewidth]{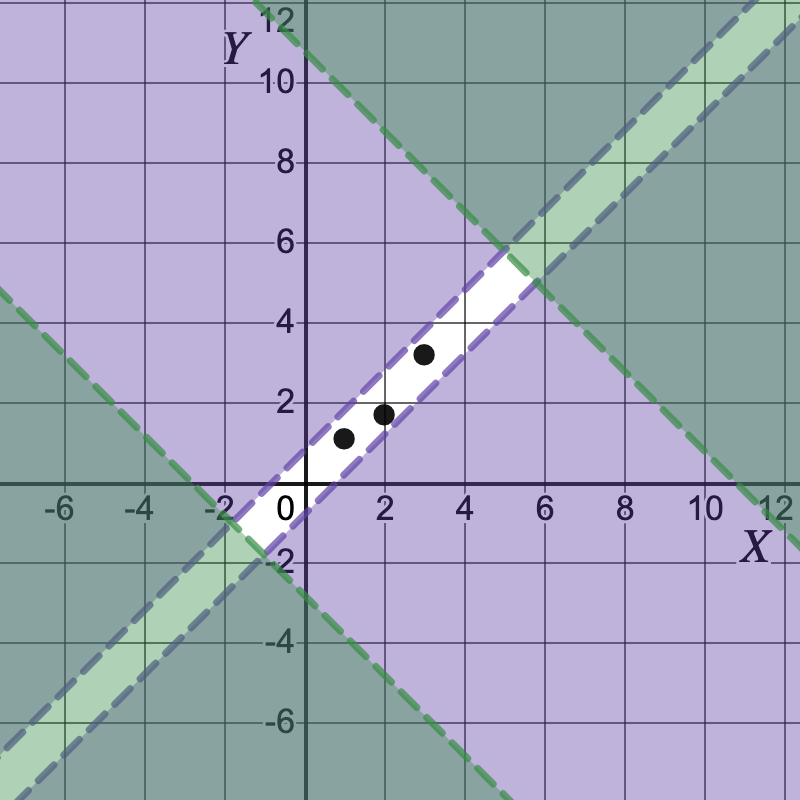}
	  \vspace{-6mm}
	  \caption{}
	  \label{subfig:views3}
	\end{subfigure}
	}
	\vspace{-4mm}
        \caption{\small
		\looseness-1
		Clear and shaded regions depict conformance and non-conformance zones, respectively.
		(a)~Correlated projections $X$ and $Y$ yield \dis forming a large conformance zone, 
		(b)~Uncorrelated (orthogonal) projections $X-Y$ and $X+Y$ yield \dis forming a smaller conformance zone.
		}
	\label{fig:views}
	\vspace{-3mm}
\end{figure}

Informally, an incongruous tuple for a pair of \views does not follow the
general trend of correlation between the \view pair. For example, if $F_1$ and
$F_2$ are positively correlated ($\rho_{F_1,F_2}>0$), an incongruous tuple $t$
deviates in opposite ways from the mean of each \view
($\delF{F_1}{t}\cdot\delF{F_2}{t}<0$). Our goal is to find \views that yield a
conformance zone with very few incongruous tuples.

\begin{example} 
	In Example~\ref{ex:badproj}, $X$ and $Y$ are positively
correlated with $\rho_{X,Y} \approx 1$. The tuple
$t=(0,4)$ is incongruous w.r.t. $\langle X, Y \rangle$, because $X(t) = 0 <
\mu(X(D)) = 2$, whereas $Y(t) = 4 > \mu(Y(D)) = 2$. Intuitively, the
incongruous tuples do not behave like the tuples in $D$ when viewed through the
\views $X$ and $Y$. Note that the narrow conformance zone of
Fig.~\ref{subfig:views3} no longer contains the incongruous tuple $(0, 4)$.
In fact, the conformance zone defined by the \dis derived from \views $X - Y$
and $X + Y$ are free from a vast majority of the incongruous tuples.
\end{example}

We proceed to state Lemma~\ref{LEMMA:MAIN}, which informally says that: any two
highly-correlated \views can be linearly combined to construct a new \view to
obtain a stronger \invariant. We write $\phi_F$ to denote the \di $\lb \le
F(\vec{A}) \le \ub$, synthesized from $F$. (All proofs are in \appOrTechRep.)

\begin{lemma}\label{LEMMA:MAIN}
    Let $D$ be a dataset and $F_1,F_2$ be two \views on $D$ s.t.\ $|\rho_{F_1,F_2}| \geq \frac{1}{2}$.
    Then, $\exists \beta_1, \beta_2 \in \Real$ s.t.\ $\beta_1^2 + \beta_2^2 = 1$ and for the new \view
    $F = \beta_1 F_1 + \beta_2 F_2$:
	 \begin{enumerate}[label=(\arabic*)]
	     \item $\stddev{F(D)} < \stddev{F_1(D)}$ and $\stddev{F(D)} < \stddev{F_2(D)}$, and
         \item $\phi_F$ is stronger than both $\phi_{F_1}$ and $\phi_{F_2}$
     on the set of tuples that are incongruous w.r.t. $\!\langle F_1,
     F_2\rangle$.
	 \end{enumerate}
\end{lemma}

\noindent We now extend the result to multiple \views in Theorem~\ref{THM:MAIN}.

\begin{theorem}[Low Standard Deviation \Invariants]\label{THM:MAIN}
    Given a dataset $D$, let $\mathcal{F} {=} \{F_1,\ldots,F_K\}$ denote a set of \views on $D$
    s.t.\ $\exists F_i,F_j{ \in }\mathcal{F}$ with $|\rho_{F_i,F_j}| {\geq} \frac{1}{2}$.
    Then, there exist a nonempty subset $I {\subseteq} \{1,\ldots,K\}$ and a \view
    $F {=} \sum_{k\in I} \beta_k F_k$, where $\beta_k {\in} \Real$ s.t.
    \begin{enumerate}[label=(\arabic*)]
        \item $\forall k\in I$, $\stddev{F(D)} < \stddev{F_k(D)}$,
        \item $\forall k\in I$, $\phi_F$ is stronger than $\phi_{F_k}$ on the subset $H$, where\\ 
	$H {=} \{t \mid 
		\forall {k{\in} I} (\beta_k \delF{F_k}{t} {\geq} 0) \vee 
		\forall {k{\in} I} (\beta_k \delF{F_k}{t} {\leq} 0)\}$, and
        \item $\forall k\not\in I$, $|\rho_{F,F_k}| < \frac{1}{2}$.
    \end{enumerate}
\end{theorem}

\looseness-1 The theorem establishes that to detect violations for tuples
in~$H$: (1)~\views with low standard deviations define stronger \invariants
(and are thus preferable), and (2)~a set of \invariants with highly-correlated
\views is suboptimal (as they can be linearly combined to generate stronger
\invariants).
Note that $H$ is a conservative estimate for the set of tuples where $\phi_F$
is stronger than each $\phi_{F_k}$; there exist tuples outside $H$ for which
$\phi_F$ is stronger.

\revisethree{ 
\smallskip

\paragraph{Bounded \views vs.\ convex polytope.} \label{polytope} 
Bounded projections (Example~\ref{ex:goodproj}) relate to the
computation of convex polytopes~\cite{turchini2017convex}. A convex polytope
is the smallest possible convex set of tuples that includes all the training
tuples; then any tuple falling outside the polytope would be considered
non-conforming. The problem with using convex polytopes is that they
overfit to the training tuples and is extremely sensitive to outliers. For
example, consider a training dataset over attributes $X$ and $Y$: $\{(1, 10),
(2,20), (3, 30)\}$. A convex polytope in this case would be a line
segment---starting at $(1, 10)$ and ending at $(3, 30)$---and the tuple $(5,
50)$ will fall outside it.

Unlike convex polytope---whose goal is to find the smallest possible
``inclusion zone'' that includes all training tuples---our goal is to find a
``conformance zone'' that reflects the \emph{trend} of the training tuples.
This is inspired from the fact that ML models aim to generalize to tuples
outside training set; thus, \dis also need to capture trends and avoid
overfitting. Our definition of good \dis (low variance and low mutual
correlation) balances overfitting and overgeneralization. Therefore, beyond the
minimal bounding hyper-box (i.e., a convex polytope) over the training tuples,
we take into consideration the distribution (variance and concentration) of the
\emph{interaction} among attributes (trends). For the above example, \dis will
model the interaction trend: $Y = 10X$ allowing the tuple $(5, 50)$, which
follows the same trend as the training tuples.
}

%% file: 4_synthesizing-data-invariants_2.tex

\subsubsection{PCA-inspired \View Derivation}\label{candidatesec}

\setlength{\textfloatsep}{0mm}
\begin{algorithm}[t]
	\caption{Procedure to generate linear \views.}
	\label{fig:pca}
	\LinesNumbered
	\small{ 
    	\SetKwInOut{Inputs}{Inputs}
		\SetKwInOut{Output}{Output}
		\Inputs{
			A dataset $D \subset \DDom^m$
		}
		\Output{
                    A set $\{(F_1,\gamma_1),\ldots,(F_K,\gamma_K)\}$ of \views and importance factors
    	}
                $D_N  \gets D {\mbox{ after dropping non-numerical attributes}}$  \label{algo2:line1}\\ 
                $D'_N  \gets [\vec{1} ; D_N]$ \label{algo2:line2} \\
                $\{\vec{w}_1,\ldots,\vec{w}_K\} \gets {\mbox{ eigenvectors of }} {D_N'}^T D_N'$ \label{algo2:line31} \\
				\ForEach{$1 \le k \le K$}{
					$\vec{w}_k' \gets {\mbox{ $\vec{w}_k$ with first element removed}}$\label{algo2:line32} \\
                                        $F_k \gets \lambda\vec{A}: \frac{\vec{A}^T\vec{w}_k'}{||\vec{w}_k'||}$\label{algo2:line33} \\
					$\gamma_k \gets \frac{1}{\log(2+\stddev{F_k(D_N)})}$ \label{algo2:line34}
				}
                \Return $\{(F_1,\frac{\gamma_1}{Z}),\ldots,(F_K,\frac{\gamma_K}{Z})\}$, where
                $Z = \sum_k \gamma_k$\label{algo2:line35}
	}
\end{algorithm}
\setlength{\textfloatsep}{5pt}
 
Theorem~\ref{THM:MAIN} sets the requirements for good \views (see
also~\cite{tveten2019principal, tveten2019tailored,
DBLP:journals/tnn/KunchevaF14} that make similar observations in different
ways). It indicates that we can start with any arbitrary \views and then
iteratively improve them. However, we can get the desired set of best \views in
one shot using an algorithm inspired by principal component analysis (PCA).
\revisetwo{PCA relies on computing eigenvectors. \label{pcadetails} There exist
different algorithms for computing eigenvectors (from the infinite space of
possible vectors). The general mechanism involves applying numerical approaches
to iteratively converge to the eigenvectors (up to a desired precision) as no
analytical solution exists in general.} Our algorithm returns \views that
correspond to the principal components of a
slightly modified version of the given dataset.
Algorithm~\ref{fig:pca} details our approach for discovering \views for
constructing \dis:

\smallskip

\begin{description}
    \looseness-1
    \item[Line~\ref{algo2:line1}] Drop all non-numerical attributes from $D$ to
    get the numeric dataset $D_N$. This is necessary because PCA only applies
    to numerical values. Instead of dropping, one can also consider embedding
    techniques to convert non-numerical attributes to numerical ones.

    \item[Line~\ref{algo2:line2}] Add a new column to $D_N$ that consists of
    the constant $1$, to obtain the modified dataset $D_N' := [\vec{1} ; D_N]$,
    where $\vec{1}$ denotes the column vector with $1$ everywhere. We do this
    transformation to capture the additive constant within principal
    components, which ensures that the approach works even for unnormalized
    data.

    \item[Line~\ref{algo2:line31}] Compute $K$ eigenvectors of the square
    matrix ${D_N'}^{T}D_N'$, where $K$ denotes the number of columns in $D'_N$.
	These eigenvectors provide coefficients to construct \views.

    \item[Lines~\ref{algo2:line32}--\ref{algo2:line33}] Remove the first
    element (coefficient for the newly added constant column) of all
    eigenvectors and normalize them to generate \views. Note that we no longer
    need the constant element of the eigenvectors since we can appropriately
    adjust the bounds, $\lb$ and $\ub$, for each \view by evaluating it on
    $D_N$.

    \item[Line~\ref{algo2:line34}] Compute importance factor for each \view.
    Since \views with smaller standard deviations are more discerning
    (stronger), as discussed in Section~\ref{quant-sem}, we assign each \view
    an importance factor ($\gamma$) that is inversely proportional to its
    standard deviation over $D_N$.

    \item[Line~\ref{algo2:line35}] Return the linear \views with corresponding
    normalized importance factors.

\end{description}

\smallskip

We now claim that the \views returned by Algorithm~\ref{fig:pca} include the
\view with minimum standard deviation and 
the correlation between any two \views 
is 0. This indicates that we cannot further improve
the \views, and thus they are optimal.

\begin{theorem}[Correctness of Algorithm~\ref{fig:pca}]\label{THM:ALGO2-CORRECTNESS}
    Given a numerical dataset $D$ over the schema $\mathcal{R}$, let
    $\mathcal{F} = \{F_1,F_2,\ldots, F_K\}$ be the set of linear \views
    returned by Algorithm~\ref{fig:pca}. Let $\sigma^* =
    \min_k^{K}\stddev{F_k(D)}$. If $\mu(A_k(D)) = 0$ for all attribute $A_k$ in
    $\mathcal{R}$, then,\footnote{When the condition $\forall A_k \;
    \mu(A_k(D)) = 0$ does not hold, slightly modified variants of the claim
    hold. However, by normalizing $D$ (i.e., by subtracting attribute mean
    $\mu(A_k(D))$ from each $A_k(D)$), it is always possible to satisfy the
    condition.}
	\begin{enumerate}[label=(\arabic*)]
            \item $\sigma^* \leq \stddev{F(D)}$ $\forall F=\vec{A}^T\vec{w}$ where $||\vec{w}||\geq 1$, and
    \item $\forall F_j, F_k \in \mathcal{F}$ s.t.\ $F_j\neq F_k$, 
        $\rho_{F_j, F_k} = 0$.
	\end{enumerate}
\end{theorem}

Using \views $F_1, \ldots, F_K$, and importance factors
$\gamma_1,\ldots,\gamma_K$, returned by Algorithm~\ref{fig:pca}, we generate
the simple (conjunctive) \invariant with $K$ conjuncts:
$
 \bigwedge_k  \lb_k \leq F_k(\vec{A}) \leq \ub_k
$.
We compute the bounds $\lb_k$ and $\ub_k$ following Section~\ref{synth-bounds}
and use the importance factor $\gamma_k$ for the $k^{th}$ conjunct in the
quantitative semantics.

\begin{example}\label{ex:int} Algorithm~\ref{fig:pca} finds the projection of the \di of
Example~\ref{ex:tml}, but in a different form. The actual airlines dataset has
an attribute $distance\; (DIS)$ that represents miles travelled by a flight. In
our experiments, we found the following \di\footnote{For ease of exposition, we
use $F(\vec{A}) \approx 0$ to denote $\epsilon_1 \le F(\vec{A}) \le
\epsilon_2$, where $\epsilon_i \approx 0$.} over the dataset of daytime flights:
{\small
\begin{align}\label{eq:one}
0.7 \times AT - 0.7 \times DT - 0.14 \times DUR - 0.07 \times DIS \approx 0
\end{align}
}
This \invariant is not quite interpretable by itself, but it is in fact a
linear combination of two expected and interpretable
\invariants:~\footnote{\revisetwo{Interpretability is not our explicit goal, but
we developed a tool~\cite{DBLP:conf/sigmod/FarihaTRG20} to explain causes of
non-conformance. More discussion and case studies are in \appOrTechRep.}}
{\small
\begin{align}
	AT - DT - DUR &\approx 0 \label{eq:two}\\
	DUR - 0.12 \times DIS &\approx 0 \label{eq:three}
\end{align}
}
Here, (\ref{eq:two}) is the one mentioned in Example~\ref{ex:tml} and
(\ref{eq:three}) follows from the fact that average aircraft speed is about
$500$ mph implying that it requires $0.12$ minutes per mile. $0.7$ $\times$
(\ref{eq:two}) + $0.56$ $\times$ (\ref{eq:three}) yields:
{\small
\begin{align*}
&0.7 \times (AT - DT - DUR) + 0.56 \times DUR - 0.56 \times 0.12 \times DIS \approx 0 \\
\implies & 0.7 \times AT - 0.7 \times DT - 0.14 \times DUR - 0.07 \times DIS \approx 0
\end{align*}
}
Which is exactly the \di~(\ref{eq:one}). Algorithm~\ref{fig:pca} found the
optimal \view of~(\ref{eq:one}), which is a linear combination of the \views
of~(\ref{eq:two}) and~(\ref{eq:three}). The reason is: there is a correlation
between the \views of~(\ref{eq:two}) and~(\ref{eq:three}) over the dataset
(Theorem~\ref{THM:MAIN}). One possible explanation of this correlation is:
whenever there is an error in the reported duration of a tuple, it violates
both~(\ref{eq:two}) and~(\ref{eq:three}). Due to this natural correlation,
Algorithm~\ref{fig:pca} returned the optimal \view of~(\ref{eq:one}), that
``covers'' both \views of~(\ref{eq:two}) or~(\ref{eq:three}).

\end{example}

\subsection{Compound \DIs}\label{sec:disjunctive}
The quality of our PCA-based simple linear \invariants relies on how many low
variance linear \views we are able to find on the given dataset. For many
datasets, it is possible we find very few, or even none, such linear \views. In
these cases, it is fruitful to search for compound \invariants; we first focus
on {\em{disjunctive \invariants}} (defined by $\psi_A$ in our language grammar).

PCA-based approach fails in cases where there exist different piecewise linear
trends within the data, as it will result into low-quality \invariants, with
very high variances. In such cases, partitioning the dataset and then learning
\invariants separately on each partition will result in significant improvement
of the learned \invariants. A disjunctive \invariant is a compound \invariant
of the form $\bigvee_k((A = c_k) \maxand \phi_k)$, where each $\phi_k$ is a
\invariant for a specific partition of $D$. Finding disjunctive \invariants
involves horizontally partitioning the dataset $D$ into smaller disjoint
datasets $D_1, D_2, \ldots, D_L$. Our strategy for partitioning $D$ is to use
categorical attributes with a small domain in $D$; in our implementation, we
use those attributes $A_j$ for which $|\{t.A_j | t\in D\}|\le50$. If $A_j$ is
such an attribute with values $v_1, v_2, \ldots, v_L$, we partition $D$ into
$L$ disjoint datasets $D_1, D_2,\ldots, D_L$, where $D_l = \{t \in D | t.A_j =
v_l\}$. Let $\phi_1, \phi_2, \ldots, \phi_L$ be the $L$ simple \dis we learn
for $D_1, D_2, \ldots, D_L$ using Algorithm~\ref{fig:pca}, respectively. We
compute the following disjunctive \di for $D$:\\
\indent $
 ((A_j = v_1) \maxand \phi_1) \vee 
 ((A_j = v_2) \maxand \phi_2) \vee 
\cdots \vee
 ((A_j = v_L) \maxand \phi_L)\\[-0.7em]
$

\looseness-1 We repeat this process and partition $D$ across multiple
attributes and generate a compound disjunctive \invariant for each attribute.
Then we generate the final compound conjunctive \di ($\Psi$) for $D$, which is
the conjunction of all these disjunctive \invariants. Intuitively, this final
\di forms a set of \emph{overlapping} hyper-boxes around the data tuples.

\subsection{Theoretical Analysis}\label{sec:complexity} 

\subsubsection{Runtime Complexity}\looseness-1 Computing simple \invariants
involves two computational steps: (1)~computing $X^TX$, where $X$ is an
$n\times m$ matrix with $n$ tuples (rows) and $m$ attributes (columns), which
takes $\mathcal{O}(nm^2)$ time, and (2)~computing the eigenvalues and
eigenvectors of an $m\times m$ positive definite matrix, which has complexity
$\mathcal{O}(m^3)$~\cite{DBLP:conf/stoc/PanC99}. Once we obtain the linear
\views using the above two steps, we need to compute the mean and variance of
these \views on the original dataset, which takes $\mathcal{O}(nm^2)$ time. In
summary, the overall procedure is cubic in the number of attributes and linear
in the number of tuples.
For computing disjunctive \invariants, we greedily pick attributes that take at
most $L$ (typically small) distinct values, and then run the above procedure
for simple \invariants at most $L$ times. This adds just a constant factor
overhead per attribute.

\subsubsection{Memory Complexity} The procedure can be implemented in
$\mathcal{O}(m^2)$ space. The key observation is that $X^T X$ can be computed
as $\sum_{i=1}^{n} t_i t_i^T$, where $t_i$ is the $i^{th}$ tuple in the
dataset. Thus, $X^TX$ can be computed incrementally by loading only one tuple
at a time into memory, computing $t_i t_i^T$, and then adding that to a running
sum, which can be stored in $\mathcal{O}(m^2)$ space. Note that instead of such
an incremental computation, this can also be done in an embarrassingly parallel
way where we horizontally partition the data (row-wise) and each partition is
computed in parallel.

\revisetwo{\subsubsection{Implication, Redundancy, and Minimality}\label{sec:impl}
Definition~\ref{def:stronger} gives us the notion of \emph{implication} on
\dis: for a dataset $D$, satisfying $\phi_1$ that is stronger than $\phi_2$
implies that $D$ would satisfy $\phi_2$ as well.
Lemma~\ref{LEMMA:MAIN} and Theorem~\ref{THM:MAIN} associate \emph{redundancy}
with correlation: correlated projections can be combined to construct a new
projection that makes the correlated projections redundant.
Theorem~\ref{THM:ALGO2-CORRECTNESS} shows that our PCA-based procedure finds
a non-redundant (orthogonal and uncorrelated) set of projections. For disjunctive
\invariants, it is possible to observe redundancy across partitions. However,
our quantitative semantics ensures that redundancy does not affect the
violation score.
Another notion relevant to data profiles (e.g., FDs) is \emph{minimality}. In
this work, we do not focus on finding the minimal set of \dis. Towards
achieving minimality for \dis, a future direction is to explore techniques for
optimal data partitioning. However, our approach computes only $m$ \dis for
each partition. Further, for a single tuple, only $m_N \cdot m_C$ \dis are applicable,
where $m_{N}$ and $m_C$ are the number of numerical and categorical attributes
in $D$ (i.e., $m = m_{N} + m_C$). The quantity $m_N \cdot m_C$ is upper-bounded by
$\frac{m^2}{4}$. }

%% file: 5_di_tml.tex

\newcommand{\arrTime}{AT\xspace}
\newcommand{\depTime}{DT\xspace}
\newcommand{\duration}{DUR\xspace}

\looseness-1 In this section, we provide a theoretical justification of why
\dis are effective in identifying tuples for which learned models are likely to
make incorrect predictions. To that end, we define \emph{\nc} tuples and show
that an ``ideal'' \di provides a sound and complete mechanism to detect \nc
tuples. In Section~\ref{sec:synth-data-inv}, we showed that low-variance \views
construct strong \dis, which yield a small conformance zone. 
We now make a similar argument, but in a slightly different way: we show that
\views with zero variance give us equality constraints that are useful for
trusted machine learning. We start with an example to provide the intuition.

\begin{example}\label{ex:one} 
	 Consider the airlines dataset $D$ and assume that all tuples in $D$
	 satisfy the equality constraint $\phi := \arrTime - \depTime - \duration = 0$
	 (i.e., $\lb = \ub = 0$). Note that for equality constraint, the corresponding
	 \view\ has {\em{zero}} variance---the lowest possible variance. Now, suppose
	 that the task is to learn some function $f(\arrTime, \depTime, \duration)$.
	 If the above constraint holds for $D$, then the ML model can instead learn
	 the function $g(\arrTime, \depTime, \duration) = f(\depTime + \duration,
	 \depTime, \duration)$. $g$ will perform just as well as $f$ on $D$: in fact,
	 it will produce the same output as $f$ on $D$. If a new serving tuple $t$
	 satisfies $\phi$, then $g(t) = f(t)$, and the prediction will be correct.
	 However, if $t$ does not satisfy $\phi$, then $g(t)$ will likely be
	 significantly different from $f(t)$. Hence, violation of the \di is a strong
	 indicator of performance degradation of the learned prediction model. Note
	 that $f$ need not be a linear function: as long as $g$ is also in the class
	 of models that the learning procedure is searching over, the above argument
	 holds.
\end{example}

Based on the intuition of Example~\ref{ex:one}, we proceed to formally define
\nc tuples. We use $[D;Y]$ to denote the {\em{annotated dataset}} obtained by
appending the target attribute $Y$ to a dataset $D$, and $\coDom$ to denote
$Y$'s domain.

\begin{definition}[\Nc\ tuple]\label{def:trustworthy}
    Given a class $\CC$ of functions with signature $\DDom^m \mapsto \coDom$,
    and an annotated dataset $[D;Y] \subset (\DDom^m\times \coDom)$,
    a tuple $t \in \DDom^m$ is \nc w.r.t. $\CC$ and $[D;Y]$, if
    $\exists f, g \in \CC$ s.t. $f(D) = g(D) = Y$ but $f(t) \neq g(t)$. 
\end{definition}

Intuitively, $t$ is \nc if there exist two different predictor functions $f$
and $g$ that agree on all tuples in $D$, but disagree on $t$. Since, we can
never be sure whether the model learned $f$ or $g$, we should be cautious about
the prediction on $t$. Example~\ref{ex:one} suggests that $t$ can be \nc when
all tuples in $D$ satisfy the equality \di $f(\vec{A}) - g(\vec{A}) = 0$ but
$t$ does not. Hence, we can use the following approach for trusted machine
learning:
\begin{enumerate}
	\addtolength{\itemindent}{1cm}
	\item  Learn \dis $\Phi$ for the dataset $D$.
	\item  Declare $t$ as \nc if $t$ does not satisfy $\Phi$.
\end{enumerate}
\smallskip

The above approach is sound and complete for characterizing \nc tuples, thanks
to the following proposition.

\begin{proposition}\label{PROP:EXISTENCE}
    There exists a \di $\Phi$ for $D$ s.t. the following statement
    is true: ``$\neg\Phi(t)$ iff $t$ is \nc\ w.r.t. $\CC$ and $[D ; Y]$ for all $t \in \DDom^m$''.
\end{proposition}

The required \di $\Phi$ is: $\forall{f,g\in\CC}:f(D)=g(D)=Y \Rightarrow
f(\vec{A}) - g(\vec{A}) = 0$.
Intuitively, when all possible pairs of functions that agree on $D$ also agree
on $t$, only then the prediction on $t$ can be trusted. (More discussion is in
\appOrTechRep.)

\subsection{Applicability}\label{applicability}

\revisetwo{\paragraph{Generalization to noisy setting.} While our analysis and
formalization for using \dis for TML focused on the noise-free setting, the
intuition generalizes to noisy data. Specifically, suppose that $f$ and $g$ are
two possible functions a model may learn over $D$; then, we expect that the
difference $f-g$ will have small variance over $D$, and thus would be a good
\di. In turn, the violation of this constraint would mean that $f$ and $g$
diverge on a tuple $t$ (making $t$ unsafe); since we are oblivious of the
function the model learned, prediction on $t$ is untrustworthy.}

\smallskip

\reviseone{\paragraph{False positives.} \Dis are designed to work in a
model-agnostic setting. Although this setting is of great practical importance,
designing a perfect mechanism for quantifying trust in ML model predictions,
while remaining completely model-agnostic, is challenging. It raises the
concern of \emph{false positives}: \dis may incorrectly flag tuples for which
the model's prediction is in fact correct. This may happen when the model
ignores the trend that \dis learn. Since we are oblivious of the prediction
task and the model, it is preferable that \dis behave rather
\emph{conservatively} so that the users can be cautious about potentially \nc
tuples. Moreover, if a model ignores some attributes (or their interactions)
during training, it is still necessary to learn \dis over them. Particularly,
in case of concept drift~\cite{tsymbal2004problem}, the ground truth may start
depending on those attributes, and by learning \dis over all attributes, we can
better detect potential model failures.

\smallskip

\looseness-1 \paragraph{False negatives.} Another concern involving \dis is of
\emph{false negatives}: linear \dis may miss nonlinear constraints, and thus
fail to identify some unsafe tuples. However, the linear dependencies modeled
in \dis persist even after sophisticated (nonlinear) attribute transformations.
Therefore, violation of \dis is a strong indicator of potential failure of a
possibly nonlinear model.

\smallskip

\paragraph{Modeling nonlinear constraints.} \looseness-1 While linear \dis are
the most common ones, we note that our framework can be easily extended to
support nonlinear \dis using \emph{kernel
functions}~\cite{scholkopf2002learning}---which offer an efficient, scalable,
and powerful mechanism to learn nonlinear decision boundaries for support
vector machines (also known as ``kernel trick''). Briefly, instead of
explicitly augmenting the dataset with transformed nonlinear attributes---which
grows exponentially with the desired degree of polynomials---kernel functions
enable \emph{implicit} search for nonlinear models. The same idea also applies
for PCA called kernel-PCA~\cite{alzate2008kernel, DBLP:conf/nips/JiangKGG18}.
While we limit our evaluation to only linear kernel, polynomial kernels---e.g.,
radial basis function (RBF)~\cite{keerthi2003asymptotic}---can be plugged into
our framework to model nonlinear \dis.

In general, our conformance language is not guaranteed to model all possible
functions that an ML model can potentially learn, and thus is not guaranteed to
find the best \di. However, our empirical evaluation on real-world datasets
shows that our language models \dis effectively.}

%% file: 6_experiment.tex
\newcommand{\ut}{untrustworthy\xspace}
\newcommand{\Ut}{Untrustworthy\xspace}

We now present experimental evaluation to demonstrate the effectiveness of \dis over our two case-study applications
(Section~\ref{sec:casestudies}): trusted machine learning and data drift.
Our experiments target the following research questions:     

\begin{itemize}
    \item How effective are \dis for trusted machine learning?
    Is there a relationship between \invariant violation score and the ML model's
    prediction accuracy? (Section~\ref{exp-invariants-for-ML})

    \item Can \dis be used to quantify data drift? How do they
    compare to other state-of-the-art drift-detection techniques?
    (Section~\ref{exp-invariants-for-drift})

\end{itemize}

\smallskip \noindent\textbf{Efficiency.} In all our experiments, our algorithms
for deriving \dis were extremely fast, and took only a few seconds even for
datasets with 6 million rows. The number of attributes were reasonably small
($\sim$40), which is true for most practical applications. As our theoretical
analysis showed (Section~\ref{sec:complexity}), our approach is linear in
the number of data rows and cubic in the number of attributes. Since the runtime
performance of our techniques is straightforward, we opted to not include
further discussion of efficiency here and instead focus this empirical analysis
on the techniques' effectiveness.

\smallskip
\noindent 
\textbf{Implementation: \system.} 
We created an open-source implementation of \dis and our method for
synthesizing them, \system, in Python 3. Experiments were run on a Windows 10
machine (3.60 GHz processor and 16GB RAM).

\smallskip
\noindent
{\large \textit{Datasets}}
\smallskip

\noindent\textbf{Airlines}~\cite{airlineSource} contains data about
flights and has 14 attributes
\reviseone{---year, month, day, day of week, departure time, arrival time,
carrier, flight number, elapsed time, origin, destination, distance, diverted,
and arrival delay.}
We used a subset of the data containing all flight information for year 2008.
\reviseone{In this dataset, most of the attributes follow uniform distribution
(e.g., month, day, arrival and departure time, etc.); elapsed time and distance
follow skewed distribution with higher concentration towards small values
(implying that shorter flights are more common); arrival delay follows a
slightly skewed gaussian distribution implying most flights are on-time, few
arrive late and even fewer arrive early.}
The training and serving datasets contain 5.4M and 0.4M rows, respectively.

\smallskip

\noindent \looseness-1 \textbf{Human Activity Recognition
(HAR)}~\cite{sztyler2016onbody} is a real-world dataset about activities for 15
individuals, 8 males and 7 females, with varying fitness levels and BMIs. We
use data from two sensors---accelerometer and gyroscope---attached to 6 body
locations---head, shin, thigh, upper arm, waist, and chest. We consider 5
activities---lying, running, sitting, standing, and walking. The dataset
contains 36 numerical attributes (2 sensors $\times$ 6 body-locations $\times$
3 co-ordinates) and 2 categorical attributes---activity-type and person-ID. We
pre-processed the dataset to aggregate the measurements over a small time window,
resulting in 10,000 tuples per person and activity, for a total of 750,000
tuples.

\smallskip

\noindent\textbf{Extreme Verification Latency (EVL)}~\cite{souzaSDM:2015}
is a widely used benchmark to evaluate drift-detection algorithms in
non-stationary environments under extreme verification latency. It contains 16
synthetic datasets with incremental and gradual concept drifts over time. The
number of attributes of these datasets vary from 2 to 6 and each of them has
one categorical attribute.

\subsection{Trusted Machine Learning}\label{exp-invariants-for-ML}

\looseness-1 We now demonstrate the applicability of \dis in the TML problem.
We show that, serving tuples that violate the training data's \dis are \nc, and
therefore, an ML model is more likely to perform poorly on those tuples.

\smallskip

\begin{figure}[t!]
	\centering
	\setlength{\tabcolsep}{5pt}
	\renewcommand\arraystretch{0.88}
	\small{
	\begin{tabular}{ccccc}
		\hline
		&  \multirow{ 2}{*}{Train} & \multicolumn{3}{c}{Serving}\\
		\cline{3-5}
		&& Daytime & Overnight & Mixed \\
		\midrule
		\textbf{Average violation} & 0.02\% & 0.02\% & 27.68\% & 8.87\%\\
		\textbf{MAE} & 18.95	 &  18.89 & 80.54 & 38.60\\
		\bottomrule
		
	\end{tabular}
	}
		\vspace{-3mm}	
	 \caption{Average \invariant violation (in percentage) and MAE (for linear regression) 
	 of four data splits on the airlines dataset. The \invariants were learned on
 	 \texttt{Train}, excluding the target attribute, \texttt{delay}.}
	 
	\label{fig:airlines-summary}
	\vspace{2mm}
	\centering
	\includegraphics[width=0.8\linewidth]{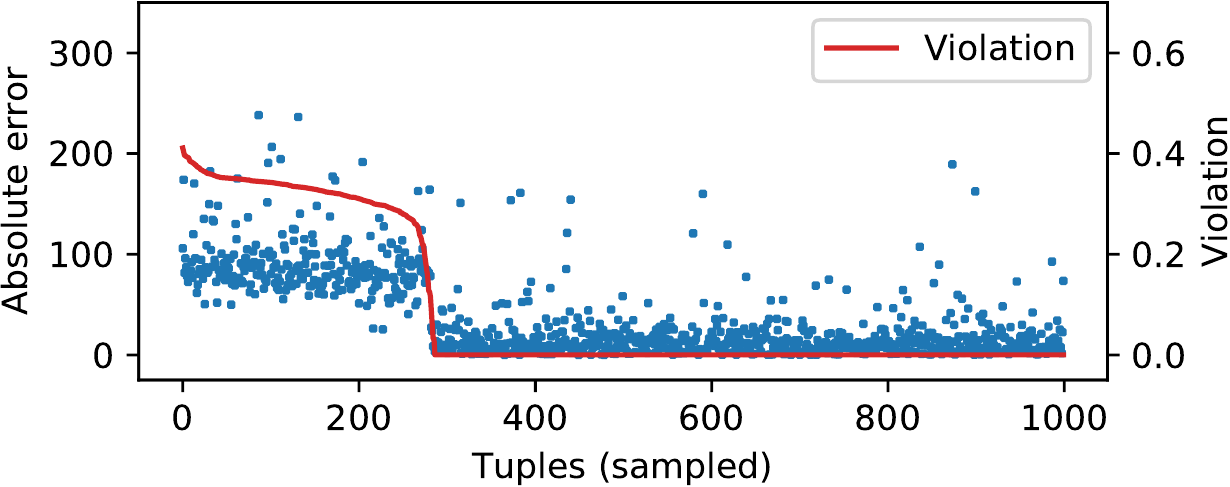}
	\vspace{-3mm}	
	\caption{\Invariant violation strongly correlates with the absolute error
	 of delay prediction of a linear regression model.}
	 \vspace{2mm}	
	\label{fig:airlines}
\end{figure}

\noindent \textbf{Airlines.} We design a regression task of predicting the
arrival delay and train a linear regression model for the task. \reviseone{Our
goal is to observe whether the mean absolute error of the predictions
(positively) correlates to the \invariant violation for the serving tuples.} In
a process analogous to the one described in Example~\ref{ex:tml}, our training
dataset (\texttt{Train}) comprises of a subset of daytime flights---flights
that have arrival time later than the departure time (in 24 hour format). We
design three serving sets: (1)~\texttt{Daytime}: similar to \texttt{Train}, but
another subset, (2)~\texttt{Overnight}: flights that have arrival time earlier
than the departure time (the dataset does not explicitly report the date of
arrival), and (3)~\texttt{Mixed}: a mixture of \texttt{Daytime}
and~\texttt{Overnight}. \reviseone{A few sample tuples of this dataset are in
Fig.~\ref{fig:flights}.}

\reviseone{Our experiment involves the following steps: (1)~\system computes
\dis $\Phi$ over \texttt{Train}, while \emph{ignoring} the target attribute
\texttt{delay}. (2)~We compute average \invariant violation for all four
datasets---\texttt{Train}, \texttt{Daytime}, \texttt{Overnight}, and
\texttt{Mixed}---against $\Phi$ (first row of Fig.~\ref{fig:airlines-summary}).
(3)~We train a linear regression model over \texttt{Train}---including
\texttt{delay}---that learns to predict arrival delay. (4)~We compute mean
absolute error (MAE) of the prediction accuracy of the regressor over the four
datasets (second row of Fig.~\ref{fig:airlines-summary}).} We find
that \invariant violation is a very good proxy for prediction error, as they
vary in a similar manner across the four datasets. The reason is that the model
implicitly assumes that the \invariants (e.g., $AT - DT - DUR \approx 0$)
derived by \system will always hold, and, thus, deteriorates when the assumption
no longer holds.

\looseness-1 \reviseone{To observe the rate of false positives and false
negatives, we investigate the relationship between constraint violation and
prediction error at tuple-level granularity.} We sample 1000 tuples from
\texttt{Mixed} and organize them by decreasing order of violations
(Fig.~\ref{fig:airlines}). \reviseone{For all the tuples (on the left) that
incur high \invariant violations, the regression model incurs high error for
them as well. This implies that \system reports no false positives. There are
some false negatives (right part of the graph), where violation is low, but the
prediction error is high. Nevertheless, such false negatives are very few.}

\smallskip

\noindent\textbf{HAR.} \looseness-1
On the HAR dataset, we design a supervised classification task to identify
persons from their activity data \reviseone{that contains 36 numerical
attributes.} We construct \texttt{train\_x} with data for sedentary activities
(lying, standing, and sitting), and \texttt{train\_y} with the corresponding
person-IDs. We learn \dis on \texttt{train\_x}, and train a Logistic Regression
(LR) classifier using the annotated dataset $[\texttt{train\_x};
\texttt{train\_y}]$. During serving, we mix mobile activities (walking and
running) with held-out data for sedentary activities and observe how the
classification's mean accuracy-drop \reviseone{(i.e., how much the mean
prediction accuracy decreases compared to the mean prediction accuracy over the
training data)} relates to average \invariant violation. \reviseone{To avoid
any artifact due to sampling bias, we repeat this experiment $10$ times for
different subsets of the data by randomly sampling $5000$ data points for each of
training and serving.} Fig.~\ref{fig:har-ml-experiment}~depicts our
findings: classification degradation has a clear positive correlation with
violation (pcc = 0.99 with p-value = 0).

\smallskip

\looseness-1 \noindent \revisetwo{\emph{Noise sensitivity.} \label{noise}
Intuitively, noise weakens \dis by increasing variance in the training data,
which results in reduced violations of the serving data. However, this is
desirable: as more noise makes machine-learned models less likely to overfit,
and, thus, more robust. In our experiment for observing noise sensitivity of
\dis, we use only mobile activity data as the serving set and start with
sedentary data as the training set. Then we gradually introduce noise in the
training set by mixing mobile activity data. As
Fig.~\ref{fig:har-ml-experiment-noise} shows, when more noise is added to the
training data, \dis start getting weaker; this leads to reduction in
violations. However, the classifier also becomes robust with more noise, which
is evident from gradual decrease in accuracy-drop (i.e., increase in accuracy).
Therefore, even under the presence of noise, the positive correlation between
classification degradation and violation persists (pcc = 0.82 with p-value =
0.002).}

\smallskip
\noindent
\fbox{
\parbox{0.96\columnwidth}{
\emph{Key takeaway:} 
\system derives \dis whose violation is a strong proxy of model prediction
accuracy. \reviseone{Their correlation persists even in the presence of noise.}}
}

\begin{figure}[t!]
	\centering
	\hspace{-4mm}
	\begin{subfigure}[t]{0.33\textwidth}
		\centering
	\includegraphics[width=0.9\linewidth]{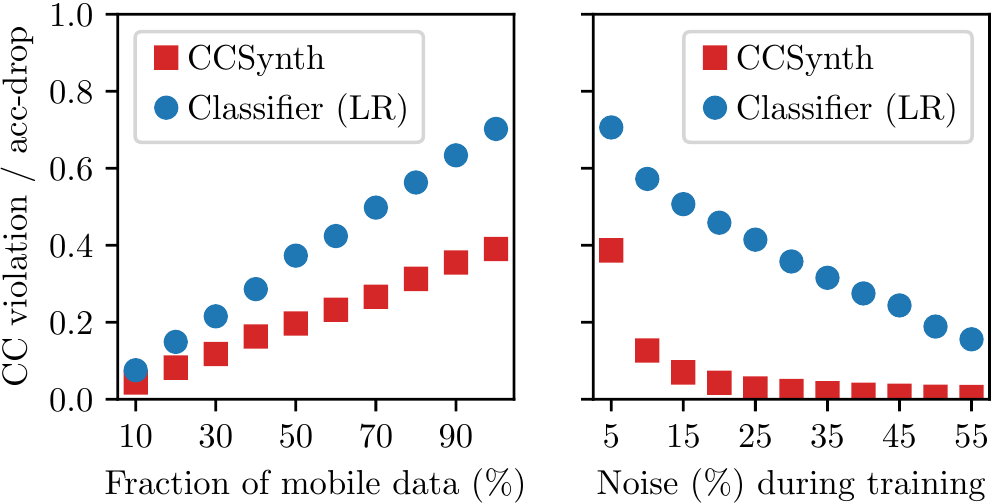}
	\vspace{-2mm}
	\caption{\phantom{randomrand}}
	\label{fig:har-ml-experiment}
	\end{subfigure}	
	\hspace{-17mm}
	\begin{subfigure}[t]{0.01\textwidth}
		\centering
	\vspace{-2mm}
	\caption{}
	\label{fig:har-ml-experiment-noise}
	\end{subfigure}	
	\hspace{7mm}
	\begin{subfigure}[t]{0.2\textwidth}
		\centering
	\includegraphics[width=0.81\linewidth]{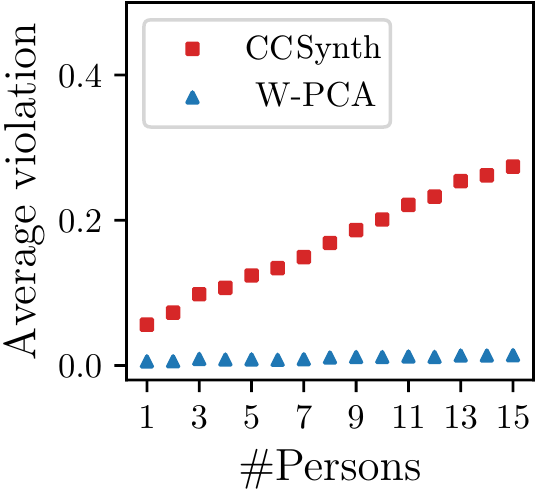}
		\vspace{-2mm}
	\caption{}
	\label{fig:gradual-drift-har}
	\end{subfigure}
	\hspace{-7mm}
		\vspace{-3mm}	
	\caption{(a)~As a higher fraction of mobile activity data is mixed with
sedentary activity data, \dis are violated more, and the classifier's
mean accuracy-drop increases. 
(b)~\revisetwo{As more noise is added during training, \dis get weaker,
leading to less violation and decreased accuracy-drop.}
(c)~\system detects the gradual local drift on the
HAR dataset as more people start changing their
activities. In contrast, weighted-PCA (W-PCA) fails to detect drift in absence of
a strong global drift.} 
		\vspace{2mm}	
	\centering
	\includegraphics[width=1\linewidth]{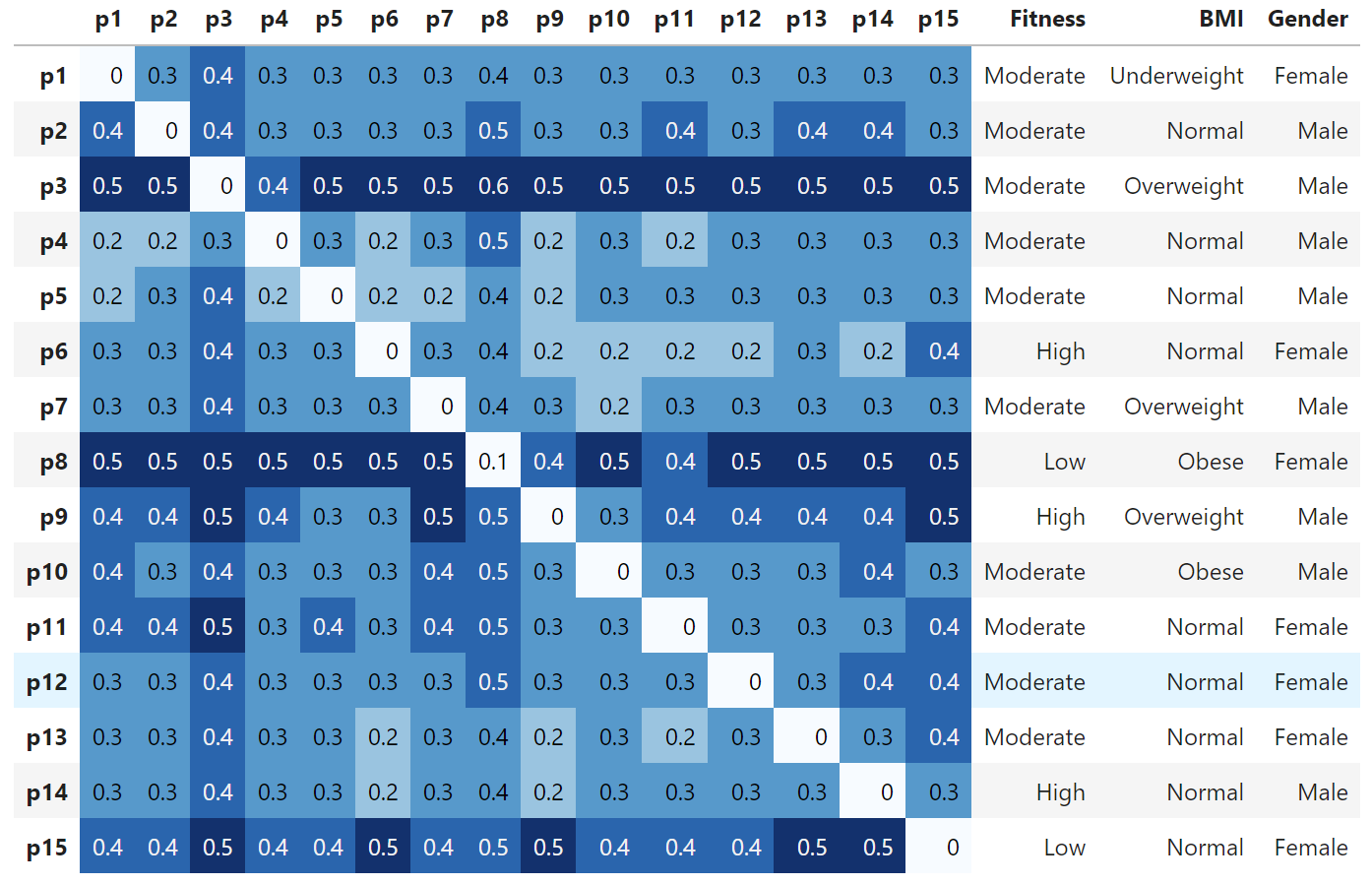}
		\vspace{-7mm}	
	\caption{\looseness-1 Inter-person \invariant violation heat map. Each person has a very low self-violation.}
	\label{fig:har-inter-person-drift-heatmap}
		\vspace{1mm}	
	\centering
\end{figure}

\smallskip

\subsection{Data Drift}\label{exp-invariants-for-drift}
%
We now present results of using \dis for drift-detection; specifically, for
\emph{quantifying} drift in data. Given a baseline dataset $D$, and a new
dataset $D'$, the drift is measured as average violation of tuples in $D'$ on
\dis learned for $D$.

\smallskip

\noindent\textbf{HAR.} We perform two drift-quantification experiments on
HAR:

\smallskip

\noindent \emph{Gradual drift.} \looseness-1 For observing how \system detects
gradual drift, we introduce drift in an organic way. The initial training
dataset contains data of exactly one activity for each person. This is a
realistic scenario as one can think of it as taking a snapshot of what a group
of people are doing during a reasonably small time window. We introduce gradual
drift to the initial dataset by altering the activity of one person at a time.
To control the amount of drift, we use a parameter $K$. \reviseone{When $K =
1$, the first person switches their activity, i.e., we replace the tuples
corresponding to the first person performing activity A with new tuples that
correspond to the same person performing another activity B.} When $K = 2$, the
second person switches their activity in a similar fashion, and so on. As we
increase $K$ from $1$ to $15$, we expect a gradual increase in the drift
magnitude compared to the initial training data. When $K = 15$, all persons
have switched their activities from the initial setting, and we expect to
observe maximum drift. We repeat this experiment $10$ times, and display the
average \invariant violation in Fig.~\ref{fig:gradual-drift-har}: the drift
magnitude (violation) indeed increases as more people alter their activities.

\begin{figure*}[t]
	\centering	
	\includegraphics[width=.85\linewidth]{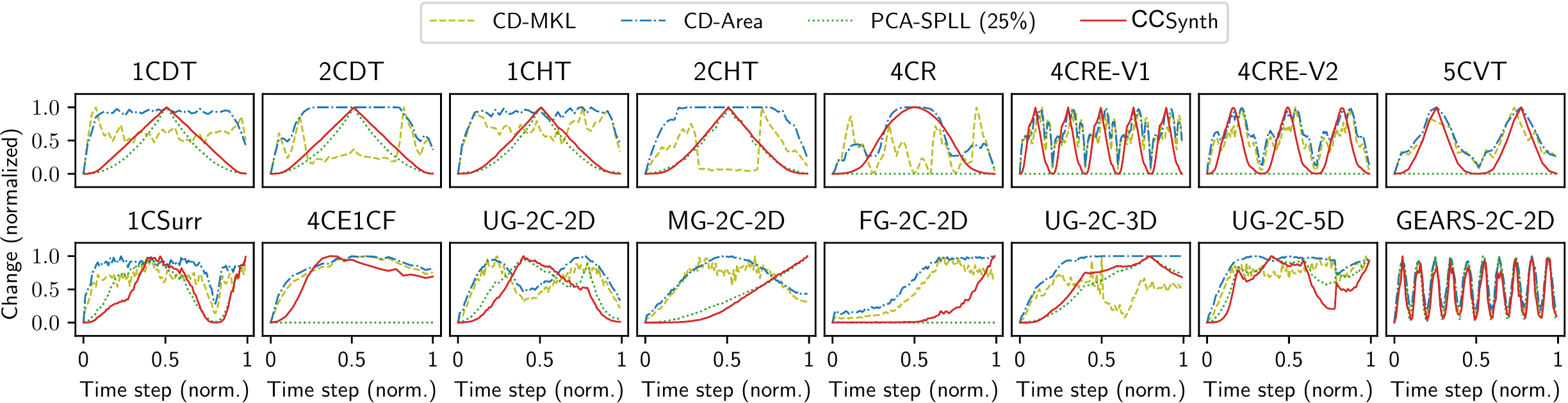}	
		\vspace{-3mm}
	 \caption{In the EVL benchmark, \system quantifies drift correctly for all
	 cases, outperforming other approaches. PCA-SPLL fails to detect drift in a few cases by
 	 discarding all principal components; CD-MKL and CD-Area are too sensitive to
 	 small drift and detect spurious drifts.}
	\vspace{-2mm}
	\label{fig:drift-baseline-comparison-EVL}
\end{figure*}

The baseline weighted-PCA approach (W-PCA) fails to model local \invariants
(who is doing what), and learns some weaker global \invariants indicating that
``a group of people are performing some activities''. Thus, it fails to detect
the gradual local drift. \system can detect drift when individuals switch
activities, as it learns \emph{disjunctive} \invariants that encode who is
doing what.

\smallskip

\noindent \looseness-1 \emph{Inter-person drift.} \reviseone{The goal of this
experiment is to observe how effectively \dis can model the representation of an
entity and whether such learned representations can be used to accurately
quantify drift between two entities.} We use half of each person's data to
learn the \invariants, and compute violation on the held-out data. \system
learns disjunctive \invariants for each person over all activities, and then we
use the violation w.r.t.\ the learned \invariants to measure how much the other
persons drift. While computing drift between two persons, we compute
activity-wise \invariant violation scores and then average them out. In
Fig.~\ref{fig:har-inter-person-drift-heatmap}, the violation score at row
\texttt{p1} and column \texttt{p2} denotes how much \texttt{p2} drifts from
\texttt{p1}. As one would expect, we observe a very low self-drift across the
diagonal. Interestingly, our result also shows that some people are more
different from others, which appears to have some correlation with (the hidden
ground truth) fitness and BMI values. This asserts that the \invariants we
learn for each person are an accurate abstraction of that person's activities,
as people do not deviate too much from their usual activity patterns.

\smallskip

\noindent \textbf{EVL.} We now compare \system against other state-of-the-art
drift detection approaches on the EVL benchmark.

\smallskip

\noindent\emph{Baselines.}
We use two drift-detection baselines as described below:

\smallskip

(1)~PCA-SPLL~\cite{DBLP:journals/tnn/KunchevaF14}\footnote{\scriptsize{SPLL
source code:} \scriptsize{\url{github.com/LucyKuncheva/Change-detection}}},
similar to us, also argues that principal components with lower variance are
more sensitive to a general drift, and uses those for dimensionality reduction.
It then models multivariate distribution over the reduced dimensions and
applies semi-parametric log-likelihood (SPLL) to detect drift between two
multivariate distributions. However, PCA-SPLL discards all high-variance
principal components and does not model disjunctive \invariants.

\smallskip
 
(2)~CD (Change
Detection)~\cite{DBLP:conf/kdd/QahtanAWZ15}\footnote{\scriptsize{CD source
code:} \scriptsize{\url{mine.kaust.edu.sa/Pages/Software.aspx}}} is another
PCA-based approach for drift detection in data streams. But unlike PCA-SPLL, it
ignores low-variance principal components. CD projects the data onto top $k$
high-variance principal components, which results into multiple univariate
distributions. We compare against two variants of CD: CD-Area, which uses the
intersection area under the curves of two density functions as a divergence
metric, and CD-MKL, which uses Maximum KL-divergence as a symmetric divergence
metric, to compute divergence between the univariate distributions.

\smallskip

\looseness-1 Fig.~\ref{fig:drift-baseline-comparison-EVL} depicts how \system
compares against CD-MKL, CD-Area, and PCA-SPLL, on 16 datasets in the
EVL benchmark. For PCA-SPLL, we retain principal components that contribute to
a cumulative explained variance below 25\%. Beyond drift detection, which just
detects if drift is above some threshold, we focus on drift
quantification. A tuple $(x,y)$ in the plots denotes that drift
magnitude for dataset at $x^{th}$ time window, w.r.t.\ the dataset at the first
time window, is $y$. Since different approaches report drift magnitudes in
different scales, we normalize the drift values within $[0, 1]$. Additionally,
since different datasets have different number of time windows,
for the ease of exposition, we normalize the time window indices. Below
we state our key findings from this experiment:

\smallskip

\looseness-1 \emph{\system's drift quantification matches the ground truth.} In
all of the datasets in the EVL benchmark, \system is able to correctly quantify
the drift, which matches the ground truth~\cite{evlVideo} exceptionally well.
In contrast, as CD focuses on detecting the drift point, it is ill-equipped to
precisely quantify the drift, which is demonstrated in several cases (e.g.,
2CHT), where CD fails to distinguish the deviation in drift magnitudes. In
contrast, both PCA-SPLL and \system correctly quantify the drift.
\revisetwo{Since CD only retains high-variance principal components, it is more
susceptible to noise and considers noise in the dataset as significant drift,
which leads to incorrect drift quantification. In contrast, PCA-SPLL and
\system ignore the noise and only capture the general notion of drift.} In all
of the EVL datasets, we found CD-Area to work better than CD-MKL, which also
agrees with the authors' experiments.

\smallskip

\looseness-1 \emph{\system models local drift.} When the dataset contains
instances from multiple classes, the drift may be just local, and not global
(e.g., 4CR dataset as shown in the Appendix). In such cases, PCA-SPLL fails to
detect drift (4CR, 4CRE-V2, and FG-2C-2D). In contrast, \system learns
disjunctive \invariants and quantifies local drifts accurately.

\smallskip
\noindent
\fbox{
\parbox{0.95\columnwidth}{
\emph{Key takeaways:} 
\system can effectively detect data drift, both global and local, is robust
across drift patterns, and significantly outperforms the state-of-the-art
methods.}
}

%% file: 7_related-work.tex

There is extensive literature on
data-profiling~\cite{DBLP:journals/vldb/AbedjanGN15} primitives that model
relationships among data attributes, such as functional dependencies
(FD)~\cite{papenbrock2015functional, DBLP:conf/sigmod/ZhangGR20} and their
variants (metric~\cite{koudas2009metric},
conditional~\cite{DBLP:conf/icde/FanGLX09},
soft~\cite{DBLP:conf/sigmod/IlyasMHBA04}, approximate~\cite{huhtala1999tane,
kruse2018efficient}, relaxed~\cite{caruccio2016discovery},
pattern~\cite{qahtan2020pattern}, etc.), differential
dependencies~\cite{song2011differential}, denial
constraints~\cite{DBLP:journals/pvldb/ChuIP13,
DBLP:journals/corr/abs-2005-08540, DBLP:journals/pvldb/BleifussKN17,
pena2019discovery}, statistical constraints~\cite{DBLP:conf/sigmod/YanSZWC20},
etc. However, none of them focus on learning approximate arithmetic
relationships that involve multiple numerical attributes in a noisy setting,
which is the focus of our work. Some variants of
FDs~\cite{DBLP:conf/sigmod/IlyasMHBA04, koudas2009metric, huhtala1999tane,
kruse2018efficient, caruccio2016discovery} consider noisy setting, but they
require noise parameters to be explicitly specified by the user. In contrast,
we do not require any explicit noise parameter.

The issue of trust, resilience, and interpretability of artificial intelligence
(AI) systems has been a theme of increasing interest
recently~\cite{DBLP:conf/cav/Jha19, DBLP:journals/crossroads/Varshney19,
DBLP:journals/corr/abs-1904-07204, DBLP:conf/mlsys/KangRBZ20}, particularly for safety-critical
data-driven AI systems~\cite{DBLP:journals/bigdata/VarshneyA17,
DBLP:conf/hicons/TiwariDJCLRSS14}. A standard way to decide whether to trust a
classifier or not, is to use the classifier-produced confidence score. However,
as a prior work~\cite{DBLP:conf/nips/JiangKGG18} argues, this is not always
effective since the classifier's confidence scores are not well-calibrated.
While some recent techniques~\cite{DBLP:conf/nips/JiangKGG18,
DBLP:conf/sigmod/SchelterRB20, DBLP:journals/corr/HendrycksG16c,
DBLP:journals/corr/abs-1812-02765} aim at validating the inferences made by
machine-learned models on unseen tuples, they usually require knowledge of the
inference task, access to the model, and/or expected cases of data shift, which
we do not. Furthermore, they usually require costly hyper-parameter tuning and
do not generate closed-form data profiles like \dis
(Fig.~\ref{relatedWorkMatrix}). Prior work on data drift, change detection, and
covariate shift~\cite{DBLP:conf/sigmod/Aggarwal03, DBLP:journals/tnn/BuAZ18,
dasu2006information, DBLP:journals/eswa/MelloVFB19, DBLP:conf/kdd/ReisFMB16,
DBLP:journals/inffus/FaithfullDK19, DBLP:conf/icml/Ho05, hooi2019branch,
DBLP:conf/sdm/KawaharaS09, DBLP:conf/vldb/KiferBG04,
DBLP:journals/eswa/SethiK17, DBLP:conf/kdd/SongWJR07, DBLP:conf/sac/IencoBPP14,
DBLP:conf/cbms/TsymbalPCP06, DBLP:journals/corr/WangA15,
DBLP:conf/sbia/GamaMCR04, DBLP:conf/sdm/BifetG07, gaber2006classification,
DBLP:journals/tnn/RutkowskiJPD15, DBLP:conf/iri/SethiKA16} relies on modeling
data distribution. However, data distribution does not capture constraints,
which is the primary focus of our work.

Few works~\cite{DBLP:journals/corr/abs-1812-02765,
DBLP:journals/corr/HendrycksG16c, DBLP:journals/corr/abs-1909-03835} use
autoencoder's~\cite{hinton2006reducing, rumelhart1985learning} input
reconstruction error to determine if a new data point is out-of-distribution.
Our approach is similar to outlier-detection~\cite{kriegel2012outlier} and
one-class-classification~\cite{DBLP:conf/icann/TaxM03}. However, \dis differ
from these approaches as they perform under the additional requirement to
generalize the data in a way that is exploited by a given class of ML models.
In general, there is a clear gap between representation learning (that models
data likelihood)~\cite{hinton2006reducing, rumelhart1985learning,
achlioptas2017learning, karaletsos2015bayesian} and the (constraint-oriented)
data-profiling techniques to address the problem of trusted AI and our aim is
to bridge this gap.

%% file: 8_summary.tex

We introduced \dis, and the notion of \nc tuples for trusted machine learning.
We presented an efficient and highly-scalable approach for synthesizing \dis;
and demonstrated their effectiveness to tag \nc tuples and quantify data drift.
The experiments validate our theory and our principle of using low-variance
\views to generate effective \dis. We have studied only two use-cases from a
large pool of potential applications using linear \dis. In future, we want to
explore more powerful nonlinear \dis using autoencoders. Moreover, we plan to
explore approaches to learn \dis in a decision-tree-like structure where
categorical attributes will guide the splitting conditions and leaves will
contain simple \dis. Further, we envision a mechanism---built on top of
\dis---to explore differences between datasets.

%% file: 9_appendix.tex

\section{System Parameters}\label{app:F}
\looseness-1
Our technique for deriving (unnormalized) importance factor $\gamma_k$, for
bounded-\view \invariant on \view $F_k$, uses the mapping $\frac{1}{\log(2 +
\sigma(F_k(D)))}$. This mapping correctly translates our principles for
quantifying violation by putting high weight on \dis constructed from
low-variance projections, and low weight on \dis constructed from high-variance
projections. While this mapping works extremely well across a large set of
applications (including the ones shown in our experimental results), our
quantitative semantics are not limited to any specific mapping. In fact, the
function to compute importance factors for bounded-\views can be user-defined
(but we do not require it from the user). Specifically, a user can plug in any
custom function to derive the (unnormalized) importance factors. Furthermore,
our technique to compute the bounds $\lb$ and $\ub$ can also be customized (but
we do not require it from the user either). Depending on the application
requirements, one can apply techniques used in machine learning literature
(e.g., cross-validation) to tighten or loosen the \dis by tuning these
parameters. However, we found our technique---even without any
cross-validation---for deriving these parameters to be very effective in most
practical applications.

\section{Proof of Lemma~\ref{LEMMA:MAIN}}\label{app:A}

\begin{proof}
    Pick $\beta_1, \beta_2$ s.t. $\beta_1^2+\beta_2^2=1$ and the following equation holds:
	\begin{align}\label{condition}
		\sign{\rho_{F_1,F_2}}\beta_1\stddev{F_1(D)} + \beta_2\stddev{F_2(D)} = 0
	\end{align}
    Let $t$ be any tuple that is incongruous w.r.t. $\langle F_1, F_2
    \rangle$. Now, we compute how far $t$ is from the mean of the \view $F$
    on $D$:
    \begin{align*}
        |\delF{F}{t}| &= |F(t) - \avg{F(D)}|
        \\
        &=  |\beta_1 F_1(t) + \beta_2 F_2(t) - \avg{ \beta_1 F_1(D) + \beta_2 F_2(D) }|
        \\
        &=  |\beta_1 \delF{F_1}{t} + \beta_2 \delF{F_2}{t}|
        \\
        &=  |\beta_1 \delF{F_1}{t}| + |\beta_2 \delF{F_2}{t}|
    \end{align*}
    The last step is correct only when $\beta_1 \delF{F_1}{t}$ and $\beta_2
    \delF{F_2}{t}$ are of same sign. We prove this by cases:

	\smallskip
	
         \noindent (Case 1). $\rho_{F_1,F_2} \geq \frac{1}{2}$.
         In this case, $\beta_1$ and
 	 $\beta_2$ are of different signs due to Equation~\ref{condition}. Moreover,
 	 since $t$ is incongruous w.r.t. $\langle F_1, F_2\rangle$, $\delF{F_1}{t}$
 	 and $\delF{F_2}{t}$ are of different signs. Hence, $\beta_1 \delF{F_1}{t}$
 	 and $\beta_2 \delF{F_2}{t}$ are of same sign.

	\smallskip    
	
         \noindent (Case 2). $\rho_{F_1,F_2} \leq -\frac{1}{2}$.
         In this case, $\beta_1$ and
  	 $\beta_2$ have the same sign due to Equation~\ref{condition}. Moreover,
  	 since $t$ is incongruous w.r.t. $\langle F_1, F_2\rangle$, $\delF{F_1}{t}$
  	 and $\delF{F_2}{t}$ are of same sign. Hence, $\beta_1 \delF{F_1}{t}$ and
  	 $\beta_2 \delF{F_2}{t}$ are of same sign.

\smallskip

Next, we compute the variance of $F$ on $D$:

    \begin{align*}
        \stddev{F(D)}^2 &{=}\frac{1}{|D|}\sum_{t\in D} ( \beta_1 \delF{F_1}{t} {+} \beta_2 \delF{F_2}{t})^2
        \\ &{=}  \beta_1^2\stddev{F_1(D)}^2  {+} \beta_2^2\stddev{F_2(D)}^2 
		\\ &\phantom{{=}  \beta_1^2\stddev{F_1(D)}^2 \; }{+} 2 \beta_1\beta_2 \rho_{F_1,F_2}\stddev{F_1(D)}\stddev{F_2(D)}
        \\ &{=}  \beta_1^2\stddev{F_1(D)}^2  {+} \beta_1^2\stddev{F_1(D)}^2  {-} 2 \beta_1^2 |\rho_{F_1,F_2}|\stddev{F_1(D)}^2
        \\ &{=}  2\beta_1^2\stddev{F_1(D)}^2 (1 - |\rho_{F_1,F_2}|)
    \end{align*}

    Hence, $\stddev{F(D)} =
    \sqrt{2(1-|\rho_{F_1,F_2}|)}|\beta_1|\stddev{F_1(D)}$, which is also equal
    to $\sqrt{2(1-|\rho_{F_1,F_2}|)}|\beta_2|\stddev{F_2(D)}$.
    Since $\sqrt{2(1-|\rho_{F_1,F_2}|)}| \leq 1$, and since $|\beta_k| < 1$, we conclude that
    $\stddev{F(D)} < \stddev{F_k(D)}$.

	Now, we compute $\frac{|\delF{F}{t}|}{\stddev{F(D)}}$ next using the above
	derived facts about $|\delF{F}{t}|$ and $\stddev{F(D)}$.
        \begin{align*}
            \frac{|\delF{F}{t}|}{\stddev{F(D)}} 
            &> \frac{|\beta_1 \delF{F_1}{t}|}{\sqrt{2(1-|\rho_{F_1,F_2}|)}|\beta_1|\stddev{F_1(D)}}
            \\
            &= \frac{|\delF{F_1}{t}|}{\sqrt{2(1-|\rho_{F_1,F_2}|)}\stddev{F_1(D)}}
            &\geq \frac{|\delF{F_1}{t}|}{\stddev{F_1(D)}}
    \end{align*}
	 \looseness-1 The last step uses the fact that $|\rho_{F_1,F_2}| \geq \frac{1}{2}$.
	 Similarly, we also get $\frac{|\delF{F}{t}|}{\stddev{F(D)}} {>}
	 \frac{|\delF{F_2}{t}|}{\stddev{F_2(D)}}$. Hence, $\phi_F$ is stronger than
	 both $\phi_{F_1}$ and $\phi_{F_2}$ on $d$, using Lemma~\ref{lemma:helper}.
	 This completes the proof.
\end{proof}

\section{Proof of Theorem~\ref{THM:MAIN}}\label{app:B}
\begin{proof}
    First, 
    we use Lemma~\ref{LEMMA:MAIN} on $F_i,F_j$ to construct $F$. We initialize $I := \{i,j\}$.
    Next, we repeatedly do the following:
    We iterate over all $F_k$, where $k\not\in I$, and check if 
    $|\rho_{F,F_k}| \geq \frac{1}{2}$ for some $k$.  If yes,
    we use Lemma~\ref{LEMMA:MAIN} (on $F$ and $F_k$) to update $F$ to be the new \view returned by the lemma. 
    We update $I := I \cup \{k\}$, and continue the iterations.
    If $|\rho_{F,F_k}| < \frac{1}{2}$ for all $k\not\in I$, then we stop. The final $F$ and index set
    $I$ can easily be seen to satisfy the claims of the theorem.
\end{proof}

\section{Proof of Theorem~\ref{THM:ALGO2-CORRECTNESS}}\label{app:C}
We first provide some additional details regarding the statement of the
theorem. Since standard deviation is not scale invariant, if there is no
constraint on the norm of the linear projections, then it is possible to scale
down the linear \views to make their standard deviations arbitrarily small.
Therefore, claim~(1) can not be proved for {\em{any}} linear \view, but only
linear \views whose $2$-norm is not too ``small''. Hence, we restate the
theorem with some additional technical conditions.

\begin{sloppypar} Given a numerical dataset $D$, let $\mathcal{F} =
\{F_1,F_2,\ldots, F_K\}$ be the set of linear \views returned by
Algorithm~\ref{fig:pca}. Let $\sigma^* = \min_k^{K}\stddev{F_k(D)}$. WLOG,
assume $\sigma^* = \stddev{F_1(D)}$ where $F_1 = \vec{A}^T\vec{w^*}$. Assume
that the attribute mean is zero for all attributes in $D$ (call this
Condition~1). Then,
\end{sloppypar}	 
	\begin{enumerate}[label=(\arabic*)]
	 \item $\sigma^* \leq \stddev{F(D)}$ for every possible linear \view $F$
	 whose $2$-norm is sufficiently large, i.e., we require $||\vec{w}|| \geq 1$
	 for $F = \vec{A}^T\vec{w}$. If we do not assume Condition~1, then the
	 requirement changes to $||\vec{w}|| \geq ||\vec{w}^{*e}|| -
	 \avg{D^T\vec{w}}||$. Here $\vec{w}^{*e}$ is the vector constructed by augmenting
	 a dimension to $\vec{w}^*$ to turn it to an eigenvector of ${D^e}^TD^e$ where $D^e =
	 [\vec{1}; D]$.
        \item $\forall F_j, F_k \in \mathcal{F}$ s.t.\ $F_j\neq F_k$, $\rho_{F_j, F_k} = 0$.
            If we do not assume Condition~1, then the correlation coefficient is close to $0$ for
            those $F_j, F_k$ whose corresponding eigenvalues are much smaller than $|D|$.
	\end{enumerate}

\begin{proof}
The proof uses the following facts:
    \begin{description}
    	\item [(Fact~1)] If we add a constant $c$ to each element of a set $S$
    	of real values to get a new set $S'$, then $\stddev{S} = \stddev{S'}$.
    	
    	\item [(Fact~2)] The Courant-Fischer min-max theorem~\cite{Horn:2012:MA:2422911} states that the 
    	vector $\vec{w}$ that minimizes
    	$||M\vec{w}||/||\vec{w}||$ is 
    	the eigenvector of $M^T M$ corresponding to the lowest eigenvalue (for any matrix $M$).
    	
    	\item [(Fact~3)] Since $D_N' := [\vec{1}; D_N]$, by definition:
    	$\stddev{D_N\vec{w}} = \frac{||D_N'\vec{w'}||}{\sqrt{|D|}}$, where
    	$\vec{w}' = \left[\begin{matrix}-\avg{D_N\vec{w}} \\
    	\vec{w}\end{matrix}\right]$
    	
    	\item [(Fact~4)] By the definition of variance, $\stddev{S}^2 = ||S||^2 -
    	\avg{S}^2$. 
	\end{description}
	
	\smallskip

    Let $F = \vec{A}^T\vec{w}$ be an arbitrary linear \view. Since $D$ is
    numerical, $D=D_N$. Let $D^e$ denote $D_N'$. (We use the superscript $e$ to
    denote the augmented vector/matrix).
	\[\arraycolsep=1.4pt\def\arraystretch{1.2}
     \begin{array}{rcll}
         \lefteqn{\stddev{D^T\vec{w}}^2}
         \\ & = & \stddev{D^T\vec{w} - \vec{1}\mu}^2  & \mbox{ (Fact~1), $\mu=\avg{D^T\vec{w}}$}
         \\ & = & \stddev{{D^e}^T\vec{w}^e}^2 & \mbox{ where $\vec{w}^e = \left[\begin{matrix}-\mu \\ \vec{w}\end{matrix}\right]$}
         \\ & = & \frac{||{D^e}^T\vec{w}^e||^2}{|D|} & \mbox{ (Fact~3)}
         \\ & \geq & \frac{||{D^e}^T\vec{w}^{*e}||^2 \cdot ||\vec{w}^e||^2}{{|D|} \cdot ||\vec{w}^{*e}||^2} & \mbox{ (Fact~2)}
         \\ & = & (\stddev{{D^e}^T\vec{w}^{*e}}^2 + b^2) \cdot \frac{||\vec{w}^e||^2}{||\vec{w}^{*e}||^2} & \mbox{ (Fact~4), $b=\avg{{D^e}^T\vec{w}^{*e}}$}
         \\ & = & (\stddev{{D}^T\vec{w}^*+c}^2 + b^2) \cdot \frac{||\vec{w^e}||^2}{||\vec{w}^{*e}||^2} & \mbox{ Expand ${D^e}^T\vec{w}^{*e}$}
         \\ & = & (\stddev{{D}^T\vec{w}^*}^2 + b^2) \cdot \frac{||\vec{w}^e||^2}{||\vec{w}^{*e}||^2} & \mbox{ (Fact~1)}
         \\ & = & ({\sigma^*}^2 + b^2) \cdot \frac{||\vec{w}^e||^2}{||\vec{w}^{*e}||^2} & \mbox{ definition of $\sigma^*$}
         \\ & \geq & {\sigma^*}^2  & \mbox{ by assumption $\frac{||\vec{w}^e||^2}{||\vec{w}^{*e}||^2} \geq 1$}
     \end{array}
    \]
    For the last step, we use the technical condition that the norm of the extension of $\vec{w}$ 
    (extended by augmenting the mean over $D\vec{w}$) is 
    at least as large as the norm of extension of $\vec{w}^*$ (extended to make it an eigenvector of ${D^e}^T{D^e}$).
    When Condition~1 holds, 
    ${||\vec{w}^e||^2} = ||\vec{w}||^2$ (because $\avg{F(D)}$ will be $0$ and therefore, 
    $\vec{w}^e = \left[\begin{array}{c}0\\ \vec{w}\end{array}\right]$),
        and ${||\vec{w}^{*e}||^2} = 1$ (for the same reason), and hence 
    $\frac{||\vec{w}^e||^2}{||\vec{w}^{*e}||^2} \geq 1$.

    For part~(2) of the claim, let $F_i = \vec{A}^T \vec{w_i}$ for all $i$, where
    $\vec{w_i}$ are the coefficients of the linear \view\ $F_i$.
		Let $c_i = \avg{F_i(D)}$.
	
	\begin{description}
		\item [(Fact~5)] If Condition~1 holds, $\forall i \; c_i = 0$.
	\end{description}
	
    By construction of $F_i$'s, we know that $w_i$ can be extended to be an eigenvector 
    $\left[\begin{array}{c} d_i\\ \vec{w_i}\end{array}\right]$ of  ${D^e}^T {D^e}$ 
        (with corresponding eigenvalue $\lambda_i$).
    In general, 
	
	\begin{description}
	\item [(Fact~6)] It is easy to work out that $d_i = \frac{-c_i}{1-\frac{\lambda_i}{|D|}}$.
	\end{description}

    Thus, we have:
    \[\arraycolsep=1.4pt\def\arraystretch{1.2}
        \begin{array}{rcl@{\quad}l}
			\lefteqn{\rho_{F_j,F_k}}\\
             & = & \frac{\sum_{t\in D}
            \delF{F_j}{t}\delF{F_k}{t}}{|D|\stddev{F_j(D)}\stddev{F_k(D)}} & \mbox{ (definition of $\rho$)}
            \\ & = & \frac{
            (D\vec{w_j} - c_j\vec{1})^T (D\vec{w_k} - c_k\vec{1})}{|D|\stddev{F_j(D)}\stddev{F_k(D)}}
            \\ & = & \frac{
                (D^e\vec{w}_j^e)^T D^e\vec{w}_k^e}{|D|\stddev{F_j(D)}\stddev{F_k(D)}} & \mbox{ 
               $w_i^e=\left[\begin{array}{c} -c_i\\\vec{w}_i\end{array}\right]$}
            \\ & = & \frac{
            \vec{w}_j^{eT} {D^e}^T D^e\vec{w}_k^e}{|D|\stddev{F_j(D)}\stddev{F_k(D)}}
            \\ & = & \frac{
            \vec{w}_j^{eT} \lambda_k \vec{w}_k^e}{|D|\stddev{F_j(D)}\stddev{F_k(D)}}
               &\mbox{ (Fact~5,6), ${D^e}^T D^e\vec{w}_k^e = \lambda_k \vec{w}_k^e$}
            \\ & = & 0 &\mbox{ (eigenvectors are orthogonal)}
        \end{array}
    \]
    When Condition~1 does not hold, Fact~5 would not hold, but Fact~6 continues to hold, and hence
    by continuity, if $|\frac{\lambda_i}{|D|}|$ is close to $0$,  then $d_i$ will be close
    to $-c_i$, and $\rho_{F_j,F_k}$ will be close to $0$.

\end{proof}

\section{Proof of Proposition~\ref{PROP:EXISTENCE}}\label{app:D}
\begin{proof}
	\begin{sloppypar}
    We show that the \di $\Phi := \forall{f,g\in\CC}: f(D) = g(D) \Rightarrow
    f(t) - g(t) = 0$, is the required \di over $D$ to detect tuples that are
    \nc with respect to $\CC$ and $[D;Y]$.
	\end{sloppypar}
		
    First, we claim that $\Phi$ is a \di for $D$. For this, we need to
    prove that every tuple in $D$ satisfies $\Phi$. 
    Consider any $t'\in D$. We need to prove
    that $f(t') = g(t')$ for all $f, g\in\CC$ s.t. $f(D)=g(D)=Y$. Since $t'\in
    D$, and since $f(D)=g(D)$, it follows that $f(t')=g(t')$. This shows that
    $\Phi$ is a \di for every tuple in $D$.

    Next, we claim that $\Phi$ is {\em{not satisfied}} by exactly tuples that
    are \nc\ w.r.t. $\CC$ and $[D;Y]$. Consider any $t'$ such that
    $\neg\Phi(t')$. By definition of $\Phi$, it follows that there exist $f,
    g\in\CC$ s.t. $f(D)=g(D)=Y$, but $f(t')\neq g(t')$. This is equivalent to
    saying that $t'$ is \nc, by definition. \end{proof}

\section{Motivation for disjunctive \dis}
We now provide an example motivating the need for disjunctive \dis.
\begin{example}
	PCA-based approach fails in cases where there exist different piecewise linear
	trends within the data. If we apply PCA to learn \dis on the entire dataset of
	Fig.~\ref{normalPCA}, it will learn two low-quality \invariants, with very high
	variance. In contrast, partitioning the dataset into three partitions
	(Fig.~\ref{disjointPCA}), and then learning \invariants separately on each
	partition, will result in significant improvement of the learned \invariants.
\end{example}

\section{Implication for \Nc Tuples}\label{app:E}
Here, we provide justification for our definition of \nc tuples.
\begin{proposition}\label{PROP:ONE}
    If $t\in\DDom^m$ is a \nc\ tuple w.r.t. $\CC$ and $[D;Y]$, then
    for any $f\in\CC$ s.t. $f(D)=Y$, there exists a $g\in\CC$ s.t. 
    $g(D)=Y$ but $f(t)\neq g(t)$.
\end{proposition}
\begin{proof}
    By the definition of \nc\ tuple, there exist $g,g' \in \CC$ s.t. $g(D) = g'(D) = Y$, but $g(t)\neq g'(t)$.
    Now, given a function $f \in \CC$ s.t. $f(D)=Y$, the value $f(t)$ can be either equal to $g(t)$ or $g'(t)$, but
    not both. WLOG, say $f(t)\neq g(t)$. Then, we have found a function $g$ s.t. 
    $g(D)=Y$ but $f(t)\neq g(t)$, which completes the proof.
\end{proof}

Note that even when we mark $t$ as \nc, it is possible that the learned model
makes the correct prediction on that tuple. However, there is a good chance that it makes a 
different prediction. Hence, it is useful to be aware of \nc
tuples.

\subsection*{\DIs as Preconditions for Trusted Machine Learning}
Let $\CC$ denote a class of functions. Given a dataset $D$, suppose that a
tuple $t$ is \nc. This means that there exist $f, g\in\CC$ s.t. $f(t) \neq
g(t)$, but $f(D)=g(D)$. Now, consider the logical claim that $f(D) = g(D)$.
Clearly, $f$ is not identical to $g$ since $f(t) \neq g(t)$. Therefore, there
is a nontrivial ``proof'' (in some logic) of the fact that ``for all tuples
$t\in D: f(t)=g(t)$''. This ``proof'' will use some properties of $D$, and let
$\phi$ be the formula denoting these facts. If $\phi_D$ is the characteristic
function for $D$, then the above argument can be written as, $$ \phi_D(\vec{A})
\Rightarrow \phi(\vec{A}), \quad \mbox{ and } \quad \phi(\vec{A}) \Rightarrow
f(\vec{A}) = g(\vec{A}) $$ where $\Rightarrow$ denotes logical implication.

\looseness-1 In words, $\phi$ is a \di for $D$ and it
serves as a {\em{precondition}} in the ``correctness proof'' that shows (a
potentially machine-learned model) $f$ is equal to (potentially a ground truth)
$g$. If a tuple $t$ fails to satisfy the precondition $\phi$, then it is
possible that the prediction of $f$ on $t$ will not match the ground truth
$g(t)$.

\begin{example}\label{ex:two}
	\looseness-1
	 Let $D = \{(0,1), (0,2), (0,3)\}$ be a dataset with two attributes $A_1,
	 A_2$, and let the output $Y$ be $1$, $2$, and $3$, respectively. Let $\CC
	 \subseteq ((\Real\times\Real)\mapsto\Real)$ be the class of linear functions
	 over two variables $A_1$ and $A_2$. Consider a new tuple $(1,4)$. It is
	 \nc\ since there exist two different functions, namely $f(A_1,A_2) = A_2$ and
	 $g(A_1,A_2) = A_1+A_2$, that agree with each other on $D$, but disagree on
	 $(1,4)$. In contrast, $(0,4)$ is not \nc\ because there is no function in
	 $\CC$ that maps $D$ to $Y$, but produces an output different from $4$. We
	 apply Proposition~\ref{PROP:EXISTENCE} on the sets $D$, $Y$, and $\CC$. Here,
	 $\CC$ is the set of all linear functions given by $\{\alpha A_1 + A_2 \mid
	 \alpha\in\Real\}$.
         The \di $\Phi$,
	 whose negation characterizes the \nc\ tuples w.r.t. $\CC$ and
	 $[D;Y]$, is $\forall{\alpha_1,\alpha_2}: \alpha_1 A_1 + A_2
	 = \alpha_2 A_1 + A_2 $, which is equivalent to $A_1 = 0$.
\end{example}

\subsection*{Sufficient Check for \Nc Tuples}
\looseness-1 In practice, finding \dis that are necessary and sufficient for
detecting if a tuple is \nc is difficult. Hence, we focus on weaker \invariants
whose violation is sufficient, but not necessary, to classify a tuple as \nc.
We can use such \invariants to get a procedure that has false negatives (fails
to detect some \nc\ tuples), but no false positives (never identifies a tuple
as \nc\ when it is not).

\subsubsection*{Model Transformation using Equality \Invariants}
For certain \dis, we can prove that a \invariant violation by $t$ implies
that $t$ is \nc by showing that those \invariants can transform a model
$f$ that works on $D$ to a different model $g$ that also works on $D$, but
$f(t)\neq g(t)$. We claim that equality \invariants (of the form $F(\vec{A}) =
0$) are useful in this regard. First, we make the point using the scenario from
Example~\ref{ex:two}.

\begin{example}\label{ex:three}
    Consider the set $\CC$ of functions, and the annotated dataset $[D ; Y]$
    from Example~\ref{ex:one}. The two functions $f$ and $g$, where $f(A_1,A_2)
    = A_2$ and $g(A_1,A_2) = A_1+A_2$, are equal when restricted to $D$; that
    is, $f(D) = g(D)$. What property of $D$ suffices to prove $f(A_1,A_2)=
    g(A_1,A_2)$, i.e., $A_2 = A_1 + A_2$? It is $A_1 = 0$. Going the other way,
    if we have $A_1 = 0$, then $f(A_1,A_2) = A_2 = A_1 + A_2 = g(A_1,A_2)$.
    Therefore, we can use the equality \invariant $A_1 = 0$ to transform the
    model $f$ into the model $g$ in such a way that the $g$ continues to match
    the behavior of $f$ on $D$. Thus, an equality \invariant can be exploited to
    produce multiple different models starting from one given model. Moreover,
    if $t$ violates the equality \invariant, then it means that the models, $f$
    and $g$, would not agree on their prediction on $t$; for example, this
    happens for $t = (1,4)$.
\end{example}

Let $F(\vec{A}) = 0$ be an equality \invariant for the dataset $D$. If a learned
model $f$ returns a real number, then it can be transformed into another model
$f + F$, which will agree with $f$ only on tuples $t$ where $F(t) = 0$. Thus,
in the presence of equality \invariants, a learner can return $f$ or its
transformed version $f + F$ (if both models are in the class $\CC$). This
condition is a ``relevancy'' condition that says that $F$ is ``relevant'' for
class $\CC$. If the model does not return a real, then we can still use
equality \invariants to modify the model under some assumptions that include
``relevancy'' of the \invariant.

\begin{figure}[t]
	\centering
	\resizebox{0.9\columnwidth}{!}{
	\begin{subfigure}{.2\textwidth}
	  \centering
	  \includegraphics[width=1\linewidth]{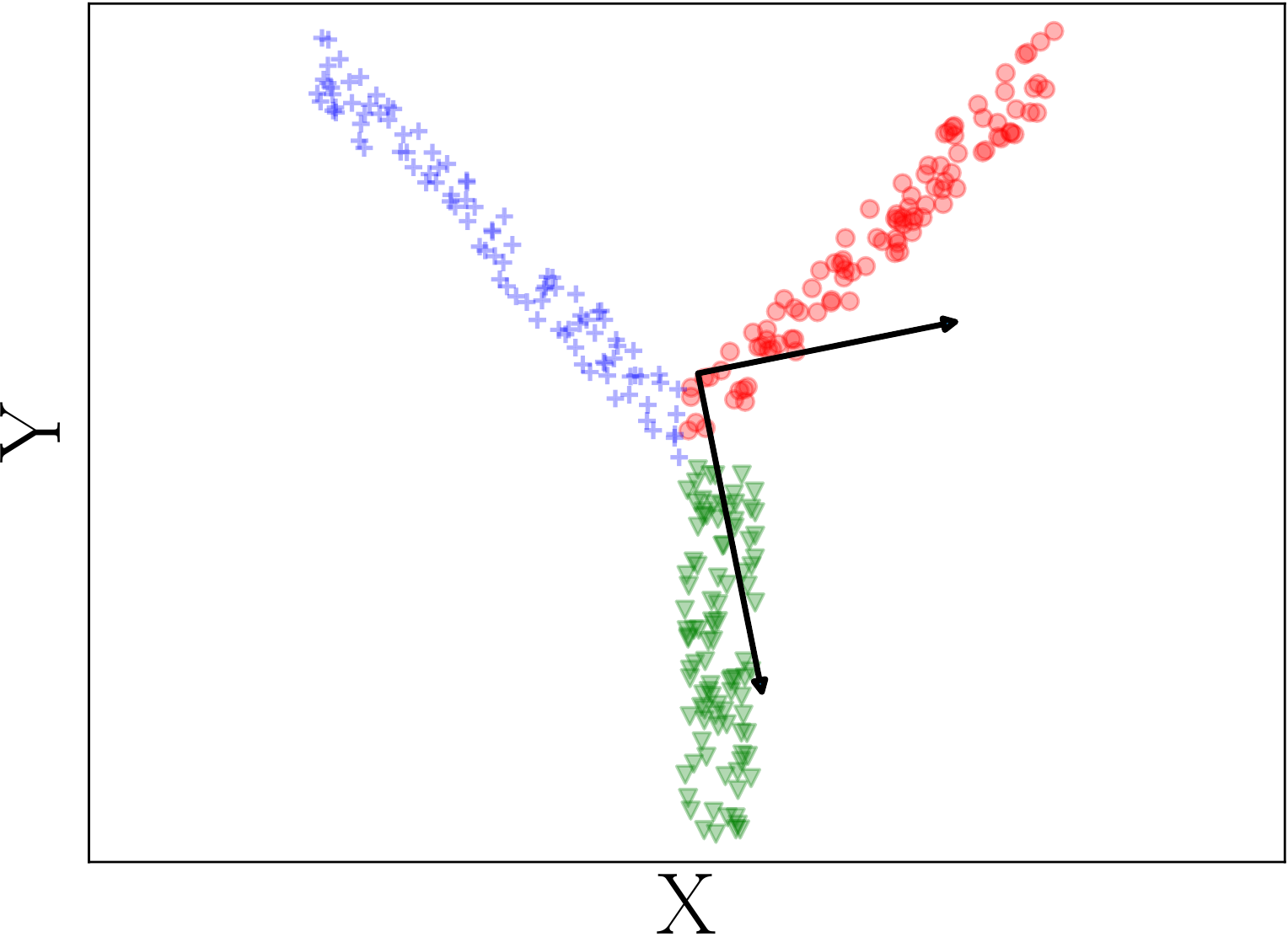}
	  \vspace{-5mm}
	  \caption{PCA}
	  \label{normalPCA}
	\end{subfigure}%
	\hspace{5mm}
	\begin{subfigure}{.2\textwidth}
	  \centering
	  \includegraphics[width=1\linewidth]{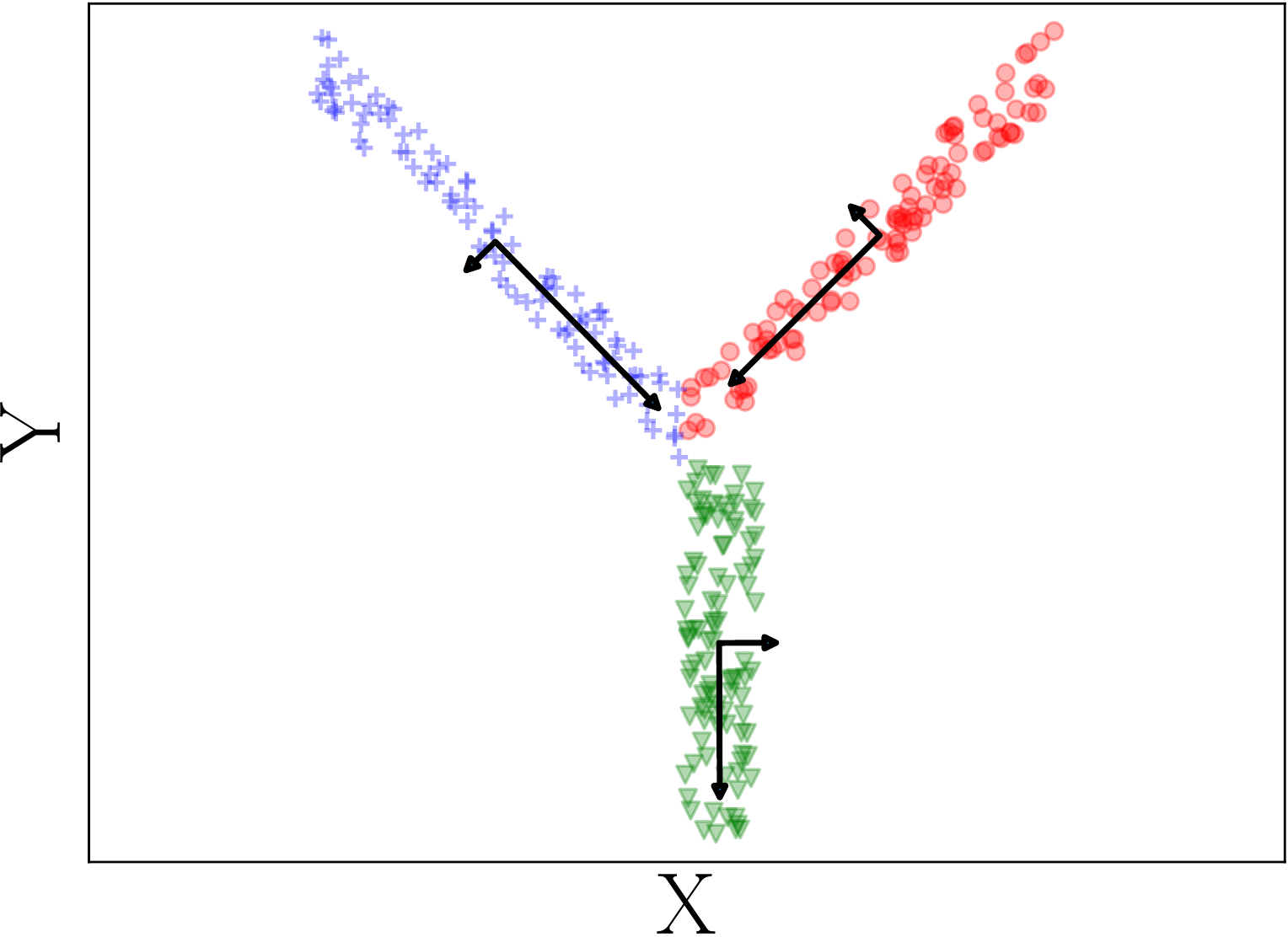}
	  \vspace{-5mm}
	  \caption{Disjoint PCA}
	  \label{disjointPCA}
	\end{subfigure}
	}
	\vspace{-3mm}
	\caption{Learning PCA-based \invariants globally results in low quality \invariants
	when data satisfies strong local \invariants.}
	\label{fig:pcaComparison}
	\vspace{2mm}
\end{figure}

\subsubsection*{A Theorem for Sufficient Check for Unsafe Tuples}
We first formalize the notions of \emph{nontrivial} datasets---which are
annotated datasets such that at least two output labels differ---and
\emph{relevant} \invariants---which are \invariants that can be used to transform
models in a class to other models in the same class.

\smallskip
\paragraph{Nontrivial.}
    An annotated dataset $[D;Y]$ is \emph{nontrivial} if there exist $i, j$ s.t. $y_i
    \neq y_j$. 
\smallskip
 
\paragraph{Relevant.}
    A \invariant $F(\vec{A})=0$ is \emph{relevant} to a class $\CC$ of models
    if whenever $f\in \CC$, then
    $\lambda t: f( \mathtt{ite}(\alpha F(t), t^c, t) ) \in \CC$
    for a constant tuple $t^c$ and real number $\alpha$.
    The if-then-else function $\mathtt{ite}(r, t^c, t)$ returns $t^c$ when $r=1$,
    returns $t$ when $r=0$, and is free to return anything otherwise.
    If tuples admit addition, subtraction, and scaling, then one 
	such if-then-else function is
    $t + r*(t^c-t)$.

\smallskip

We now state a sufficient condition for identifying \nc tuples. 

\begin{theorem}[Sufficient Check for \Nc Tuples]\label{THM:NFP}
    Let $[D;Y] \subset \DDom^m\times\coDom$ be an annotated dataset, 
    $\CC$ be a class of functions, and
    $F$ be a \view on $\DDom^m$ s.t.
\begin{itemize}\setlength{\itemindent}{0.5in}
        \item[A1.] $F(\vec{A})=0$ is a strict \invariant for $D$, 
        \item[A2.] $F(\vec{A})=0$ is relevant to $\CC$,
        \item[A3.] $[D;Y]$ is nontrivial, and
        \item[A4.] there exists $f\in\CC$ s.t. $f(D) = Y$.
    \end{itemize}
    For $t\in \DDom^m$, if $F(t)\neq 0$, then $t$ is \nc.
\end{theorem}

\begin{proof}
    WLOG, let $t_1, t_2$ be the two tuples in $D$ s.t. $y_1\neq y_2$ (A3). 
    Since $f(D) = Y$ (A4), it follows that
    $f(t_1) = y_1 \neq y_2 = f(t_2)$. Let $t$ be a new tuple s.t. $F(t)\neq 0$.
    Clearly, $f(t)$ can not be equal to both $y_1$ and $y_2$.
    WLOG, suppose $f(t) \neq y_1$.
    Now, consider the function $g$ defined by
    $
    \lambda \tau : f( \mathtt{ite}(F(\tau), t_1, \tau) )
    $.
    By~(A2), we know that $g\in \CC$.
    Note that $g(D) = Y$ since for any tuple $t_i \in D$, $F(t_i)=0$ (A1),
    and hence $g(t_i) = f( \mathtt{ite}(0, t_1, t_i)) = f(t_i) = y_i$.
    Thus, we have two models, $f$ and $g$, s.t. $f(D) = g(D) = Y$.
    To prove that $t$ is a \nc\ tuple, 
    we have to show that $f(t) \neq g(t)$. Note that $g(t) = f( \mathtt{ite}(F(t), t_1, t) ) = f(t_1) = y_1$ (by definition of $g$).
    Since we already had $f(t) \neq y_1$, it follows that we have $f(t)\neq g(t)$. This completes the proof.
\end{proof}

We caution that our definition of \nc\ is liberal: existence of even one pair
of functions $f,g$---that differ on $t$, but agree on the training set $D$---is
sufficient to classify $t$ as \nc. It ignores issues related to the
probabilities of finding these models by a learning procedure. Our intended use
of Theorem~\ref{THM:NFP} is to guide the choice for the class of \invariants,
given the class $\CC$ of models, so that we can use violation of a \invariant
in that class as an indication for caution. For most classes of models, linear
arithmetic \invariants are relevant. Our formal development has completely
ignored that data (in machine learning applications) is noisy, and exact
equality \invariants are unlikely to exist. However, the development above can
be extended to the noisy case by replacing exact equality is replaced by
approximate equality. For example, when learning from dataset $D$ and
ground-truth $f'$, we may not always learn a $f$ that exactly matches $f'$ on
$D$, but is only close (in some metric) to $f'$. Similarly, equality
\invariants need not require $F(t) = 0$ for all $t \in D$, but only $F(t)
\approx 0$ for some suitable definition of approximate equality. For ease of
presentation, we have restricted ourselves to the simpler setting, which
nevertheless brings out the salient points.

\begin{example}\label{ex:four}
    Consider the annotated dataset $[D;Y]$ and the class $\CC$, from
    Example~\ref{ex:three}. Consider the equality \invariant $F(A_1,A_2) = 0$,
    where the \view $F$ is defined as $F(A_1, A_2) = A_1$. Clearly, $F(D) = \{0
    \; 0\; 0\}$, and hence, $F(A_1,A_2) = 0$ is a \invariant for $D$. The
    \invariant is also relevant to the class of linear models $\CC$. Clearly,
    $[D;Y]$ is nontrivial, since $y_1 = 1 \neq 2 = y_2$. Also, there exists
    $f\in\CC$ (e.g., $f(A_1,A_2) = A_2$) s.t. $f(D) = Y$. Now, consider the
    tuple $t = (1,4)$. Since $F(t) = 1 \neq 0$, Theorem~\ref{THM:NFP} implies
    that $t$ is \nc.
\end{example}

\subsection*{SQL Check Constraints} Due to the simplicity of the conformance
language to express \dis, they can be easily enforced as SQL check constraints
to prevent insertion of \nc tuples to a database.

\section{Applications of \dis} In database systems, \dis can be used to detect
change in data and query workloads, which can help in database
tuning~\cite{koch2013}. They have application in data cleaning (error detection
and missing value imputation): the violation score serves as a measure of
error, and missing values can be imputed by exploiting relationships among
attributes that \dis capture. \Dis can detect outliers by exposing tuples that
significantly violate them. Another interesting data management application is
\emph{data-diff}~\cite{DBLP:conf/kdd/SuttonHGC18} for exploring differences
between two datasets: our disjunctive \invariants can explain how different
partitions of two datasets vary.

\smallskip

In machine learning, \dis can be used to suggest when to retrain a
machine-learned model. Further, given a pool of machine-learned models and the
corresponding training datasets, we can use \dis to \emph{synthesize} a new
model for a new dataset. A simple way to achieve this is to pick the model such
that \invariants learned from its training data are minimally violated by the
new dataset. Finally, identifying non-conforming tuples is analogous to
{\em{input validation}} that performs sanity checks on an input before it is
processed by an application.

\section{Visualization of Local Drift}

When the dataset contains instances from multiple classes, the drift may be
just local, and not global. Fig.~\ref{fig:4CR} demonstrates a scenario for the
4CR dataset over the EVL benchmark. If we ignore the color/shape of the tuples,
we will not observe any significant drift across different time steps.

\begin{figure}[t]
 	\centering
	\vspace{-1mm}	
 	\includegraphics[width=\linewidth]{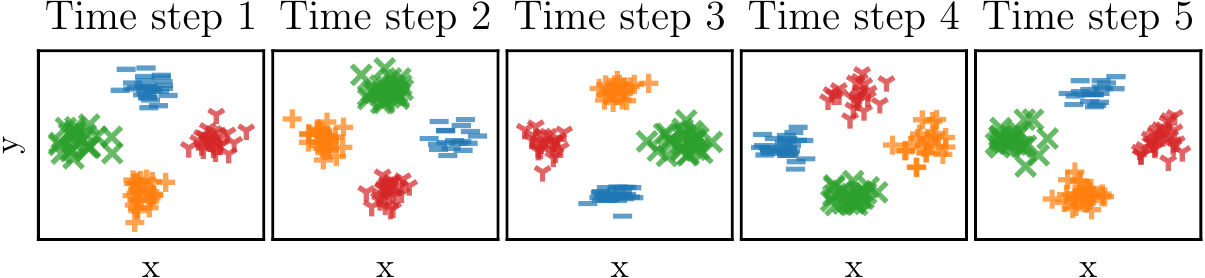}	
	\vspace{-6mm}	
 	 \caption{\looseness-1 Snapshots over time for 4CR dataset with local drift. 
 	 It reaches maximum drift from the initial distribution at time step 3 and
 	 goes back to the initial distribution at time step 5.}
	 \vspace{3mm}
 	\label{fig:4CR}

	\includegraphics[width=0.9\linewidth]{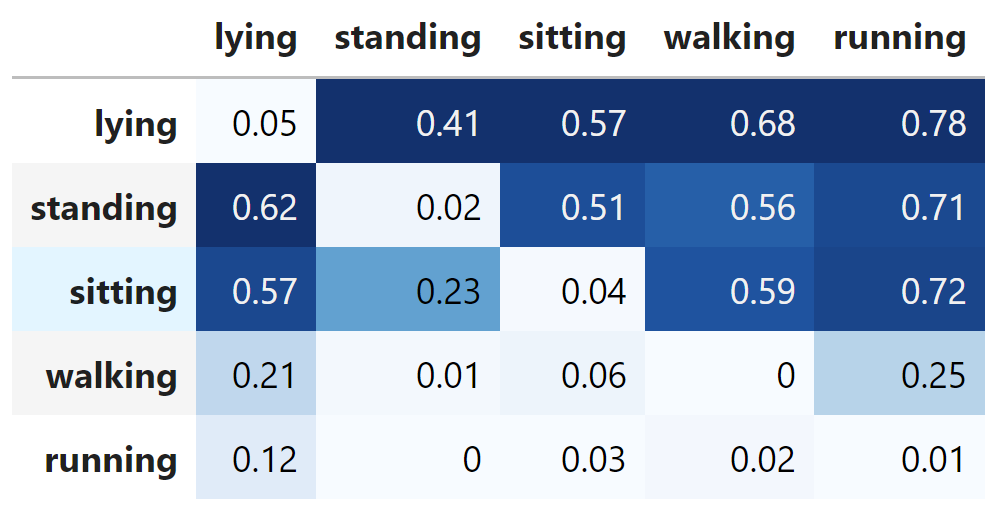}	
		\vspace{-3mm}	
	\caption{Inter-activity \invariant violation heat map. 
	 Mobile activities violate the \invariants of the sedentary
	 activities more.} 
	 \vspace{4mm}	
	\label{fig:har-inter-activity-drift-heatmap}
\end{figure}
 
 \section{More Data-Drift Experiments}
\looseness-1 \noindent \emph{Inter-activity drift.} Similar to inter-person
\invariant violation, we also compute inter-activity \invariant violation over
the HAR dataset (Fig.~\ref{fig:har-inter-activity-drift-heatmap}). Note the
asymmetry of violation scores between activities, e.g., \texttt{running} is
violating the \invariants of \texttt{standing} much more than the other way
around. A close observation reveals that, all mobile activities violate all
sedentary activities more than the other way around. This is because, the
mobile activities behave as a ``safety envelope'' for the sedentary activities.
For example, while a person walks, they also stand (for a brief moment); but
the opposite does not happen.

\section{Explaining Non-conformance}\label{extune}
When a serving dataset is determined to be sufficiently deviated or drifted
from the training set, the next step often is to characterize the difference. A
common way of characterizing these differences is to perform a causality or
responsibility analysis to determine which attributes are most responsible for
the observed drift (non-conformance). We use the violation values produced by
\dis, along with well-established principles of causality, to quantify
responsibility for non-conformance.

\smallskip 

\begin{figure*}
	\centering
	\resizebox{0.9\linewidth}{!}
	{
	\begin{subfigure}[b]{.16\textwidth}
	  \centering
	  \includegraphics[width=\linewidth, height=70mm]{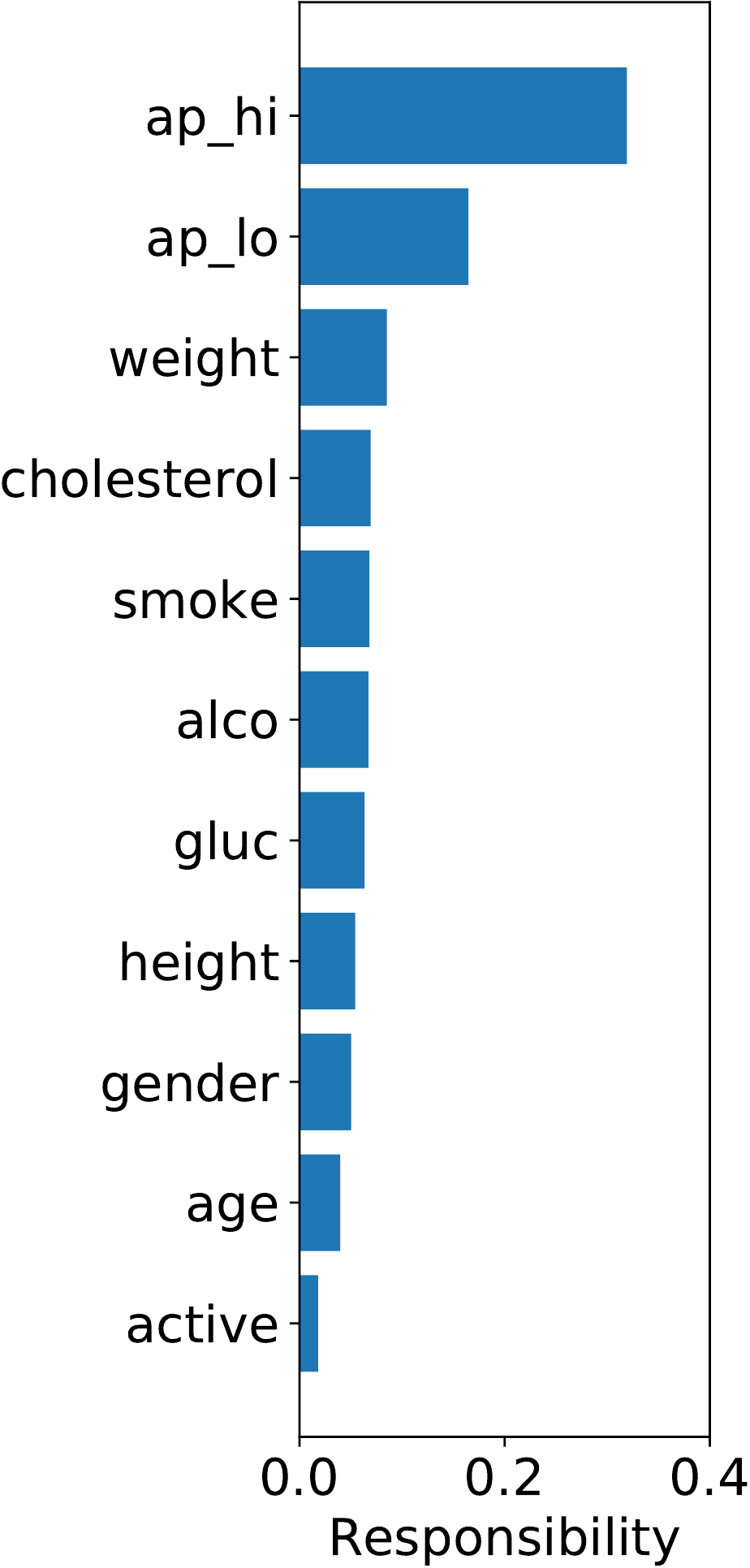}
	  \caption{}
	  \label{cardio}
	\end{subfigure}
	\hspace{2mm}
	\begin{subfigure}[b]{.16\textwidth}
	  \centering
	  \includegraphics[width=\linewidth, height=70mm]{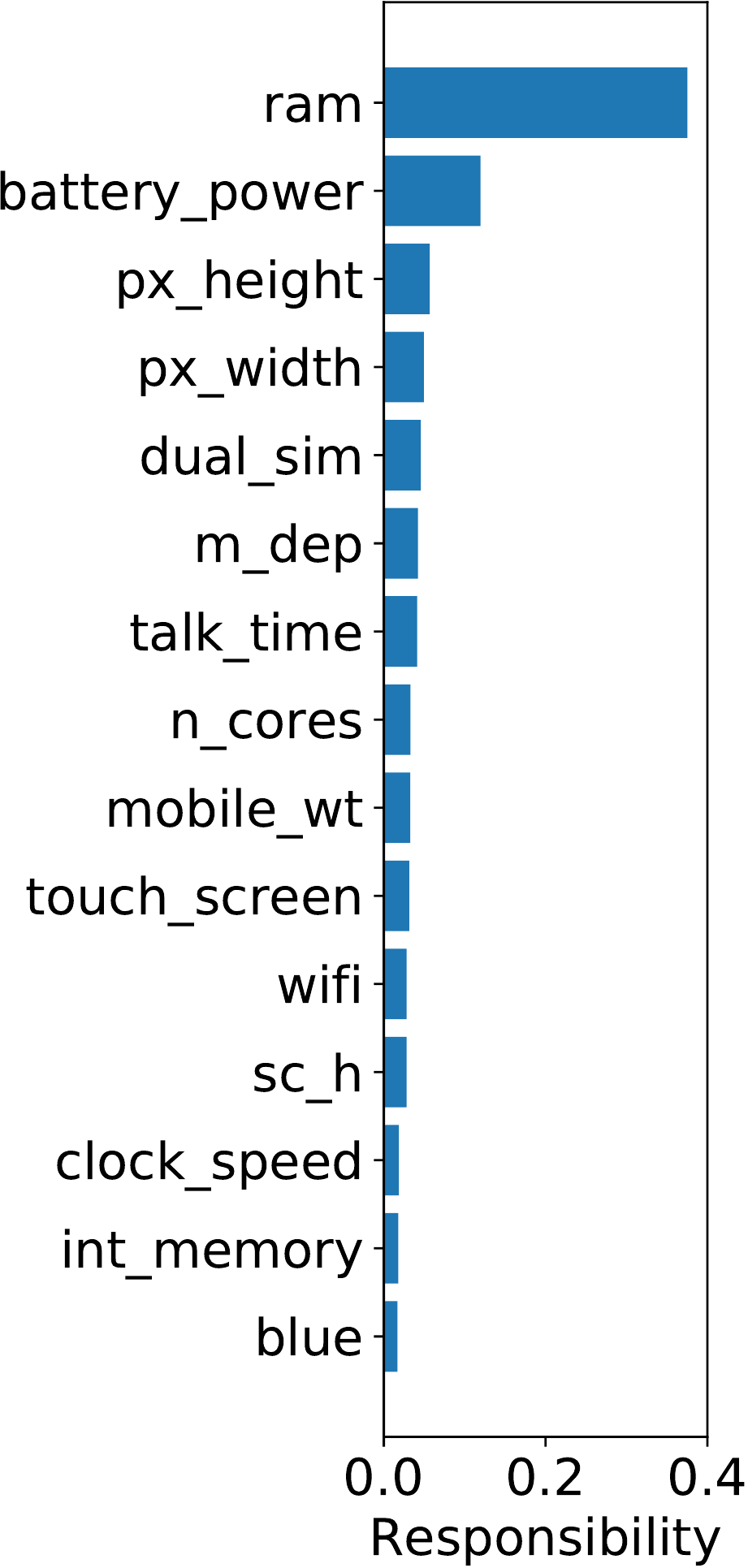}
	  \caption{}
	  \label{mobilePrice}
	\end{subfigure}
	\hspace{2mm}
	\begin{subfigure}[b]{.16\textwidth}
	  \centering
	  \includegraphics[width=\linewidth, height=70mm]{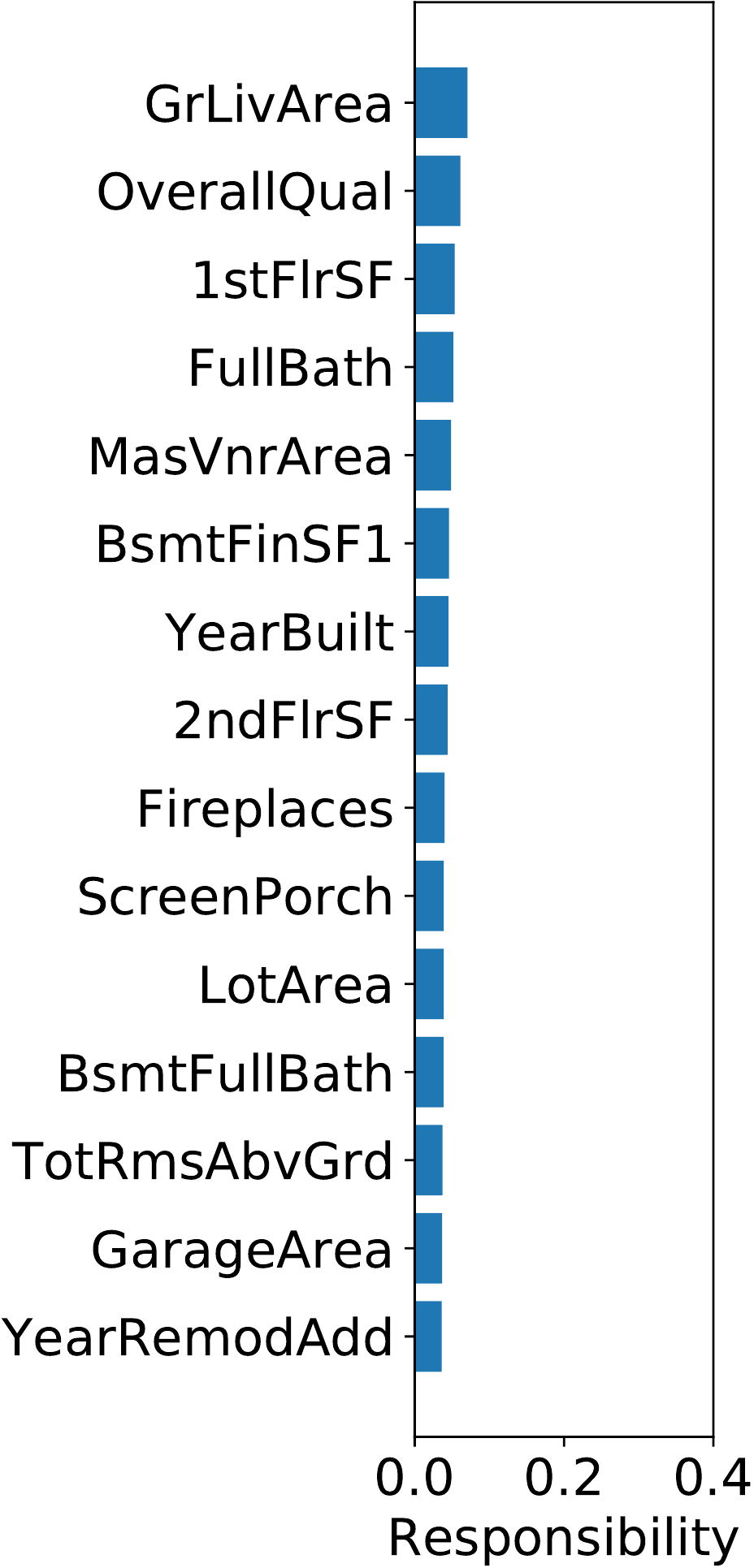}
	  \caption{}
	  \label{housePrice}
	\end{subfigure}
	\hspace{2mm}
	\begin{subfigure}[b]{.36\textwidth}
		\centering
		\includegraphics[width=.8\linewidth, height=70mm]{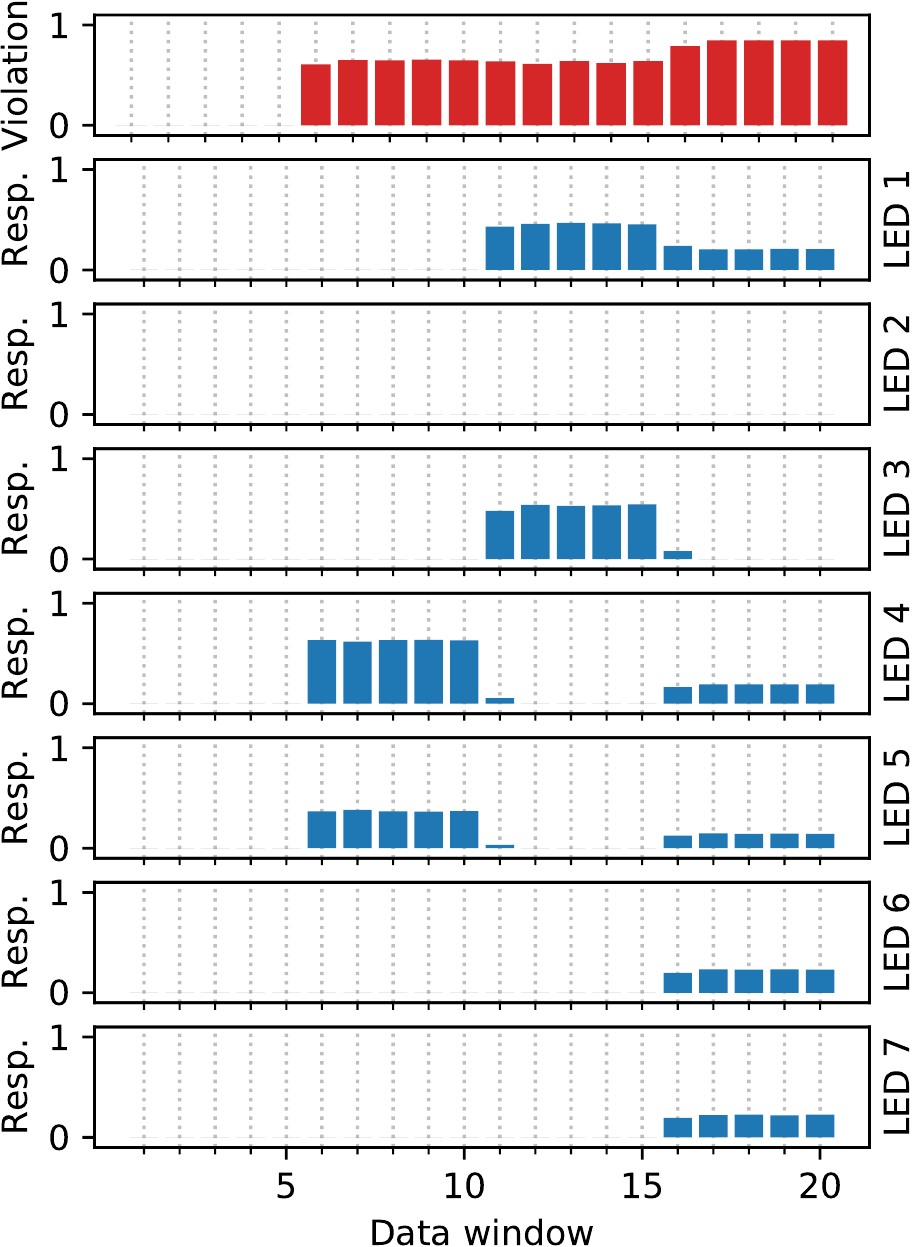}
	  	\caption{}
	    \label{fig:LED-drift}
	\end{subfigure}
	\vspace{0mm}
	}
	\caption{Responsibility assignment on attributes for drift on 
	(a)~Cardiovascular disease: trained on patients with no disease and served on patients with disease,
	(b)~Mobile Prices: trained on cheap mobiles and served on expensive mobiles and 
	(c)~House Prices: trained on house with price $<=$ 100K and served on house with price $>=$ 300K.
	(d)~Detection of drift on LED dataset. The dataset drifts every 5
		 windows (25,000 tuples). At each drift, a certain set of LEDs
		 malfunction and take responsibility of the drift.}
	\vspace{2mm}
	\label{fig:extune}
\end{figure*}

\noindent\textbf{\extune.} \looseness-1 We built a tool
\extune~\cite{DBLP:conf/sigmod/FarihaTRG20}, on top of \system, to compute the
responsibility values as described next. Given a training dataset $D$ and a
non-conforming tuple $t \in {\DDom}^m$, we measure the \emph{responsibility} of
the $i^{th}$ attribute $A_i$ towards the non-conformance as follows: (1)~We
intervene on $t.A_i$ by altering its value to the mean of $A_i$ over $D$ to
obtain the tuple $t^{(i)}$. (2)~In $t^{(i)}$, we compute how many additional
attributes need to be altered to obtain a tuple with no violation. If $K$
additional attributes need to be altered, $A_i$ has responsibility
$\frac{1}{K+1}$. (3)~This responsibility value for each tuple $t$ can be
averaged over the entire serving dataset to obtain an aggregate responsibility
value for $A_i$.
Intuitively, for each tuple, we are ``fixing'' the value of $A_i$ with a ``more
typical'' value, and checking how close (in terms of additional fixes required)
this takes us to a conforming tuple. The larger the number of additional fixes
required, the lower the responsibility of $A_i$.

\smallskip

\noindent\textbf{Datasets.} We use four datasets for this
evaluation: (1)~\emph{Cardiovascular Disease}~\cite{cardioSource} is a
real-world dataset that contains information about cardiovascular patients with
attributes such as height, weight, cholesterol level, glucose level, systolic
and diastolic blood pressures, etc. (2)~\emph{Mobile
Prices}~\cite{mobilePriceSource} is a real-world dataset that contains
information about mobile phones with attributes such as ram, battery power,
talk time, etc. (3)~\emph{House Prices}~\cite{housePriceSource} is a real-world
dataset that contains information about houses for sale with attributes such as
basement area, number of bathrooms, year built, etc. (4)~\emph{LED} (Light
Emitting Diode)~\cite{DBLP:journals/jmlr/BifetHKP10} is a synthetic benchmark.
The dataset has a digit attribute, ranging from 0 to 9, 7 binary
attributes---each representing one of the 7 LEDs relevant to the digit
attribute---and 17 irrelevant binary attributes. This dataset includes gradual
concept drift every 25,000 rows.

\smallskip

\noindent\textbf{Case studies.} \extune produces bar-charts of responsibility
values as depicted in Fig.~\ref{fig:extune}.
Figures~\ref{cardio},~\ref{mobilePrice}, and~\ref{housePrice} show the
explanation results for Cardiovascular Disease, Mobile Price, and House Price
datasets, respectively. For the cardiovascular disease dataset, the training
and serving sets consist of data for patients without and with cardiovascular
disease, respectively. For the House Price and Mobile Price datasets, the
training and serving sets consist of houses and mobiles with prices below and
above a certain threshold, respectively. As one can guess, we get many useful
insights from the non-conformance responsibility bar-charts such as: ``abnormal
(high or low) blood pressure is a key cause for non-conformance of patients
with cardiovascular disease w.r.t. normal people'', ``RAM is a distinguishing
factor between expensive and cheap mobiles'', ``the reason for houses being
expensive depends holistically on several attributes''.

Fig.~\ref{fig:LED-drift} shows a similar result on the LED dataset. Instead
of one serving set, we had 20 serving sets (the first set is also used as a training
set to learn \dis). We call each serving set a window where each window
contains 5,000 tuples. This dataset introduces gradual concept drift every
25,000 rows (5 windows) by making a subset of LEDs malfunctioning. As one can
clearly see, during the initial 5 windows, no drift is observed. In the next 5
windows, LED 4 and LED 5 starts malfunctioning; in the next 5 windows, LED 1
and LED 3 starts malfunctioning, and so on.

\begin{figure*}
	\centering
	\resizebox{0.7\textwidth}{!}
	{
	\begin{subfigure}[b]{.12\textwidth}
	  \centering
	  \includegraphics[width=1\linewidth]{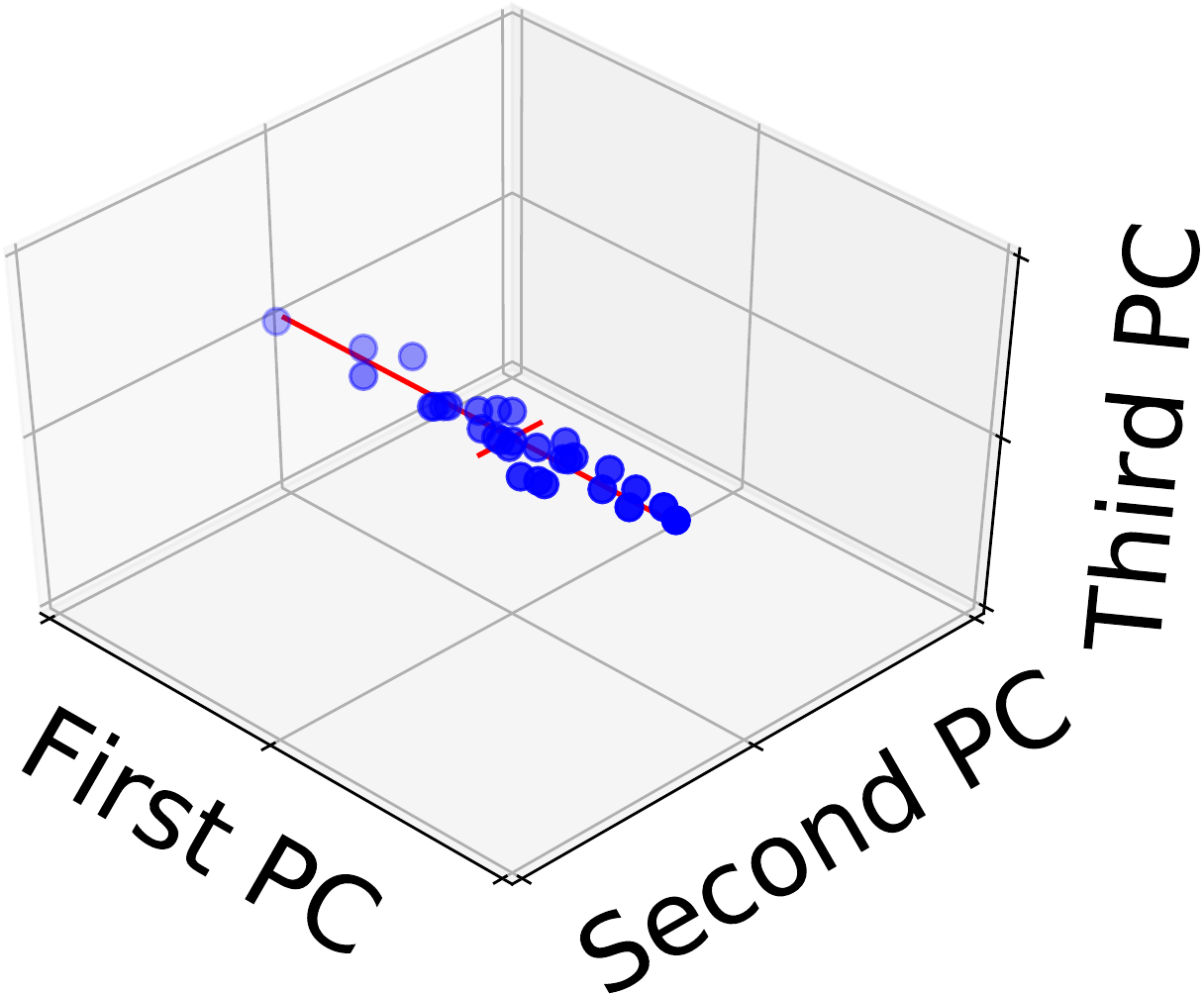}
	  \caption{}
	  \label{subfig:3dPCA}
	\end{subfigure}%
	\hspace{2mm}
	\begin{subfigure}[b]{.12\textwidth}
	  \centering
	  \includegraphics[width=1\linewidth]{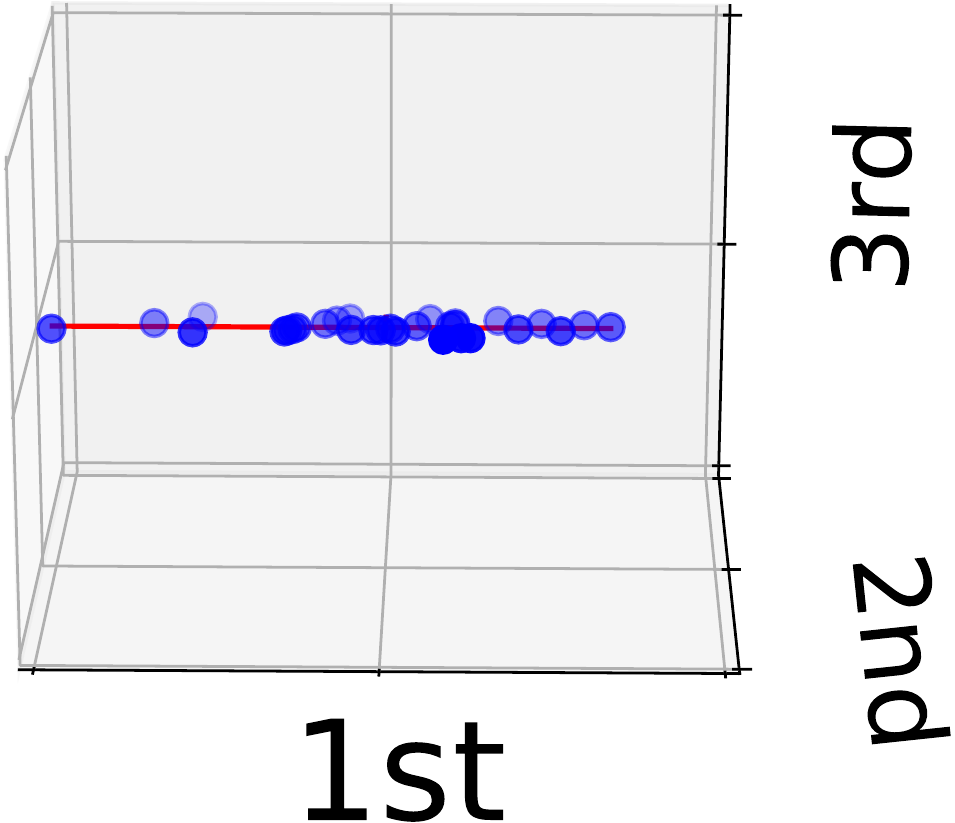}
	  \vspace{-5mm}
	  \caption{}
	  \label{subfig:firstPCA}
	\end{subfigure}
	\hspace{0mm}
	\begin{subfigure}[b]{.12\textwidth}
	  \centering
	  \includegraphics[width=1\linewidth]{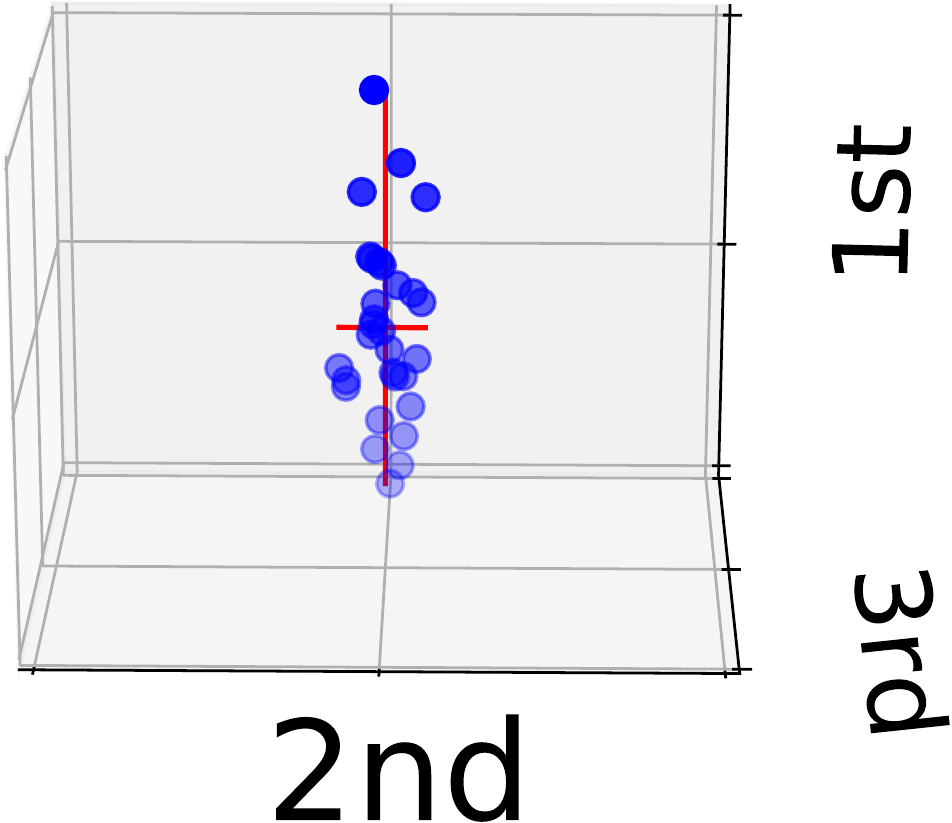}
	  \vspace{-5mm}
	  \caption{}
	  \label{subfig:secondPCA}
	\end{subfigure}
	\hspace{0mm}
	\begin{subfigure}[b]{.12\textwidth}
	  \centering
	  \includegraphics[width=1\linewidth]{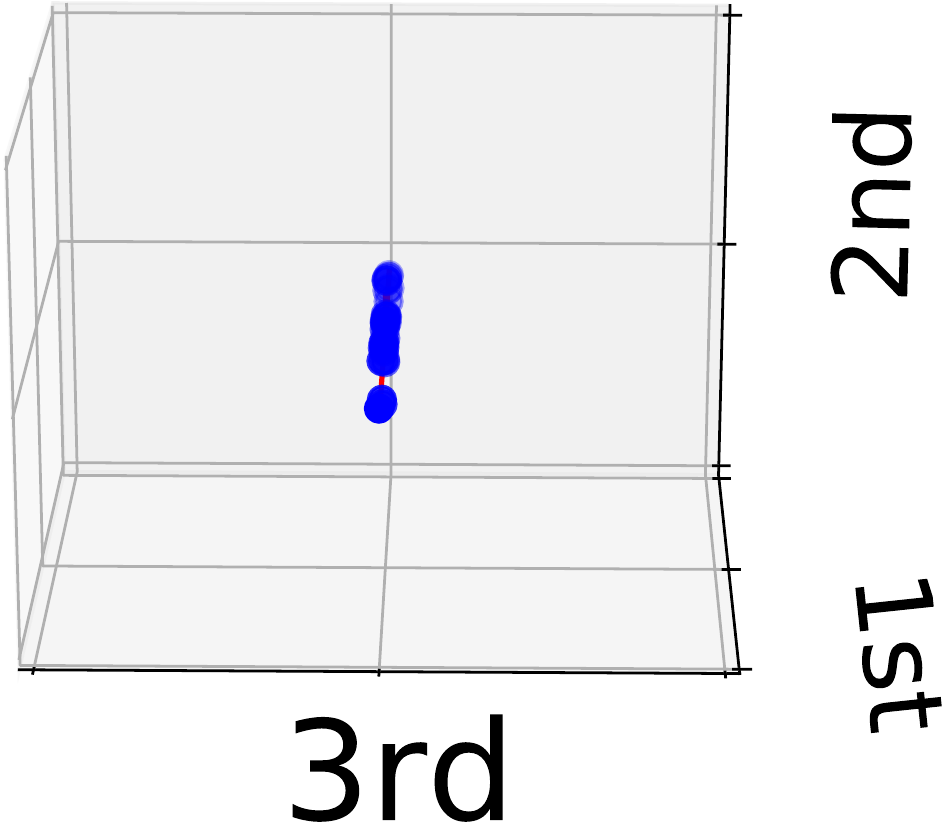}
	  \caption{}
	  \label{subfig:thirdPCA}
	\end{subfigure}
	}
	\caption{\small
	(a)~3D view of a set of tuples projected onto the space of principal components (PC). 
	 (b)~The first PC gives the projection with highest standard deviation and
	 thus constructs the weakest \di with a very broad range for its bounds.
	 (c)~The second PC gives a projection with moderate
	 standard deviation and constructs a relatively stronger \di.
	 (d)~The third PC gives the projection with lowest standard deviation and
	 constructs the strongest \di.
	 }
	 \label{fig:pcadetails} 
	 \vspace{0mm} 
 \end{figure*}

\section{Contrast with Prior Art}

\paragraph{Simple \dis vs. least square techniques.} Note that the lowest
variance principal component of $[\vec{1} ; D_N]$ is related to the ordinary
least square (OLS)---commonly known as linear regression---estimate for
predicting $\vec{1}$ from $D_N$; but OLS minimizes error for the target
attribute only. Our PCA-inspired approach is more similar to total least
squares (TLS)---also known as orthogonal regression---that minimizes
observational errors on all predictor attributes. However, TLS returns only the
lowest-variance \view (Fig.~\ref{subfig:thirdPCA}). In contrast, PCA offers
multiple \views at once (Figs.~\ref{subfig:firstPCA}, \ref{subfig:secondPCA},
and~\ref{subfig:thirdPCA}) for a set of tuples (Fig.~\ref{subfig:3dPCA}), which
range from low to high variance and have low mutual correlation (since they are
orthogonal to each other). Intuitively, \dis constructed from all projections
returned by PCA capture various aspects of the data, as it forms a bounding
hyper-box around the data tuples. However, to capture the relative importance
of \dis, we inversely weigh them according to the variances of their \views in
the quantitative semantics.

\medskip

\paragraph{Compound \invariants vs. denial constraints.} If we try to express
the compound \invariant $\psi_2$ of Example~\ref{ex:constraints} using the
notation from traditional denial constraints~\cite{DBLP:journals/pvldb/ChuIP13}
(under closed-world semantics), where $\month$ always takes values from
$\{$``May'', ``June'', ``July''$\}$, we get the following:
	{
	\begin{align*}
     	\Delta:	   	\neg \; ((\month &= \text{``May''})  \!\!\!\!\!\! & \wedge  \;\; \neg \;(-2 \leq \mathcolorbox{lightgray}{AT - DT - DUR} \leq 0)) \\
     \wedge  \; \neg \; ((\month &= \text{``June''}) \!\!\!\!\!\! & \wedge  \;\; \neg \;(\phantom{-}0  \leq \mathcolorbox{lightgray}{AT - DT - DUR}  \leq 5))\\
     \wedge  \; \neg \; ((\month &= \text{``July''}) \!\!\!\!\!\! & \wedge  \;\; \neg \;(-5 \leq \mathcolorbox{lightgray}{AT - DT - DUR}  \leq 0))
	\end{align*}
	}	
\looseness-1 Note however that arithmetic expressions that specify linear
combination of numerical attributes (highlighted fragment signifying a
projection) are disallowed in denial constraints, which only allow raw
attributes and constants within the constraints. Furthermore, existing
techniques that compute denial constraints offer no mechanism to discover
constraints involving such a composite attribute (projection). Under an
open-world assumption, \dis are more conservative---and
therefore, more suitable for certain tasks such as TML---than denial
constraints. For example, a new tuple with $\month = $ ``August'' will satisfy
the above constraint $\Delta$ but not the compound \di $\psi_2$ of
Example~\ref{ex:constraints}.

\smallskip

\paragraph{Data profiling.} \Dis, just like other constraint models, fall under
the umbrella of data profiling using
metadata~\cite{DBLP:journals/vldb/AbedjanGN15}. There is extensive literature
on data-profiling primitives that model relationships among data attributes,
such as unique column combinations~\cite{DBLP:journals/pvldb/HeiseQAJN13},
functional dependencies (FD)~\cite{papenbrock2015functional,
DBLP:conf/sigmod/ZhangGR20} and their variants (metric~\cite{koudas2009metric},
conditional~\cite{DBLP:conf/icde/FanGLX09},
soft~\cite{DBLP:conf/sigmod/IlyasMHBA04}, approximate~\cite{huhtala1999tane,
kruse2018efficient}, relaxed~\cite{caruccio2016discovery}, etc.), differential
dependencies~\cite{song2011differential}, order
dependencies~\cite{DBLP:journals/vldb/LangerN16,
DBLP:journals/pvldb/SzlichtaGGKS17}, inclusion
dependencies~\cite{papenbrock2015divide, DBLP:journals/jiis/MarchiLP09}, denial
constraints~\cite{DBLP:journals/pvldb/ChuIP13,
DBLP:journals/corr/abs-2005-08540, DBLP:journals/pvldb/BleifussKN17,
pena2019discovery}, and statistical
constraints~\cite{DBLP:conf/sigmod/YanSZWC20}. However, none of them focus on
learning approximate arithmetic relationships that involve multiple numerical
attributes in a noisy setting, which is the focus of our work.

\looseness-1 Soft FDs~\cite{DBLP:conf/sigmod/IlyasMHBA04} model correlation and
generalize traditional FDs by allowing uncertainty, but are limited in modeling
relationships between only a pair of attributes. Metric
FDs~\cite{koudas2009metric} allow small variations in the data, but the
existing work focuses on verification only and not discovery of metric FDs.
Some variants of FDs~\cite{huhtala1999tane, kruse2018efficient,
caruccio2016discovery, koudas2009metric} consider noisy setting, but they
require the allowable noise parameters to be explicitly specified by the user.
However, determining the right settings for these parameters is non-trivial.
Most existing approaches treat constraint violation as Boolean, and do not
measure the degree of violation. In contrast, we do not require any explicit
noise parameter and provide a way to quantify the degree of violation of \dis.

Conditional FDs~\cite{DBLP:conf/icde/FanGLX09} require the FDs to be satisfied
conditionally (e.g., a FD may hold for US residents and a different FD for
Europeans). Denial constraints (DC) are a universally-quantified
first-order-logic formalism~\cite{DBLP:journals/pvldb/ChuIP13} and can adjust
to noisy data, by adding predicates until the constraint becomes exact over the
entire dataset. However, this can make DCs large, complex, and uninterpretable.
While approximate denial constraints~\cite{pena2019discovery} exist, similar to
approximate FD techniques, they also rely on the users to provide the error
threshold.

\smallskip

\paragraph{Input validation.} Our work here contributes to, while also building
upon, work from machine learning, programming languages, and software
engineering. In software engineering, input validation has been used to improve
reliability~\cite{Cohen80}. For example, it is especially used in web
applications where static and dynamic analysis of the code, that processes the
input, is used to detect vulnerabilities~\cite{Su:PLDI2007}. For monitoring
deployed systems, few prior works exploit
\invariants~\cite{DBLP:journals/arobots/JiangED17,
DBLP:conf/hicons/TiwariDJCLRSS14}. To prevent unwanted outcomes, input
validation techniques~\cite{DBLP:conf/issre/HayesO99,
DBLP:conf/edo/BuehrerWS05} are used in software systems. However, such
mechanisms are usually implemented by deterministic rules or \invariants, which
domain experts provide. In contrast, we learn \dis in an unsupervised manner.

\smallskip

\paragraph{Trusted AI.} The issue of trust, resilience, and interpretability of
artificial intelligence (AI) systems has been a theme of increasing interest
recently~\cite{DBLP:conf/cav/Jha19, DBLP:journals/crossroads/Varshney19,
DBLP:journals/corr/abs-1904-07204}, particularly for high-stake and
safety-critical data-driven AI systems~\cite{DBLP:journals/bigdata/VarshneyA17,
DBLP:conf/hicons/TiwariDJCLRSS14}. A standard way to decide whether to trust a
classifier or not, is to use the classifier-produced confidence score. However,
unlike classifiers, regressors lack a natural way to produce such confidence
scores. To evaluate model performance, regression diagnostics check if the
assumptions made by the model during training are still valid for the serving
data. However, they require knowledge of the ground-truths for the serving
data, which is often unavailable.

\smallskip

\paragraph{Data drift.} \looseness-1 Prior work on data drift, change detection,
and covariate shift~\cite{DBLP:conf/sigmod/Aggarwal03,
DBLP:journals/tnn/BuAZ18, dasu2006information, DBLP:journals/eswa/MelloVFB19,
DBLP:conf/kdd/ReisFMB16, DBLP:journals/inffus/FaithfullDK19,
DBLP:conf/icml/Ho05, hooi2019branch, DBLP:conf/sdm/KawaharaS09,
DBLP:conf/vldb/KiferBG04, DBLP:journals/eswa/SethiK17, DBLP:conf/kdd/SongWJR07,
DBLP:conf/sac/IencoBPP14} relies on modeling data distribution, where change is
detected when the data distribution changes. However, data distribution does
not capture constraints, which is the primary focus of our work. Instead of
detecting drift globally, only a handful of works model local
concept-drift~\cite{DBLP:conf/cbms/TsymbalPCP06} or drift for imbalanced
data~\cite{DBLP:journals/corr/WangA15}. Few data-drift detection mechanisms
rely on availability of classification accuracy
~\cite{DBLP:conf/sbia/GamaMCR04, DBLP:conf/sdm/BifetG07,
gaber2006classification, DBLP:journals/tnn/RutkowskiJPD15} or classification
``blindspots''~\cite{DBLP:conf/iri/SethiKA16}. Some of these works focus on
adapting change in data, i.e., learning in an environment where change in data
is expected~\cite{DBLP:conf/sdm/BifetG07, DBLP:journals/csur/GamaZBPB14,
ouyang2011study, DBLP:conf/aistats/SubbaswamySS19,
DBLP:journals/jfi/YuAWSWP19}. Such adaptive techniques are useful to obtain
better performance for specific tasks; however, their goal is orthogonal to
ours.

\smallskip

\balance

\sloppy \paragraph{Representation learning, outlier detection, and one-class
classification.} \looseness-1 Few
works~\cite{DBLP:journals/corr/abs-1812-02765, DBLP:journals/corr/HendrycksG16c}
 , related to our \di-based approach, use
autoencoder's~\cite{hinton2006reducing, rumelhart1985learning} input
reconstruction error to determine if a new data point is out of distribution.
Another mechanism~\cite{DBLP:journals/corr/abs-1909-03835} learns data
\emph{assertions} via autoencoders towards effective detection of invalid
serving inputs. However, such an approach is task-specific and needs a specific
system (e.g., a deep neural network) to begin with. Our approach is similar to
outlier-detection approaches~\cite{kriegel2012outlier} that define outliers as
the ones that deviate from a generating mechanism such as local correlations.
We also share similarity with
one-class-classification~\cite{DBLP:conf/icann/TaxM03}, where the training data
contains tuples from only one class. In general, there is a clear gap between
representation learning approaches (that models data
likelihood)~\cite{hinton2006reducing, rumelhart1985learning,
achlioptas2017learning, karaletsos2015bayesian} and the (constraint-oriented)
data-profiling techniques to address the problem of trusted AI. Our aim is to
bridge this gap by introducing \dis that are more abstract, yet informative,
descriptions of data, tailored towards characterizing trust in ML predictions.